\documentclass[final]{siamltex}
\usepackage{amsfonts}
\usepackage{amssymb}
\usepackage{psfrag}
\usepackage{amsmath}
\usepackage{graphicx}
\usepackage{multirow}
\usepackage{cite}
\usepackage[normalem]{ulem}
\usepackage{color}

\usepackage{booktabs}
\usepackage{siunitx}
\usepackage{subcaption}

\allowdisplaybreaks[4]

\def\({\left(}
\def\[{\left[}
\def\){\right)}
\def\]{\right]}

\newtheorem{ques}{Problem}
\newtheorem{remark}{Remark}

\newcommand{\triplenorm}[1]{\left\vert\kern-0.25ex\left\vert\kern-0.25ex\left\vert #1 \right\vert\kern-0.25ex\right\vert\kern-0.25ex\right\vert}

\numberwithin{equation}{section}
\numberwithin{theorem}{section}

\begin{document}
\title{Thermodynamically consistent modelling and simulation of two-fluid magnetohydrodynamic equations\thanks{This work is  supported by the National key R \& D Program of China (No.2022YFE03040002) and the National Natural Science Foundation of China (No.12371434).}}

\author{Ting Xiao\thanks{School of Mathematics, Sichuan University, Chengdu, P.R. China. Email: {\tt xiaoting0122@163.com}.} \and Qiaolin He\thanks{Corresponding author. School of Mathematics, Sichuan University, Chengdu, P.R. China.  Email: {\tt qlhejenny@scu.edu.cn}.}}

\maketitle

\begin{abstract}
Based on a rigorous thermodynamic framework, this work develops a two-fluid magnetohydrodynamic model grounded in the Helmholtz free energy formalism. The model maintains full thermodynamic consistency by simultaneously satisfying energy conservation and entropy production laws in two-fluid systems. By analyzing the convex-concave structure of the Helmholtz free energy density, we systematically derive key thermodynamic variables—chemical potential, entropy density, and internal energy—in a self-consistent manner. Building on this foundation, we construct a temporally discrete numerical scheme that inherits the thermodynamic consistency of the continuous model. The scheme is proven to adhere rigorously to both the first and second laws of thermodynamics. For the implemented two-dimensional degenerate system, we establish comprehensive a priori error estimates in space and time. Numerical simulations validate the model's effectiveness in capturing essential plasma phenomena, demonstrating its applicability to complex physical scenarios.
\end{abstract}


\section{Introduction}
Two-fluid magnetohydrodynamic (MHD) models serve as fundamental tools in plasma physics research \cite{Chen1984, fitzpatrick2022plasma}, presenting distinct advantages over traditional single-fluid formulations \cite{Burby2017, Park1999}. By resolving ion and electron dynamics separately, these models capture high-frequency and small-scale electromagnetic–fluid interactions beyond the reach of single-fluid approximations, thereby enhancing physical fidelity in high-temperature, high-density environments such as tokamaks \cite{browning2015two, Zakharov1992}. Consequently, two-fluid descriptions have emerged as a pivotal theoretical framework for investigating complex plasma phenomena, including magnetic reconnection and magnetically confined fusion \cite{ZLi2019, Brackbill1985}, despite posing significantly greater challenges in modeling and computation.

Despite extensive progress in two-fluid MHD simulations, most existing studies  have prioritized computational accuracy or energy conservation, often neglecting entropy production—a fundamental requirement of the second law of thermodynamics \cite{sommerfeld2012lectures}.
While several models and numerical implementations achieve partial energy consistency, they generally fail to guarantee entropy increase or satisfy other thermodynamic constraints \cite{Jardin2007, ZLi2019, Keiji2014, Loverich2011}. 
This limitation is partly due to the widespread use of the ideal gas assumption in current two-fluid formulations. 
Although convenient for closure and frequently applicable to high-temperature plasmas in devices such as tokamaks \cite{White2017, Angioni2009}, this assumption inherently limits the representation of non-ideal thermodynamic behavior and obstructs a self-consistent derivation of thermodynamic quantities.
In contrast, phase-field models for non-isothermal fluids demonstrate that embedding both the first and second laws of thermodynamics into the model formulation and discretization is essential for physical reliability \cite{KouJisheng2018, Kou2017, JLiu2015}. Inspired by these developments, we are motivated to establish a thermodynamically consistent two-fluid MHD framework that unifies energy conservation with entropy production.

On the numerical side, mixed finite element methods—such as Raviart–Thomas and Nédélec elements—have proven to be effective in preserving key structural properties of electromagnetic systems, particularly the divergence-free condition of the magnetic field ($\nabla\cdot\mathbf{B}=0$) \cite{Arnold2006,boffi2013mixed}. Such structure-preserving formulations have been extensively studied in the context of single-fluid MHD, resulting in stable and physically consistent numerical schemes along with rigorous error estimates \cite{Hu2017,Monk2003,MaoShipeng2025,Guermond2003}. 
However, analogous developments for thermodynamically consistent two-fluid MHD remain scarce. Furthermore, the high-dimensional coupling and the inf–sup conditions associated with three-dimensional mixed finite element spaces introduce additional computational challenges. 
This motivates our focus on a degenerate two-dimensional formulation for theoretical analysis and numerical validation.

The primary objective of this study is to reconstruct the thermodynamic framework for two-fluid magnetohydrodynamics (MHD) while retaining the key simplifications introduced by Jardin et al. \cite{Jardin2007}, including quasi-neutrality, negligible electron inertia, and infinite speed of light.
Within this framework, the Helmholtz free energy is employed as a unified thermodynamic potential, from which all thermodynamic variables are derived in a self-consistent manner. This approach ensures strict thermodynamic consistency at the continuous level, satisfying both the first and second laws of thermodynamics, while also enabling the incorporation of certain non-ideal effects omitted in the original model. Without sacrificing its capacity to capture essential plasma physics, the proposed formulation maintains computational tractability and broad applicability.

Another objective of this work is to develop a numerical scheme that preserves thermodynamic consistency at the discrete temporal level. 
We introduce a free-energy form that balances simplicity and mathematical rigor: its structure is convex with respect to the plasma number density and concave with respect to temperature, enabling the self-consistent derivation of thermodynamic quantities and ensuring internal coherence of the model. Building on this formulation, we design a semi-implicit time discretization and rigorously prove its properties of discrete energy conservation and entropy non-decrease. Given the computational complexity of the full three-dimensional model, we further construct a fully discrete finite element scheme for a degenerate two-dimensional formulation and establish corresponding spatiotemporal error estimates. The error analysis follows the framework of inf–sup stability, where auxiliary variables are eliminated following the approach of Mao et al. \cite{Mao2025error}.

The remainder of this paper is organized as follows.
Section~\ref{sec:model} introduces the Helmholtz free energy as the thermodynamic potential, derives the corresponding thermodynamic variables, and formulates the two‐fluid MHD model under several key physical assumptions.
Section~\ref{sec:newmodel} reformulates the governing equations using the relations among pressure, temperature, and chemical potential, and rigorously demonstrates that the resulting model satisfies both the first and second laws of thermodynamics.
Section~\ref{sec:numerial_model} presents a semi-implicit temporal discretization scheme based on an auxiliary velocity variable and the convex–concave structure of the Helmholtz free energy density, and establishes its discrete thermodynamic consistency. Section~\ref{sec:2Dmodel} focuses on a two-dimensional degenerate formulation, develops its mixed finite element discretization, and provides a priori error estimates with respect to both space and time.
Section~\ref{sec:experiments} presents numerical experiments to validate the proposed methodology, and Section~\ref{sec:conclusions} concludes the paper with a summary of findings and future research directions.


\section{Mathematical model}
\label{sec:model}
In this section, we formulate the thermodynamic variables and present the governing equations for a two-fluid magnetohydrodynamic (2F-MHD) model under a set of key simplifying assumptions.
\subsection{Formulations of Thermodynamical Variables}

In this model, the Hel-mholtz free energy functionals for ions (i) and electrons (e) are defined as:
\begin{equation*}
	\begin{aligned}
		\mathcal{F}_\alpha= \int_{\Omega} f_{\alpha} d x,\quad \alpha=i,e,
	\end{aligned}
\end{equation*}
where $f_\alpha$ denotes the free energy density per unit volume, defined as
\begin{equation*}
	\begin{aligned}
		f_{\alpha}=k_BT_{\alpha}n\ln n-\frac{k_B}{r-1}T^2_{\alpha}n.
	\end{aligned}
\end{equation*}
Under the quasi-neutrality condition $n = n_i = n_e$, ion and electron densities are identical.
The first term of $f_\alpha$ corresponds to the standard ideal-gas contribution, whereas the second term introduces a nonlinear temperature dependence that offers enhanced flexibility in modeling high-temperature, high-density plasmas. 
This formulation can be regarded as a systematic extension of the ideal-gas free energy, motivated by simplified representations of non-ideal equations of state such as the Peng–Robinson model~\cite{KouJisheng2018}. 
Additional corrections—such as volume-exclusion effects~\cite{Onuki2007}—may be incorporated when required, with thermodynamic consistency verifiable within the same variational framework.

It should be emphasized that this formulation is not intended to replace the standard ideal-gas description widely used in core plasma modeling. Rather, it provides a flexible thermodynamic framework that preserves fundamental thermodynamic relations while allowing for controlled extensions. Notably, retaining the coefficient $\frac{k_B}{r-1}$ maintains a quasi-linear relationship between pressure and internal energy, which facilitates physical interpretation and numerical implementation. Here, $r = \frac{5}{3}$ denotes the specific heat ratio, and $k_B$ is the Boltzmann constant.

From this free-energy density, the standard thermodynamic relations yield the chemical potential $\mu_\alpha$ and entropy density $s_\alpha$ as follows:
\begin{align} \label{eq:thermo_smu}
	& \mu_{\alpha} = \left(\frac{\delta \mathcal{F}_\alpha(n,T_\alpha) }{\delta n}\right)_{T_{\alpha}}=\frac{\partial f_\alpha(n,T_\alpha)}{\partial n}=k_BT_\alpha(\ln n+1)-\frac{k_B}{r-1}T^2_{\alpha}, \quad \\
	& s_{\alpha} = -\left(\frac{\delta \mathcal{F}_\alpha(n,T_{\alpha}) }{\delta T_{\alpha}}\right)_{n}=-\frac{\partial f_\alpha(n,T_\alpha)}{\partial T_\alpha}=-k_Bn\ln n+\frac{2k_B}{r-1}T_\alpha n.
\end{align}

\subsection{The original model}
We now present the governing equations, which are formulated based on the models of \cite{Jardin2007, ZLi2019, Braginskii1965}. Under the assumption of charge neutrality, the plasma number density satisfies $n = n_i = n_e$. Denoting the ion mass-flow velocity by $\mathbf{u}$ and the electron velocity by $\mathbf{u}_e$, the mass conservation equation is given by:
\begin{equation} \label{eq:mass}
	\begin{aligned}
		\frac{\partial n}{\partial t} + \nabla \cdot (n \mathbf{u}) = 0.
	\end{aligned}
\end{equation}

From the classical Braginskii two-fluid plasma model~\cite{ZLi2019}, 
the momentum conservation and internal energy transport equations for ions and electrons are: 
\begin{align}
	& n M_{i}\left(\frac{\partial \mathbf{u}}{\partial t} + \mathbf{u} \cdot \nabla{\mathbf{u}} \right) =-\nabla p_i+ne(\mathbf{E}+\mathbf{u}\times \mathbf{B}) -\nabla \cdot \pmb{\sigma}_i + M_{i} n \mathbf{g}-\mathbf{R} ,\label{eq:initial_mom_ui}\\
	& n M_{e}\left(\frac{\partial \mathbf{u}_e}{\partial t} + \mathbf{u}_e \cdot \nabla{\mathbf{u}_e} \right) =-\nabla p_e-ne(\mathbf{E}+\mathbf{u}_e\times \mathbf{B}) -\nabla \cdot \pmb{\sigma}_e +\mathbf{R} ,\label{eq:initial_mom_ue}\\
	&\frac{\partial \epsilon_i}{\partial t} + \nabla \cdot \left(\epsilon_i\mathbf{u}  \right) =-p_i\nabla\cdot \mathbf{u} - \pmb{\sigma}_i: \nabla \mathbf{u}-\nabla \cdot \mathbf{q}_i,\label{eq:initial_epsiloni}\\
	& \frac{\partial \epsilon_e}{\partial t} + \nabla \cdot \left(\epsilon_e\mathbf{u}_e  \right) =-p_e\nabla\cdot\mathbf{u}_e+ \frac{\mathbf{J}}{n e} \cdot \mathbf{R}+\nabla\left(\frac{\mathbf{J}}{ne}\right):\pmb{\sigma}_e -\nabla \cdot \mathbf{q}_e \label{eq:initial_epsilon_e}.
\end{align}
Here, $M_{i}$ and $M_e$ represent the molar masses of ions and electrons, respectively; $\mathbf{E}$ and $\mathbf{B}$ denote the electric and magnetic fields; 
and $\mathbf{J}$ is the current density. A gravitational force $\mathbf{g}$ is included. The electron–ion momentum transfer term is is modeled in a simplified form, proportional to the plasma current, as $\mathbf{R}=\xi_Rne\mathbf{J}$. Due to $M_e\ll M_i$,  the electron inertia term is neglected, which yields
\begin{equation}\label{eq:Olm_law0}
	\begin{aligned}
		0=-\nabla p_e-ne(\mathbf{E}+\mathbf{u}_e\times \mathbf{B})-\nabla
		\cdot \pmb{\sigma}_e +\mathbf{R}.
	\end{aligned}
\end{equation}
Based on the expression for the current density, $\mathbf{J} = n e (\mathbf{u} - \mathbf{u}_e)$, the electron velocity $\mathbf{u}_e = \mathbf{u} - \frac{\mathbf{J}}{n e}$ is substituted into equations \eqref{eq:Olm_law0} and \eqref{eq:initial_epsilon_e} to obtain
\begin{align}
	&\mathbf{E}+\mathbf{u}\times \mathbf{B}=\frac{1}{ne}(\mathbf{R}+\mathbf{J}\times\mathbf{B}-\nabla p_e-\nabla
	\cdot \pmb{\sigma}_e),\label{eq:initial_Olm_law}\\
	&\frac{\partial \epsilon_e}{\partial t} + \nabla \cdot \left(\epsilon_e\mathbf{u}  \right) =-p_e\nabla\cdot\mathbf{u}+ \nabla \cdot \left(\epsilon_e\frac{\mathbf{J}}{ne}  \right) +p_e\nabla\cdot\left(\frac{\mathbf{J}}{ne}\right)+\frac{\mathbf{J}}{n e} \cdot \mathbf{R} \label{eq:initial_epsilone}\\
	&\quad+\nabla\left(\frac{\mathbf{J}}{ne}\right):\pmb{\sigma}_e -\nabla \cdot \mathbf{q}_e\nonumber,
\end{align}
thereby yielding equation \eqref{eq:initial_Olm_law} as the generalized Ohm's law.  
Combining \eqref{eq:initial_Olm_law} with the ion momentum equation and adopting the approach of [\cite{Jardin2007}], wherein the electron stress tensor is neglected, yields:

\begin{equation}\label{eq:initial_mom_u}
	\begin{aligned}
		n M_{i}\left(\frac{\partial \mathbf{u}}{\partial t} + \mathbf{u} \cdot \nabla{\mathbf{u}} \right) =\mathbf{J}\times \mathbf{B}-\nabla \cdot \pmb{\sigma} + M_{i} n \mathbf{g},
	\end{aligned}
\end{equation}
where $\pmb{\sigma} = p \mathbf{I}+\pmb{\sigma}_{i},$ and $ p = p_i+p_e.$
The ion viscous stress term is taken to have the form $\pmb{\sigma}_{i} = -\eta D(\mathbf{u})-(\lambda \nabla \cdot \mathbf{u}) \mathbf{I}$
with $D(\mathbf{u}) = \nabla \mathbf{u} + \nabla \mathbf{u}^{T}$ and $\lambda = \xi - \frac{2}{3}\eta$.
Here $\eta$ and $\xi$ are the incompressible and compressible viscosity coefficients that satisfy the positivity constraints $\eta>0$ and $\xi>\frac{2}{3}\eta$. The electron stress tensor is modeled as a hyper-resistivity, with a coefficient $\xi_e$, and is given by: $ \pmb{\sigma}_{e} = \xi_e\xi_R\nabla(\frac{\mathbf{J}}{ne}).$
The pressure $p_{\alpha}$ and the internal energy $\epsilon_{\alpha}$ are given by thermodynamic relations
\begin{equation}\label{eq:thermo-rela-pe}
	\begin{aligned}
		p_{\alpha}=n\mu_{\alpha}-f_{\alpha},\quad
		\epsilon_{\alpha}=f_{\alpha}+s_{\alpha}T_{\alpha},\quad \alpha=i,e.
	\end{aligned}
\end{equation}

Combining with Faraday's law and Amp\`ere's law, we have
\begin{align}
	&\frac{\partial \mathbf{B} }{\partial t } = -\nabla \times \mathbf{E}, \label{eq:Maxsys1}\\
	&\mu_{0} \mathbf{J} = \nabla \times \mathbf{B}, \label{eq:Maxsys2}
\end{align}
where, \eqref{eq:Maxsys2} adopts the assumption of an infinite speed of light, i.e., $
\nabla \times \mathbf{B}-\mu_0 \mathbf{J}=\frac{1}{c^2} \frac{\partial \mathbf{E}}{\partial t} \approx 0 $, which also implies $\nabla \cdot \mathbf{J}=0$. 
The complete two-fluid MHD model is thus given by Eqs.  \eqref{eq:mass}, \eqref{eq:initial_mom_u},\eqref{eq:initial_epsiloni},\eqref{eq:initial_epsilone},
\eqref{eq:initial_Olm_law},\eqref{eq:Maxsys1}, and \eqref{eq:Maxsys2}. The solenoidal constraint for the magnetic field, $\nabla \cdot \mathbf{B} = 0$, which follows from Eq. \eqref{eq:Maxsys1}, is ensured by the initial condition $\nabla \cdot \mathbf{B}_0 = 0$.

The total internal energy density is defined as $\epsilon=\epsilon_e+\epsilon_{i}$. The total energy density (including internal, kinetic, and electromagnetic energy) is:
\begin{equation*}
	\begin{aligned}
		E_{T} = \epsilon + \frac{1}{2}\rho |\mathbf{u}|^{2}+\frac{1}{2\mu_{0}}|\mathbf{B}|^2,
	\end{aligned}
\end{equation*}
where $\rho = n M_i$. We now verify that the model satisfies the energy conservation law. From the momentum balance equation \eqref{eq:initial_mom_u}, the transport of kinetic energy density can be written as
\begin{align}\label{eq:energyH}
	\frac{1}{2} \frac{\partial\left(\rho|\mathbf{u}|^2\right)}{\partial t}+\frac{1}{2} \nabla \cdot\left(\mathbf{u}\left(\rho|\mathbf{u}|^2\right)\right)
	& =\rho \mathbf{u} \cdot\left(\frac{\partial \mathbf{u}}{\partial t}+\mathbf{u} \cdot \nabla \mathbf{u}\right)+\frac{1}{2} \mathbf{u} \cdot \mathbf{u}\left(\frac{\partial \rho}{\partial t}+\nabla \cdot(\rho \mathbf{u})\right) \nonumber\\
	& =-\nabla \cdot( \pmb{\sigma} \cdot \mathbf{u})+  \pmb{\sigma}: \nabla \mathbf{u}+\mathbf{u}\cdot(\mathbf{J}\times\mathbf{B}+M_in\mathbf{g}).
\end{align}
Using \eqref{eq:Maxsys1} ,\eqref{eq:Maxsys2}  and \eqref{eq:initial_Olm_law}, it is further obtained
\begin{align}\label{eq:energyB} 
	\frac{1}{\mu_{0}}\mathbf{B}\cdot\frac{\partial\mathbf{B}}{\partial t}
	=&\frac{1}{\mu_{0}}\big((\nabla\times\mathbf{B})\mathbf{E}-\nabla\cdot(\mathbf{E}\times\mathbf{B})\big)\nonumber\\
	=&-\mathbf{u}\cdot(\mathbf{J}\times\mathbf{B})-\frac{\mathbf{J}}{ne}\cdot\mathbf{R}-p_e\nabla\cdot\left(\frac{\mathbf{J}}{ne}\right)-\nabla\left(\frac{\mathbf{J}}{ne}\right):\pmb{\sigma}_e\nonumber\\
	&+\nabla\cdot\left(\mathbf{J}\cdot\frac{p_{e}}{ne}+\pmb{\sigma_{e}}\cdot\frac{\mathbf{J}}{ne}-\frac{1}{\mu_{0}}(\mathbf{E}\times\mathbf{B})\right),
\end{align}
where we used the vector identity $\nabla\cdot(\mathbf{E}\times\mathbf{B})=\mathbf{B}\cdot(\nabla\times\mathbf{E})-\mathbf{E}\cdot(\nabla\times\mathbf{B})$ and $\frac{\mathbf{J}\cdot\nabla\cdot\pmb{\sigma}_e}{ne}+\nabla(\frac{\mathbf{J}}{ne}):\pmb{\sigma}_e=\nabla\cdot(\pmb{\sigma}_e\cdot\frac{\mathbf{J}}{ne}).$ Adding \eqref{eq:initial_epsiloni} and \eqref{eq:initial_epsilone} yields the balance equation for the total internal energy density ($\epsilon=\epsilon_i+\epsilon_{e}$) :
\begin{align}\label{eq:energyU}
	\frac{\partial \epsilon}{\partial t} + \nabla \cdot \left(\epsilon\mathbf{u}  \right) =& \nabla\cdot\left(\epsilon_{e}\frac{\mathbf{J}}{ne}\right)+p_e\nabla\cdot\left(\frac{\mathbf{J}}{ne}\right)+\frac{\mathbf{J}}{n e} \cdot \mathbf{R}\\
	&+ \nabla\left(\frac{\mathbf{J}}{n e}\right) : \pmb{\sigma}_{e} -\nabla \cdot \mathbf{q} - \pmb{\sigma}: \nabla \mathbf{u}\nonumber,
\end{align}
where the heat flux $\mathbf{q}$ is defined by 
$\mathbf{q}=\mathbf{q}_i+\mathbf{q}_e,$ $ \mathbf{q_{\alpha}}=-\Theta_\alpha\nabla T_{\alpha}, \alpha=i,e,$ and $\Theta_\alpha$ denotes the heat diffusion coeﬃcient, generally depending on molar density and temperature. Summing \eqref{eq:energyH}, \eqref{eq:energyU}, and \eqref{eq:energyB}, we obtain the total energy balance equation:
\begin{align}\label{eq:energy}
	\frac{\partial E_{T}}{\partial t} + \nabla \cdot \left(E_{T}\mathbf{u}  \right) =& \nabla\cdot\left(\mathbf{J}\frac{\epsilon_{e}}{ne}+\mathbf{J}\frac{p_{e}}{ne}+\pmb{\sigma}_e\cdot\frac{\mathbf{J}}{ne}-\frac{1}{\mu_{0}}(\mathbf{E}\times\mathbf{B})+\frac{1}{2\mu_0}\mathbf{u}|\mathbf{B}|^2+ \pmb{\sigma}\cdot \mathbf{u}\right)\nonumber \\
	&+\mathbf{u}M_in\mathbf{g}-\nabla \cdot \mathbf{q}.
\end{align}
This implies energy conservation in the absence of gravitational potential, when all flux terms on the right-hand side vanish on the boundary.

\section{New formulations and thermodynamical consistency}
\label{sec:newmodel}

In this section, we propose a new 2F-MHD model. By leveraging a fundamental thermodynamic relation between the gradients of pressure, temperature, and chemical potential, we reformulate the momentum, internal energy, and generalized Ohm's law equations. This reformulation provides a straightforward means of verifying the model's adherence to the first and second laws of thermodynamics.

\subsection{New formulations}
Following the approach in \cite{ShenJ2015,KouJisheng2018} for the phase-field model, we start from the thermodynamic identity $p_{\alpha} = n\mu_{\alpha} - f_{\alpha},\alpha=i,e$ and compute the pressure gradient:
\begin{equation*}
	\begin{aligned}
		\nabla p_{\alpha}=\nabla(n\mu_{\alpha}-f_{\alpha})=n\nabla \mu_{\alpha}+\mu_{\alpha}\nabla n-\mu_{\alpha}\nabla n+s_{\alpha}\nabla T_{\alpha}=n\nabla \mu_{\alpha}+s_{\alpha}\nabla T_{\alpha}.
	\end{aligned}
\end{equation*}
Summing over ions and electrons, the total pressure gradient becomes:
$\nabla p=\nabla p_i+\nabla p_e= n\nabla\mu +s_{i}\nabla T_{i}+s_e\nabla T_e, $ where $\mu=\mu_i+\mu_e.$
Using this reformulation of the pressure gradient, the momentum balance equation can be written as
\begin{align}\label{eq:momen_eq}
	n M_{i}\left(\frac{\partial \mathbf{u}}{\partial t} + \mathbf{u} \cdot \nabla{\mathbf{u}} \right) =&- n \nabla \mu -s_i\nabla T_i-s_e\nabla T_e+ \mathbf{J}\times \mathbf{B} + M_{i} n \mathbf{g}\\
	& + \nabla \cdot (\eta D(\mathbf{u})) + \nabla (\lambda \nabla \cdot \mathbf{u}),\nonumber
\end{align}
The ion internal energy equation is rewritten as:
\begin{equation}\label{eq:epsilon_ion1}
	\begin{aligned}
		\frac{\partial \epsilon_i}{\partial t} + \nabla \cdot \left(\epsilon_i\mathbf{u}  \right) 
		= -\nabla \cdot \mathbf{q}_i -\nabla \cdot \left(\mathbf{u}p_i \right)  +\mathbf{u} \cdot(n\nabla \mu_i+s_i\nabla T_i)- \pmb{\sigma}_{i} :\nabla \mathbf{u}.
	\end{aligned}
\end{equation}
Using $\epsilon_i = f_i + T_is_i$ and $p_i=n\mu_i-f_i$, the balance equation for the ion internal energy density becomes:
\begin{equation}\label{eq:epsilon_i}
	\begin{aligned}
		\frac{\partial \epsilon_i}{\partial t} + \nabla \cdot \big(\mathbf{u} (n\mu_i+s_iT_i) \big) =  -\nabla \cdot \mathbf{q}_i +\mathbf{u}\cdot(n\nabla \mu_i+s_i\nabla T_i)- \pmb{\sigma}_{i} :\nabla \mathbf{u}.
	\end{aligned}
\end{equation}
Similarly, the electron internal energy equation reads
\begin{align}\label{eq:epsilon_e}
	\frac{\partial \epsilon_e}{\partial t} + \nabla \cdot \left(\big(\mathbf{u}-\frac{\mathbf{J}}{ne}\big) (n\mu_e+s_eT_e) \right)  =&  -\nabla \cdot \mathbf{q}_e+\left(\mathbf{u}-\frac{\mathbf{J}}{ne}\right)\cdot(n\nabla \mu_e+s_e\nabla T_e)\nonumber\\
	& + \frac{\mathbf{J}}{n e} \cdot \mathbf{R}
	+ \nabla\left(\frac{\mathbf{J}}{n e}\right) : \pmb{\sigma}_{e}.
\end{align}
Summing the ion and electron internal energy equations yields the total internal energy equation:
\begin{align}\label{eq:epsilon_total}
	\frac{\partial \epsilon}{\partial t} + \nabla \cdot \big(\mathbf{u} (n\mu+s_iT_i+s_eT_e) \big)  =&  -\nabla \cdot \mathbf{q} +\mathbf{u}\cdot(n\nabla \mu+s_i\nabla T_i+s_e\nabla T_e)\\
	& +\nabla\cdot\left(\frac{\mathbf{J}}{ne}(n\mu_e+s_eT_e)\right)-\frac{\mathbf{J}}{ne}(n\nabla \mu_e+s_e\nabla T_e)\nonumber\\
	&+ \frac{\mathbf{J}}{n e} \cdot  \mathbf{R}+ \nabla\left(\frac{\mathbf{J}}{n e}\right) : \pmb{\sigma}_{e}	- \pmb{\sigma}_{i} :\nabla \mathbf{u}\nonumber.
\end{align}
Thus, the generalized Ohm's law becomes:
\begin{eqnarray}\label{eq:Olm_law}
	&&\mathbf{E}+\mathbf{u}\times \mathbf{B}=\frac{1}{ne}\left(\mathbf{R}+\mathbf{J}\times\mathbf{B}- (n\nabla \mu_e+s_e\nabla T_e)-\nabla
	\cdot \pmb{\sigma}_e\right).
\end{eqnarray}
We consider the fluid within a closed, fixed-volume domain $\Omega$, with the boundary conditions on $\partial \Omega$ for $t > 0$ prescribed as:
\begin{align}\label{eq:boundary}
	&\mathbf{u}=0, \quad \nabla n \cdot \boldsymbol{\nu}_{\partial \Omega}=0, \quad \nabla T_{\alpha}\cdot \boldsymbol{\nu}_{\partial \Omega}=0,\quad \alpha=i,e,\\ 
	&\boldsymbol{J}\cdot\boldsymbol{\nu}_{\partial \Omega}=0,\quad\boldsymbol{B}\cdot\boldsymbol{\nu}_{\partial \Omega}=0,\quad \boldsymbol{E}\times\boldsymbol{\nu}_{\partial\Omega}=0,\nonumber
\end{align}	
where $\boldsymbol{\nu}_{\partial \Omega}$ is the outward unit normal to $\partial \Omega$, with initial conditions prescribed for $n$, $T_{\alpha}$, $\mathbf{u}$, and $\mathbf{B}$.

In summary, the new 2F-MHD model comprises: the mass balance \eqref{eq:mass}, momentum balance \eqref{eq:momen_eq}, internal energy balances for ions and electrons \eqref{eq:epsilon_i}, \eqref{eq:epsilon_e}, Maxwell's equations \eqref{eq:Olm_law}, \eqref{eq:Maxsys1}, \eqref{eq:Maxsys2}, and the associated initial and boundary conditions. The following subsection will demonstrate that this formulation inherently satisfies the fundamental laws of thermodynamics.

\subsection{Thermodynamical consistency}
We now verify that the proposed model satisfies the first and second laws of thermodynamics.
\begin{theorem}
	The new 2F-MHD system satisfies the first law of thermodynamics, namely the conservation of total energy:
	\begin{align}\label{eq:epsilonlaw}
		\frac{\partial \mathcal{E}}{\partial t} = \int_{\Omega} \mathbf{u}\cdot (M_{i} n \mathbf{g}) d x, 
	\end{align}
	where the total energy functional is defined by $\mathcal{E} = \mathcal{H} + \mathcal{U} + \mathcal{B} $, $ \mathcal{H} = \frac{1}{2} \int_{\Omega}\rho|\mathbf{u}|^{2} dx$, $\mathcal{U} = \int_{\Omega} \epsilon d x $  and $\mathcal{B} =  \frac{1}{2 \mu_{0}} \int_{\Omega} |\mathbf{B}|^{2} d x $ with $\rho = nM_i$.
	
\end{theorem}
\begin{proof}
	The total energy balance \eqref{eq:epsilonlaw} is derived by summing the time derivatives of the kinetic energy $\frac{\partial \mathcal{H} }{\partial t}$ (from substituting the momentum equation \eqref{eq:momen_eq} into \eqref{eq:energyH}), the electromagnetic energy $\frac{\partial \mathcal{B}}{\partial t}$ (from combining \eqref{eq:energyB} with \eqref{eq:Maxsys1}–\eqref{eq:Maxsys2} and \eqref{eq:Olm_law}), and the integrated total energy $\frac{\partial \mathcal{U}}{\partial t}$ (from integrating \eqref{eq:epsilon_total} over $\Omega$).
\end{proof}

We now prove that the proposed 2F-MHD system complies with the second law of thermodynamics. Let $(\cdot,\cdot)$ and $\|\cdot\|$ denote the inner product and norm in $L^2(\Omega)$, respectively. The total entropy is defined as $\mathcal{S} = \mathcal{S}_i + \mathcal{S}_e$, where $\mathcal{S}_\alpha$ is the entropy of species $\alpha$ ($i$ or $e$). Given the relation $s_{\alpha} = (\epsilon_{\alpha}-f_{\alpha})/T_{\alpha}$, it follows that
\begin{equation}\label{eq:entropy_ie}
	\begin{aligned}
		\frac{\partial \mathcal{S}_{\alpha}}{\partial t}=\left(\frac{\partial s_{\alpha}}{\partial t}, 1\right)=\left(\frac{\partial(\epsilon_{\alpha}-f_{\alpha})}{\partial t}, \frac{1}{T_{\alpha}}\right)-\left(\frac{\partial T_{\alpha}}{\partial t}, \frac{s_{\alpha}}{T_{\alpha}}\right).
	\end{aligned}
\end{equation}
To derive the total entropy, we begin by finding the variation of the Helmholtz free energy density. For each species $\alpha = i, e$, we take $\mu_{\alpha}$ and $s_{\alpha}$ as the variations of $f_\alpha$, leading to:
\begin{equation}\label{eq:Helm_ie}
	\frac{\partial f_{\alpha}}{\partial t} = \mu_{\alpha}\frac{\partial n}{\partial t} - s_{\alpha}\frac{\partial T_{\alpha}}{\partial t} =
	\begin{cases}
		-\mu_i \nabla \cdot(n \mathbf{u}) - s_i \frac{\partial T_i}{\partial t} & (\alpha=i), \\
		-\mu_e \nabla \cdot\left(n\left( \mathbf{u}-\frac{\mathbf{J}}{ne}\right)\right)- s_e\frac{\partial T_e}{\partial t} & (\alpha=e).
	\end{cases}
\end{equation}
The continuity equation for electrons ensures
\begin{equation*}
	\begin{aligned}
		\frac{\partial n}{\partial t}+\nabla\cdot\left(n\big(\mathbf{u}-\frac{\mathbf{J}}{ne}\big)\right)=\frac{\partial n}{\partial t}+\nabla\cdot(n\mathbf{u})+\nabla\cdot\mathbf{J}=0,
	\end{aligned}
\end{equation*}
since $\mathbf{J}=\mu_0\nabla\times \mathbf{B}$ implies $\nabla\cdot\mathbf{J}=0$.

\begin{theorem}
	The new 2F-MHD system satisfies the second law of thermodynamics for ions, electrons, and total entropy as follows:
	
	\textbf{1). Ions and Electrons:}
	\begin{align}
		&\frac{\partial \mathcal{S}_i}{\partial t}=\left\|\frac{\Theta_i^{\frac{1}{2}}}{T_i} \nabla T_i\right\|^2+\frac{1}{2}\left\|\left(\frac{\eta}{T_i}\right)^{\frac{1}{2}} D(\mathbf{u})\right\|^2+\left\|\left(\frac{\lambda}{T_i}\right)^{\frac{1}{2}} \nabla \cdot \mathbf{u}\right\|^2,\label{eq:entropylaw_i}\\
		&\frac{\partial \mathcal{S}_e}{\partial t}=\left\|\frac{\Theta_e^{\frac{1}{2}}}{T_e} \nabla T_e\right\|^2+\left\|\left(\frac{\xi_e\xi_R}{T_e}\right)^{\frac{1}{2}} \nabla\left(\frac{\mathbf{J}}{ne}\right)\right\|^2+\left\|\left(\frac{\xi_R}{T_e}\right)^{\frac{1}{2}} \mathbf{J}\right\|^2.\label{eq:entropylaw_e}
	\end{align}
	
	\textbf{2). Total entropy} ($\mathcal{S}=\mathcal{S}_i+\mathcal{S}_e$):
	\begin{align}\label{eq:entropylaw_all}
		\frac{\partial \mathcal{S}}{\partial t}=
		&\sum_{\alpha=i,e}\left\|\frac{\Theta_\alpha^{1 / 2}}{T_\alpha} \nabla T_\alpha\right\|^2+\frac{1}{2}\left\|\left(\frac{\eta}{T_i}\right)^{1 / 2} D(\mathbf{u})\right\|^2+\left\|\left(\frac{\lambda}{T_i}\right)^{1 / 2} \nabla \cdot \mathbf{u}\right\|^2 \\ 
		&+\left\|\left(\frac{\xi_e\xi_R}{T_e}\right)^{1 / 2} \nabla\left(\frac{\mathbf{J}}{ne}\right)\right\|^2+\left\|\left(\frac{\xi_R}{T_e}\right)^{1 / 2}\mathbf{J}\right\|^2.	\nonumber		
	\end{align}
\end{theorem}

\begin{proof}
	For any $\alpha\in\{i,e\}$, combining the internal energy equation \eqref{eq:epsilon_i} or \eqref{eq:epsilon_e} with the Helmholtz free energy variation \eqref{eq:Helm_ie}, and using the identity $\nabla \cdot(s_\alpha T_\alpha \mathbf{v}_\alpha)=T_\alpha \nabla \cdot(s_\alpha \mathbf{v}_\alpha)+s_\alpha \mathbf{v}_\alpha \cdot \nabla T_\alpha$(where $\mathbf{v}_i=\mathbf{u}$, $\mathbf{v}_e=\mathbf{u}-\frac{\mathbf{J}}{ne}$), we obtain
	\begin{align}
		\frac{\partial(\epsilon_i-f_i)}{\partial t}=&-T_i \nabla \cdot(s_i \mathbf{u})-\nabla \cdot \mathbf{q}_i-\pmb{\sigma}_{i }: \nabla \mathbf{u}+s_i \frac{\partial T_i}{\partial t},\label{eq:proof3b_i}\\
		\frac{\partial(\epsilon_e-f_e)}{\partial t}=&-T_e \nabla \cdot\left(s_e\big(\mathbf{u}-\frac{\mathbf{J}}{ne}\big)\right)-\nabla \cdot \mathbf{q}_e+\frac{\mathbf{J}}{ne}\cdot\mathbf{R}\label{eq:proof3b_e}\\
		&+\boldsymbol{\sigma}_{e }: \nabla\left(\frac{\mathbf{J}}{ne}\right)+s_e\frac{\partial T_e}{\partial t}.\nonumber
	\end{align}
	Substituting \eqref{eq:proof3b_i} or \eqref{eq:proof3b_e} into the entropy variation formula \eqref{eq:entropy_ie} yields:
	\begin{align}
		\frac{\partial \mathcal{S}_i}{\partial t} =&-(\nabla \cdot(s_i \mathbf{u}), 1)+\left\|\frac{\Theta_i^{1 / 2}}{T_i} \nabla T_i\right\|^2-\left(\pmb{\sigma}_{i }: \nabla \mathbf{u}, \frac{1}{T_i}\right) ,\label{eq:proof3c_i}\\
		\frac{\partial \mathcal{S}_e}{\partial t} =&-\left(\nabla \cdot\Big(s_e \big(\mathbf{u}-\frac{\mathbf{J}}{ne}\big)\Big), 1\right)+\left\|\frac{\Theta_e^{1 / 2}}{T_e} \nabla T_e\right\|^2\nonumber\\
		&+\left(\frac{\mathbf{J}}{ne}\cdot\mathbf{R}+\boldsymbol{\sigma}_{e }: \nabla\big(\frac{\mathbf{J}}{ne}\big), \frac{1}{T_e}\right) .\label{eq:proof3c_e}
	\end{align}
	Substitution of the constitutive relations for $\pmb{\sigma}_{i}$ and $\pmb{\sigma}_{e}$ yields the entropy equations \eqref{eq:entropylaw_i} and \eqref{eq:entropylaw_e}. Their sum gives the total entropy evolution \eqref{eq:entropylaw_all}, thereby demonstrating the nonnegativity of the entropy production.

\end{proof}

\section{Thermodynamically consistent numerical method}
\label{sec:numerial_model}

In this section, we develop semi-implicit time-marching schemes for the reformulated equations that strictly preserve thermodynamics. A cornerstone of this approach is the convex or concave structure of the Helmholtz free energy density, which provides a foundation for thermodynamic consistency. To decouple the strong nonlinearities among the number density, temperature, electromagnetic fields, and velocity, an auxiliary velocity variable, defined as a function of density and temperature, is introduced.

\begin{lemma}
	The Helmholtz free energy density $f_{\alpha},\alpha=i,e$ is convex in $n$ $(\frac{\partial^2 f_{\alpha}}{\partial n^2}>0)$ and concave in $T_{\alpha}$ $(\frac{\partial^2 f_{\alpha}}{\partial T_{\alpha}^2}<0)$.
\end{lemma}
\begin{proof}
	For $f_{\alpha}=k_BT_{\alpha}n\ln n-\frac{k_B}{r-1}T_{\alpha}^2n$, we directly compute $\frac{\partial^2 f_{\alpha}}{\partial n^2}=k_BT_{\alpha}\frac{1}{n}>0$ and $\frac{\partial^2 f_{\alpha}}{\partial T_{\alpha}^2}=-\frac{2k_B}{r-1}n<0$.
\end{proof}

We now construct the semi-implicit time marching scheme. Let the time interval $[0, T]$ be partitioned uniformly into $N$ steps, with a time step of $\Delta t = T/N$ and temporal nodes $t_k = k \Delta t$ for $k=0, \dots, N$. We define an auxiliary velocity as
\begin{equation}\label{eq:scheme2}
	\begin{aligned}
		\mathbf{u}^{k}_{\star} = \mathbf{u}^{k} -\frac{\Delta t}{\rho^{k}}\left(n^{k+1}\nabla \mu^{k+1} +s_i^{k} \nabla T_i^{k+1}+s_e^{k} \nabla T_e^{k+1}\right),
	\end{aligned}
\end{equation}
where $\rho^{k} = n^{k}M_{i} $. The quantity $\mathbf{u}_{\star}^k$ can be viewed as an approximation of $\mathbf{u}^{k+1}$ obtained by neglecting the convection, viscosity terms, Lorentz force and gravity terms in the momentum balance equation. Subsequently, a semi-implicit scheme is designed as below:
\begin{subequations}
	\begin{align}
		& \textbf{Mass equation}\nonumber\\[-2pt]
		& \frac{n^{k+1}-n^{k}}{\Delta t} + \nabla \cdot (n^{k+1} \mathbf{u}^{k}_{\star}) = 0, \label{eq:scheme3}\\
		& \textbf{Momentum equation} \nonumber\\[-2pt]
		&\rho^k \frac{\mathbf{u}^{k+1}-\mathbf{u}_{\star}^{k}}{\Delta t} +\rho^{k+1}\mathbf{u}^{k}_{\star} \nabla
		\mathbf{u}^{k+1} \label{eq:scheme4}\\
		&=  \nabla \cdot (\eta D(\mathbf{u}^{k+1})) + \nabla (\lambda \nabla \cdot \mathbf{u}^{k+1})+\mathbf{J}^{k+1}\times \mathbf{B}^{k+1} + M_{i} n^{k+1} \mathbf{g}^{k+1},\nonumber\\
		& \textbf{Ion internal energy equation} \nonumber\\[-2pt]
		& \frac{\epsilon_i^{k+1}-\epsilon_i^{k}}{\Delta t} +\nabla \cdot \big( \mathbf{u}^{k}_{\star}(n^{k+1}\mu^{k+1}_i+s_i^kT_i^{k+1})\big)\label{eq:scheme5}\\
		&= -\nabla \cdot\mathbf{q}_i^{k+1} +\mathbf{u}^{k}_{\star}\cdot(n^{k+1}\nabla \mu_i^{k+1}+s_i^k\nabla T_i^{k+1})+\eta D(\mathbf{u}^{k+1}):\nabla \mathbf{u}^{k+1} \nonumber\\
		&\quad+ \lambda|\nabla \cdot \mathbf{u}^{k+1}|^{2} 
		+\frac{1}{2 \Delta t} n^{k} M_{i}\left( |\mathbf{u}^{k+1}-\mathbf{u}^{k}_{\star}|^{2} + |\mathbf{u}^{k}_{\star}-\mathbf{u}^{k}|^{2}\right), \nonumber\\
		& \textbf{Electron internal energy equation} \nonumber\\[-2pt]
		&  \frac{\epsilon^{k+1}_{e}-\epsilon^{k}_{e}}{\Delta t} +\nabla \cdot \big( \mathbf{u}^{k}_{\star}(n^{k+1}\mu^{k+1}_e+s_e^kT_e^{k+1})\big) = -\nabla \cdot\mathbf{q}^{k+1}_{e}+ \frac{\mathbf{J}^{k+1}}{n^{k+1}e} \cdot\mathbf{R}^{k+1}\label{eq:scheme6}\\ 
		&\quad+ \nabla\left(\frac{\mathbf{J}^{k+1}}{n^{k+1} e}\right) : \pmb{\sigma}^{k+1}_{e}+\mathbf{u}^{k}_{\star}\cdot(n^{k+1}\nabla \mu_e^{k+1}+s_e^k\nabla T_e^{k+1})+\frac{1}{2\mu_0\Delta t}|\mathbf{B}^{k+1}-\mathbf{B}^k|^2\nonumber\\
		&\quad-\frac{\mathbf{J}^{k+1}}{n^{k+1}e}\cdot(n^{k+1}\nabla \mu^{k+1}_e+s^k_e\nabla T^{k+1}_e)+\nabla\cdot\left(\frac{\mathbf{J}^{k+1}}{n^{k+1}e}(n^{k+1}\mu^{k+1}_e+s^k_eT^{k+1}_e)\right),\nonumber\\
		& \textbf{Maxwell equation} \nonumber\\[-2pt]
		& \frac{\mathbf{B}^{k+1}-\mathbf{B}^{k}}{\Delta t} = -\nabla \times \mathbf{E}^{k+1}, \label{eq:scheme7}\\
		& \mathbf{E}^{k+1} + \mathbf{u}^{k+1} \times \mathbf{B}^{k+1} \label{eq:scheme8}\\
		&= \frac{1}{n^{k+1} e}\left(\mathbf{R}^{k+1} + \mathbf{J}^{k+1} \times \mathbf{B}^{k}-(n^{k+1}\nabla\mu^{k+1}_e+s^k_e\nabla T^{k+1}_e) - \nabla \cdot \pmb{\sigma}^{k+1}_{e}\right),\nonumber\\
		& \mu_{0} \mathbf{J}^{k+1} = \nabla \times \mathbf{B}^{k+1},\label{eq:scheme9}\\
		& \nabla \cdot \mathbf{B}^{k+1} = 0, \label{eq:scheme10}
	\end{align}
\end{subequations}
where
\begin{align}
	&\mathbf{R}^{k+1}=\xi_Rn^{k+1}e\mathbf{J}^{k+1},\quad\pmb{\sigma}^{k+1}_e=\xi_e\xi_R\nabla\left(\frac{\mathbf{J}^{k+1}}{n^{k+1}e}\right),\nonumber\\
	&\mathbf{q}_\alpha^{k+1} = - \Theta_\alpha^{k} \nabla T_\alpha^{k+1},\quad \Theta_\alpha^{k} =  \Theta_\alpha(n^{k}),\quad\alpha=i,e.\nonumber
\end{align}
This scheme is designed to ensure discrete energy conservation and entropy non-decrease.

\begin{theorem}\label{semi}
	The semi-implicit system \eqref{eq:scheme3}--\eqref{eq:scheme10} satisfies the first law of  thermodynamics as
	\begin{align}\label{eq:epsilonlaw_semi}
		\frac{\mathcal{E}^{k+1}-\mathcal{E}^k}{\Delta t}=\int_{\Omega}\mathbf{u}^{k+1}\cdot(M_in^{k+1}\mathbf{g}^{k+1})dx, 
	\end{align}
	where $\mathcal{E}^k = \mathcal{H}^k + \mathcal{U}^k + \mathcal{B}^k $, $ \mathcal{H}^k = \frac{1}{2} \int_{\Omega}\rho^k|\mathbf{u}^k|^{2} dx$, $\mathcal{U}^k = \int_{\Omega} \epsilon^k d x $  and $\mathcal{B}^k =  \frac{1}{2 \mu_{0}} \int_{\Omega} |\mathbf{B}^k|^{2} d x $.
	
\end{theorem}

\begin{proof}
	Multiplying both sides of \eqref{eq:scheme4} by $\mathbf{u}^{k+1}$
	and integrating it over $\Omega$, we
	obtain
	\begin{align}\label{eq:proof2a}
		&\left(\rho^k\frac{\mathbf{u}^{k+1}-\mathbf{u}_{\star}^k}{\Delta t}, \mathbf{u}^{k+1}\right)+\left(\rho^{k+1} \mathbf{u}_{\star}^k \cdot \nabla \mathbf{u}^{k+1}, \mathbf{u}^{k+1}\right)=-\left\|\sqrt{\lambda} \nabla \cdot \mathbf{u}^{k+1}\right\|^2\\
		&\quad-\frac{1}{2}\left\|\sqrt{\eta} D\left(\mathbf{u}^{k+1}\right)\right\|^2+\int_{\Omega}\mathbf{u}^{k+1}\cdot(\mathbf{J}^{k+1}\times\mathbf{B}^{k+1}+M_in^{k+1}\mathbf{g}^{k+1})d x.\nonumber
	\end{align}
	Using \eqref{eq:scheme3} and $\mathcal{H}_{\star}^k=\frac{1}{2} \int_{\Omega} \rho^k\left|\mathbf{u}_{\star}^k\right|^2 d \mathbf{x}$ , we estimate
	\begin{align}\label{eq:proof2b}
		&\left(\rho^k \frac{\mathbf{u}^{k+1}-\mathbf{u}_{\star}^k}{\Delta t}, \mathbf{u}^{k+1}\right)-\frac{1}{2\Delta t}\left(\rho^k,\left|\mathbf{u}^{k+1}-\mathbf{u}_{\star}^k\right|^2\right)\\
		&=
		\frac{1}{2\Delta t}\left(\rho^k,\left|\mathbf{u}^{k+1}\right|^2-\left|\mathbf{u}_{\star}^k\right|^2\right)=\frac{\mathcal{H}^{k+1}-\mathcal{H}_{\star}^k}{\Delta t}-\left(\rho^{k+1} \mathbf{u}_{\star}^k \cdot \nabla \mathbf{u}^{k+1}, \mathbf{u}^{k+1}\right).\nonumber
	\end{align}
	Substituting \eqref{eq:proof2b} into \eqref{eq:proof2a} yields
	\begin{align}\label{proof2c}
		\frac{\mathcal{H}^{k+1}-\mathcal{H}_{\star}^k}{\Delta t}=&-\frac{1}{2 \Delta t}\left(\rho^k,\left|\mathbf{u}^{k+1}-\mathbf{u}_{\star}^k\right|^2\right) -\left\|\sqrt{\lambda} \nabla \cdot \mathbf{u}^{k+1}\right\|^2-\frac{1}{2}\left\|\sqrt{\eta} D\left(\mathbf{u}^{k+1}\right)\right\|^2\nonumber\\
		&+ \int_{\Omega}\mathbf{u}^{k+1}\cdot(\mathbf{J}^{k+1}\times\mathbf{B}^{k+1}+M_in^{k+1}\mathbf{g}^{k+1})d x.
	\end{align}
	Multiplying \eqref{eq:scheme2} by $\mathbf{u}^k_{\star}$ and integrating over $\Omega$ gives
	\begin{align}\label{proof2d}
		\frac{\mathcal{H}_{\star}^k-\mathcal{H}^k}{\Delta t} =&\frac{1}{\Delta t}\left(\rho^k\left(\mathbf{u}_{\star}^k-\mathbf{u}^k\right), \mathbf{u}_{\star}^k\right)-\frac{1}{2\Delta t}\left(\rho^k,\left|\mathbf{u}_{\star}^k-\mathbf{u}^k\right|^2\right) \\ 
		=&-\left(n^{k+1} \nabla \mu^{k+1}+s_i^k \nabla T_i^{k+1}+s_e^k \nabla T_e^{k+1}, \mathbf{u}_{\star}^k\right)-\frac{1}{2\Delta t}\left(\rho^k,\left|\mathbf{u}_{\star}^k-\mathbf{u}^k\right|^2\right) .\nonumber
	\end{align}
	From \eqref{eq:scheme7}, \eqref{eq:scheme8} and \eqref{eq:scheme9}, we have
	\begin{align}\label{proof2f}
		&\frac{\mathcal{B}^{k+1}-\mathcal{B}^k}{\Delta t}=-\frac{1}{2\mu_{0}\Delta t}(1,|\mathbf{B}^{k+1}-\mathbf{B}^k|^2)+\left(\mathbf{u}^{k+1}\times\mathbf{B}^{k+1},\mathbf{J}^{k+1}\right)\\
		&\quad-\left(\frac{1}{n^{k+1}e}(\mathbf{R}^{k+1}+\mathbf{J}^{k+1}\times\mathbf{B}^k-(n^{k+1}\nabla\mu^{k+1}_e+s^k_e\nabla T^{k+1}_e)-\nabla\cdot\pmb{\sigma}_e^{k+1}),\mathbf{J}^{k+1}\right).\nonumber
	\end{align}
	Summing \eqref{eq:scheme5} and \eqref{eq:scheme6} gives the discrete evolution equation for the total internal energy $\epsilon^k=\epsilon^k_i+\epsilon^k_e$. Integrating over $\Omega$ yields $\frac{\mathcal{U}^{k+1}-\mathcal{U}^k}{\Delta t}$. Finally, summing \eqref{proof2c}, \eqref{proof2d}, \eqref{proof2f} and $\frac{\mathcal{U}^{k+1}-\mathcal{U}^k}{\Delta t}$ together yields \eqref{eq:epsilonlaw_semi}.
\end{proof}

We now proceed to the discrete entropy production analysis to verify the scheme's adherence to the second law of thermodynamics. To facilitate the proof, we introduce the following auxiliary velocities for ions and electrons:
\begin{equation}\label{eq:v_ie}
	\mathbf{v}^k_\alpha:= 
	\begin{cases}
		\mathbf{u}^k_\star & (\alpha=i), \\
		\mathbf{u}^k_\star-\frac{\mathbf{J}^{k+1}}{n^{k+1}e} & (\alpha=e).
	\end{cases}
\end{equation}
We first derive a discrete analog of the Helmholtz free energy variation equation \eqref{eq:Helm_ie} for the discrete system. 

\begin{lemma}[Discrete Helmholtz free energy inequality]\label{Helm_dis}
	The discrete Helmholtz free energy densities of ions (i) and electronics (e) satisfy
	\begin{align}\label{eq:dis_free2ie}
		\frac{f_\alpha^{k+1}-f_\alpha^k}{\Delta t} \leq -s_\alpha^k \frac{T_\alpha^{k+1}-T_\alpha^k}{\Delta t}-\mu_\alpha^{k+1} \nabla \cdot\left(n^{k+1} \mathbf{v}_{\alpha}^k\right).
	\end{align}
\end{lemma}

\begin{proof}
	By convexity of $ f_\alpha $ in $ n $ and concavity in $ T_\alpha $:  
	\begin{align}
		\frac{f_{\alpha}^{k+1}-f_{\alpha}^k}{\Delta t} & =\frac{f_{\alpha}\left(n^{k+1}, T_{\alpha}^{k+1}\right)-f_\alpha\left(n^k, T_{\alpha}^{k+1}\right)}{\Delta t}+\frac{f_{\alpha}\left(n^k, T_{\alpha}^{k+1}\right)-f_{\alpha}\left(n^k, T_{\alpha}^k\right)}{\Delta t} \nonumber\\
		& \leq \mu_{\alpha}^{k+1} \frac{n^{k+1}-n^k}{\Delta t}-s_{\alpha}^k \frac{T_{\alpha}^{k+1}-T_{\alpha}^k}{\Delta t}.\nonumber
	\end{align}	
	For $\alpha=i,e$, substitute discrete mass conservation $\frac{n^{k+1}-n^k}{\Delta t} +\nabla\cdot(n^{k+1}\mathbf{v}_{\alpha}^k)=0,$ respectively (for $\alpha=e$, using $\nabla\cdot\mathbf{J}^{k+1}=0$).
\end{proof}
Equations \eqref{eq:scheme5}  and  \eqref{eq:scheme6}  can be rewritten as
\begin{align}
	& \frac{\epsilon^{k+1}_{i}-\epsilon^{k}_{i}}{\Delta t} + T^{k+1}_{i}\nabla \cdot (s^{k}_{i}  \mathbf{u}^{k}_{\star}) = -\nabla \cdot\mathbf{q}^{k+1}_{i}+ \eta D(\mathbf{u}^{k+1}):\nabla \mathbf{u}^{k+1} \label{eq:dis_energyii}\\
	&\quad -\mu^{k+1}_{i} \nabla \cdot (\mathbf{u}^{k}_{\star} n^{k+1}) + \lambda|\nabla \cdot \mathbf{u}^{k+1}|^{2} +\frac{1}{2 \Delta t} n^{k} M_{i}\left( |\mathbf{u}^{k+1}-\mathbf{u}^{k}_{\star}|^{2} + |\mathbf{u}^{k}_{\star}-\mathbf{u}^{k}|^{2}\right),\nonumber\\
	& \frac{\epsilon^{k+1}_{e}-\epsilon^{k}_{e}}{\Delta t} + T^{k+1}_{e}\nabla \cdot \left(s^{k}_{e}  \big(\mathbf{u}^{k}_{\star}-\frac{\mathbf{J}^{k+1}}{n^{k+1}e}\big)\right) = -\nabla \cdot\mathbf{q}^{k+1}_{e}+ \nabla\left(\frac{\mathbf{J}^{k+1}}{n^{k+1} e}\right) : \pmb{\sigma}^{k+1}_{e} \nonumber\\
	&\quad-\mu^{k+1}_{e} \nabla \cdot \left(\big(\mathbf{u}^{k}_{\star} -\frac{\mathbf{J}^{k+1}}{n^{k+1}e}\big)n^{k+1}\right)+\frac{\mathbf{J}^{k+1}}{n^{k+1}e} \cdot\mathbf{R}^{k+1} +\frac{1}{2\mu_0\Delta t}|\mathbf{B}^{k+1}-\mathbf{B}^k|^2.\label{eq:dis_energyee}
\end{align}

\begin{theorem}
	The semi-implicit scheme \eqref{eq:scheme2}--\eqref{eq:scheme10} satisfies the second law of thermodynamics for ions ($i$), electrons ($e$), and total entropy :
	
	\textbf{1). Ions and Electrons:}
	\begin{align}
		& \frac{\mathcal{S}_i^{k+1}-\mathcal{S}_i^k}{\Delta t} \geq\left\|\frac{\sqrt{\Theta_i^{k+1}}}{T_i^{k+1}} \nabla T_i^{k+1}\right\|^2+\left(\eta D\left(\mathbf{u}^{k+1}\right): \nabla \mathbf{u}^{k+1}, \frac{1}{T_i^{k+1}}\right) \label{eq:dis_entropyi}\\
		&\quad+\left(\lambda\left|\nabla \cdot \mathbf{u}^{k+1}\right|^2+\frac{1}{2 \Delta t} n^{k} M_{i}\left(\left|\mathbf{u}^{k+1}-\mathbf{u}_{\star}^k\right|^2+\left|\mathbf{u}_{\star}^k-\mathbf{u}^k\right|^2\right), \frac{1}{T_i^{k+1}}\right) \geq 0,\nonumber\\
		& \frac{\mathcal{S}_e^{k+1}-\mathcal{S}_e^k}{\Delta t} \geq\left\|\frac{\sqrt{\Theta_e^{k+1}}}{T_e^{k+1}} \nabla T_e^{k+1}\right\|^2 +\left( \xi_R|\mathbf{J}^{k+1}|^{2}+\xi_e\xi_R \left|\nabla\left(\frac{\mathbf{J}^{k+1}}{n^{k} e}\right)\right| ^2, \frac{1}{T_e^{k+1}}\right)\nonumber\\
		&\quad+\left(\frac{1}{2\mu_0\Delta t}|\mathbf{B}^{k+1}-\mathbf{B}^k|^2, \frac{1}{T_e^{k+1}}\right) \geq 0.\label{eq:dis_entropye}
	\end{align}
	
	\textbf{2). Total Entropy} ($\mathcal{S}=\mathcal{S}_i+\mathcal{S}_e$):
	\begin{align}\label{eq:dis_entropyall}
		&\frac{\mathcal{S}^{k+1}-\mathcal{S}^k}{\Delta t}  \geq \sum_{\alpha=i,e}\left\|\frac{\sqrt{\Theta_\alpha^{k+1}}}{T_\alpha^{k+1}} \nabla T_\alpha^{k+1}\right\|^2 +\frac{1}{2}\left(\eta |D\left(\mathbf{u}^{k+1}\right)|^2, \frac{1}{T_i^{k+1}}\right) \\
		&\quad+\left(\lambda\left|\nabla \cdot \mathbf{u}^{k+1}\right|^2+\frac{1}{2 \Delta t} n^{k} M_{i}\left(\left|\mathbf{u}^{k+1}-\mathbf{u}_{\star}^k\right|^2+\left|\mathbf{u}_{\star}^k-\mathbf{u}^k\right|^2\right), \frac{1}{T_i^{k+1}}\right)\nonumber\\
		&\quad+\left( \xi_R|\mathbf{J}^{k+1}|^{2} +\xi_e\xi_R \left|\nabla\left(\frac{\mathbf{J}^{k+1}}{n^{k+1} e}\right)\right| ^2+\frac{1}{2\mu_0\Delta t}|\mathbf{B}^{k+1}-\mathbf{B}^k|^2, \frac{1}{T_e^{k+1}}\right)\geq 0,\nonumber
	\end{align}
	where $\mathcal{S}^k_\alpha=\int_{\Omega} s_\alpha^k dx, \quad \alpha=i,e.$
\end{theorem}

\begin{proof}
	For any $\alpha\in\{i,e\}$, from the entropy density definition $s_\alpha=\frac{1}{T_\alpha}(\epsilon_{\alpha}-f_\alpha)$, we obtain
	\begin{equation}\label{eq:dis_entropyTot1ie}
		\begin{aligned}
			\frac{\mathcal{S}_\alpha^{k+1}-\mathcal{S}_\alpha^k}{\Delta t} 
			=\left(\frac{\epsilon_{\alpha}^{k+1}-\epsilon_{\alpha}^k-f_\alpha^{k+1}+f_\alpha^k}{\Delta t}, \frac{1}{T_\alpha^{k+1}}\right)-\left(\frac{T_\alpha^{k+1}-T_\alpha^k}{\Delta t}, \frac{s_\alpha^k}{T_\alpha^{k+1}}\right).
		\end{aligned}
	\end{equation}
	Substituting the discrete free energy inequality \eqref{eq:dis_free2ie} into \eqref{eq:dis_entropyTot1ie} yields
	\begin{equation}\label{dis_entropyTot2}
		\begin{aligned}
			\frac{\mathcal{S}_\alpha^{k+1}-\mathcal{S}_\alpha^k}{\Delta t} \geq  \left(\frac{\epsilon_{\alpha}^{k+1}-\epsilon_{\alpha}^k}{\Delta t}, \frac{1}{T_\alpha^{k+1}}\right)+\left(\mu_\alpha^{k+1} \nabla \cdot\left(n^{k+1} \mathbf{v}_{\alpha}^k\right), \frac{1}{T_\alpha^{k+1}}\right).
		\end{aligned}
	\end{equation}
	From \eqref{eq:dis_energyii} and \eqref{eq:dis_energyee}, we can respectively obtain $\left(\frac{\epsilon_\alpha^{k+1}-\epsilon_{\alpha}^k}{\Delta t}, \frac{1}{T_\alpha^{k+1}}\right),\alpha=i,e.$
	Substituting them into \eqref{dis_entropyTot2}, we obtain the inequalities  \eqref{eq:dis_entropyi} and \eqref{eq:dis_entropye}.
	Finally, summing the ion and electron results, we get the total entropy \eqref{eq:dis_entropyall}.
	
\end{proof}

The semi-discrete system retains weak nonlinear coupling, despite the semi-implicit approach having largely decoupled. This is solved efficiently by the following fully decoupled, linearized iterative scheme, which builds upon the auxiliary velocity.
\begin{align}
	& \mathbf{u}^{k,l}_{\star} = \mathbf{u}^{k} -\frac{\Delta t}{n^{k}M_i}(n^{k+1,l}\nabla \mu^{k+1,l+1} +s_i^{k} \nabla T_i^{k+1,l}+s_e^{k} \nabla T_e^{k+1,l}), \label{eq:schemeM1} \\
	& \frac{n^{k+1,l+1}-n^{k}}{\Delta t} + \nabla \cdot (n^{k+1,l}\mathbf{u}_{\star}^{k,l}) = 0, \label{eq:schemeM2} \\
	&  n^{k} M_{i} \frac{\mathbf{u}^{k+1,l+1}-\mathbf{u}_{\star}^{k}}{\Delta t} +n^{k+1,l}M_i\mathbf{u}^{k,l}_{\star} \nabla \mathbf{u}^{k+1,l+1} =\nabla \cdot (\eta D(\mathbf{u}^{k+1,l+1})) \label{eq:schemeM3}\\
	&\quad + \nabla (\lambda \nabla \cdot \mathbf{u}^{k+1,l+1})+\mathbf{J}^{k+1,l}\times \mathbf{B}^{k+1,l} + M_{i} n^{k+1,l+1} \mathbf{g}^{k+1,l+1}.\nonumber
\end{align}
For the electromagnetic fields $\mathbf{E}$, $\mathbf{B}$, and $\mathbf{J}$, we have:
\begin{align}
	& \frac{\mathbf{B}^{k+1,l+1}-\mathbf{B}^{k}}{\Delta t} = -\nabla \times \mathbf{E}^{k+1,l+1}, \label{eq:schemeM7}\\
	& \mathbf{E}^{k+1,l+1} + \mathbf{u}^{k+1,l+1} \times \mathbf{B}^{k+1,l+1}= \frac{1}{n^{k+1,l}e}\left(\mathbf{R}^{k+1,l+1} + \mathbf{J}^{k+1,l+1} \times \mathbf{B}^{k}\right.\nonumber\\
	&\left.\quad-\left(n^{k+1,l}\nabla \mu^{k+1,l+1}_e+s^k_e\nabla T^{k+1,l}_e\right)-\nabla \cdot \pmb{\sigma}^{k+1,l+1}_{e}\right),\label{eq:schemeM8}\\
	& \mu_{0} \mathbf{J}^{k+1,l+1} = \nabla \times \mathbf{B}^{k+1,l+1},\label{eq:schemeM9}\\
	& \nabla \cdot \mathbf{B}^{k+1,l+1} = 0. \label{eq:schemeM10}
\end{align}
For the internal energy equations of ions and electrons, we have:
\begin{align}
	& \frac{\epsilon_i^{k+1,l+1}-\epsilon_i^{k}}{\Delta t} +\nabla \cdot \left( \mathbf{u}^{k}_{\star}(n^{k+1,l}\mu^{k+1,l+1}_i+s_i^kT_i^{k+1,l})\right)= -\nabla \cdot\mathbf{q}_i^{k+1,l+1} \label{eq:schemeM4}\\
	&\quad+\lambda|\nabla \cdot \mathbf{u}^{k+1,l+1}|^{2} + \frac{\eta}{2}\left|D(\mathbf{u}^{k+1,l+1})\right|^2+\mathbf{u}^{k}_{\star}\cdot(n^{k+1,l}\nabla \mu_i^{k+1,l+1}+s_i^k\nabla T_i^{k+1,l})\nonumber\\
	&\quad+\frac{1}{2 \Delta t} n^{k} M_{i}\left( |\mathbf{u}^{k+1,l+1}-\mathbf{u}^{k}_{\star}|^{2} + |\mathbf{u}^{k}_{\star}-\mathbf{u}^{k}|^{2}\right),\nonumber \\
	&  \frac{\epsilon^{k+1,l+1}_{e}-\epsilon^{k}_{e}}{\Delta t} +\nabla \cdot \left( \mathbf{u}^{k}_{\star}(n^{k+1,l}\mu^{k+1,l+1}_e+s_e^kT_e^{k+1,l})\right) = -\nabla \cdot\mathbf{q}^{k+1,l+1}_{e}\label{eq:schemeM5}\\ 
	&\quad +\xi_R|\mathbf{J}^{k+1,l+1}|^2+ \xi_e\xi_R\left|\nabla\left(\frac{\mathbf{J}^{k+1,l+1}}{n^{k+1,l} e}\right) \right|^2 +\mathbf{u}^{k}_{\star}\cdot(n^{k+1,l}\nabla \mu_e^{k+1,l+1}+s_e^k\nabla T_e^{k+1,l})\nonumber\\
	&\quad+\nabla\cdot\left(\frac{\mathbf{J}^{k+1,l+1}}{n^{k+1,l}e}(n^{k+1,l}\mu^{k+1,l+1}_e+s^k_eT^{k+1,l}_e)\right)+\frac{1}{2\mu_0\Delta t}|\mathbf{B}^{k+1,l+1}-\mathbf{B}^k|^2\nonumber\\
	&\quad-\frac{\mathbf{J}^{k+1,l+1}}{n^{k+1,l}e}\cdot(n^{k+1,l}\nabla \mu^{k+1,l+1}_e+s^k_e\nabla T^{k+1,l}_e).\nonumber
\end{align}
Here, the superscripts $l$ and $l+1$ denote the $l$th and $(l+1)$th iterations, respectively.
The variables $\mu_\alpha^{k+1,l+1}$ and $\epsilon_{\alpha}^{k+1,l+1}$ $(\alpha=i,e)$ are defined as:
\begin{align}
	\mu_{\alpha}^{k+1,l+1} =&\mu_{\alpha}(n^{k+1,l},T_{\alpha}^{k+1,l})+\frac{\partial\mu_{\alpha}}{\partial n}(n^{k+1,l},T_{\alpha}^{k+1,l})(n^{k+1,l+1}-n^{k+1,l}),\label{eq:schmemu}\\
	\epsilon^{k+1,l+1}_{\alpha}=& \epsilon_\alpha(n^{k+1,l+1},T_{\alpha}^{k+1,l})+\frac{\partial\epsilon_{\alpha}}{\partial T_{\alpha}}(n^{k+1,l+1},T_{\alpha}^{k+1,l})(T_{\alpha}^{k+1,l+1}-T_{\alpha}^{k+1,l}).\label{eq:schmemuepe}
\end{align} 

The iterative method \eqref{eq:schemeM1}--\eqref{eq:schemeM5} is implemented via the following decoupled procedure:

The implementation follows a sequential, decoupled strategy. First, a unique solution for $n^{k+1,l+1}$ is obtained from the linear elliptic system formed by combining \eqref{eq:schemeM1} and \eqref{eq:schemeM2}, utilizing the linearity of \eqref{eq:schmemu}. This, in turn, makes \eqref{eq:schemeM3} a linear elliptic problem for $\mathbf{u}^{k+1,l+1}$, guaranteeing a unique solution. The electromagnetic fields $\mathbf{E}^{k+1,l+1}$, $\mathbf{B}^{k+1,l+1}$, and $\mathbf{J}^{k+1,l+1}$ are then computed explicitly from \eqref{eq:schemeM7}--\eqref{eq:schemeM10}—though their full 3D coupling is complex and deferred, the 2D well-posedness is treated next. Finally, the linear dependence in \eqref{eq:schmemuepe} ensures that \eqref{eq:schemeM4}--\eqref{eq:schemeM5} become uniquely solvable linear elliptic equations for $T_i^{k+1,l+1}$ and $T_e^{k+1,l+1}$.

\section{Fully Discrete Error Analysis for 2D Degenerate Equations}
\label{sec:2Dmodel}
While the preceding section established a thermodynamically consistent semi-implicit scheme for the full 3D two-fluid MHD system, a complete 3D error analysis remains prohibitively complex. We therefore focus on a 2D degenerate formulation that retains the essential thermodynamic and electromagnetic structure of the full model while facilitating a rigorous theoretical analysis.

In the 2D formulation, the ion velocity $\mathbf{u}=(u_1,u_2)^T$ and magnetic field $\mathbf{B}=(B_1,B_2)^T$ are confined to the plane, while the electric field $E$ and current density $J$ become scalar quantities normal to it. Together with the scalar variables for plasma density $n$, ion temperature $T_i$, and electron temperature $T_e$, these variables collectively define a reduced two-fluid MHD system that is consistent with the previously established thermodynamic framework. The weak formulation of this system is presented below, from which the corresponding strong form can be directly recovered.
\subsection{Finite element space}
We define the functional spaces for the relevant variables as follows: (\romannumeral 1) $Q := H^1(\Omega)$, used for scalar variables $n, T_i, T_e$; (\romannumeral 2) $V := H_0^1 (\Omega)=\left\{v \in H^1(\Omega): v=0 \text { on } \partial \Omega\right\} $, used for scalar variables $E$ and $J$; (\romannumeral 3) $\mathbf{V} := [H_0^1 (\Omega)]^2$, used for the velocity field $\mathbf{u}$; (\romannumeral 4) $\mathbf{W}:= H_0(\text{div}, \Omega)=\{\boldsymbol{w} \in H(\operatorname{div}, \Omega), \boldsymbol{w} \cdot \boldsymbol{n}=0 \text { on } \partial \Omega\}$, used for magnetic field $\mathbf{B}$. 
We adopt the standard Sobolev norms:
\begin{equation*}
	\begin{aligned}
		\|u\|_m=\left[\int_{\Omega} \sum_{\alpha \leq m}\left|D^\alpha u\right|^2 d \mathbf{x}\right]^{\frac{1}{2}},\quad
		\|u\|_{\infty}=\sup _{\mathbf{x} \in \Omega}|u(\mathbf{x})|.
	\end{aligned}
\end{equation*}
For any $\mathbf{C}\in \mathbf{W}$, we use the norm $\|\mathbf{C}\|_\mathbf{W}=(||\mathbf{C}\|^2+\|\nabla\cdot \mathbf{C}\|^2)^{\frac{1}{2}}.$
For sequences of functions $\left\{w^{k+1}(\mathbf{x})\right\}$ associated with time levels $k=0,1,2, \ldots, N-1$, we introduce the following notations:
\begin{equation*}
	\begin{aligned}
		\triplenorm{w}_{\infty, m}:=\max _{0 \leq k \leq N-1}\left\|w^{k+1}\right\|_m ,\quad
		\triplenorm{w}_{0, m}:=\left[\sum_{k=0}^{k=N-1}\left\|w^{k+1}\right\|_m^2 \Delta t\right]^{\frac{1}{2}}.
	\end{aligned}
\end{equation*}
Similarly, for a continuous function $w(\mathbf{x}, t)$ defined on the  time interval $(0, \mathrm{~T})$, we define:
\begin{equation*}
	\begin{aligned}
		\|w\|_{\infty, m}:=\sup _{0<t<T}\|w(\cdot, t)\|_m ,\quad
		\|w\|_{0, m}:=\left[\int_0^T\|w(\cdot, t)\|_m^2 d t\right]^{\frac{1}{2}}.
	\end{aligned}
\end{equation*}

Let $\Omega_h$ be a uniformly shape-regular triangulation of the polygonal domain $\Omega$ with maximal element diameter no greater than $h$.
The associated finite element spaces are defined as follows: (\romannumeral 1) $Q_h\subset Q$: Lagrange element of degree $m_1$; (\romannumeral 2) $\mathbf{V}_h\subset\mathbf{V}$: vector-valued Lagrange elements of degree $m_2$; (\romannumeral 3) $V_h\subset V$: Lagrange element of degree $m_3$; (\romannumeral 4) $\mathbf {W} _h \subset \mathbf {W} $: Raviart-Thomas (RT) elements of degree $m_4 $(denoted $RT_{m_4}$). We assume the following standard approximation properties \cite{boffi2013mixed,Liu2000}:
\begin{align}\label{eq:gamma}
	\inf _{v \in V_h}\|u-v\|_s &\leq C h^{m_1+1-s}\|u\|_{m_1+1}, \forall u \in  H^{m_1+1}(\Omega); \nonumber\\
	\inf _{\mathbf{v} \in \mathbf{V}_h}\|\mathbf{u}-\mathbf{v}\|_s &\leq C h^{m_2+1-s}\|\mathbf{u}\|_{m_2+1}, \forall \mathbf{u} \in \mathbf{V} \cap\left(H^{m_2+1}(\Omega)\right)^2 ;\nonumber\\
	\inf _{v \in V_h}\|u-v\|_s &\leq C h^{m_3+1-s}\|u\|_{m_3+1}, \forall u \in V \cap H^{m_3+1}(\Omega) ;\nonumber\\
	\inf _{\mathbf{C} \in \mathbf{W}_h}\|\mathbf{B}-\mathbf{C}\|_s &\leq C h^{m_4+1-s}\|\mathbf{B}\|_{m_4+1}, \forall \mathbf{B} \in \mathbf{W} \cap\left(H^{m_4+1}(\Omega)\right)^2;\nonumber\\
	\inf _{\mathbf{C} \in \mathbf{W}_h}\|\operatorname{div}(\mathbf{B}-\mathbf{C})\| &\leq C h^{m_4+1}\|\operatorname{div} \mathbf{B}\|_{m_4+1}, \forall \mathbf{B} \in \mathbf{W} \cap\left(H^{m_4+1}(\Omega)\right)^2,
\end{align}
where, $m_1, m_2, m_3 \geq 1, m_4 \geq 0$ and $s\in{0,1}$.

\begin{remark}\label{rem:deRham}
	The finite element spaces $V_h$ and $\mathbf{W}_h$ satisfy the following inclusion property:
	\begin{equation*}
		\begin{aligned}
			\nabla \times V_h \subset \mathbf{W}_h,
		\end{aligned}
	\end{equation*}
	where the two-dimensional curl operator is defined by  $\nabla \times v:=\left(\frac{\partial v}{\partial y},-\frac{\partial v}{\partial x}\right)$. This inclusion forms part of the discrete 2D de Rham complex \cite{boffi2013mixed}: $V_h \xrightarrow{\nabla \times} \mathbf{W}_h\xrightarrow{\nabla \cdot} S_h,$
	where $S_h \subset L_0^2(\Omega)$ denotes a suitable piecewise polynomial space (e.g., piecewise constants). A typical example is the pair $V_h=P_m$ (Lagrange elements of degree $m$) and $\mathbf{W}_h=RT_{m-1}$ (Raviart–Thomas elements of degree $m-1$).
\end{remark}

We observe that it is convenient to group the electromagnetic variables as $\boldsymbol{\xi}=(\mathbf{B},J)$, forming the mixed Sobolev space and its finite element counterpart:
$$\mathbf{X}^{\boldsymbol{\xi}}:=\mathbf{W}\times V,\quad \mathbf{X}^{\boldsymbol{\xi}}_h:=\mathbf{W}_h\times V_h.$$
We also introduce the finite element subspace:
$$\mathbf{X}_h^{\boldsymbol{\xi}, 0}=\{\mathbf{C}\in \mathbf{W}_h,G\in V_h: -\mu_0(G,F)+(\mathbf{C},\nabla\times F)=0,\quad \forall F\in V_h\}.$$
All variables collectively belong to the space $(n,\mathbf{u},T_i,T_e,E,\mathbf{B},J)\in \mathbf{X}$,
where
\begin{equation*}
	\begin{aligned}
		\mathbf{X}:=Q\times \mathbf{V}\times Q\times Q\times V \times \mathbf{X}^{\boldsymbol{\xi}}.
	\end{aligned}
\end{equation*}
For $\boldsymbol{\eta}=(\mathbf{C},G)\in \mathbf{X}^{\boldsymbol{\xi}}$ and $F\in V$, we define the bilinear form $\mathbf{b}(\cdot,\cdot)$ on $\mathbf{X}^{\boldsymbol{\xi}}\times V$ by:
\begin{equation*}
	\begin{aligned}
		\mathbf{b}(\boldsymbol{\eta},F)=-\mu_0(G,F)+(\mathbf{C},\nabla\times F).
	\end{aligned}
\end{equation*}
To simplify the notation in the equations, we also define  $a(\mathbf{u},\mathbf{v})=\eta(\nabla \mathbf{u},\nabla \mathbf{v})+(\eta+\lambda)(\nabla\cdot \mathbf{u},\nabla\cdot \mathbf{v}).$

\subsection{Variational Formulation of 2D Equations}
We now give the variational formulation of the 2D degenerate system.
For any $(w,\mathbf{v},q,g,F,\mathbf{C}, G) \in\mathbf{X}$, find $(n,\mathbf{u},T_i,T_e,E,\mathbf{B},J)\in\mathbf{X}$ such that
\begin{subequations}
	\begin{align}
		&\left(\frac{\partial n}{\partial t},w\right) -(n \mathbf{u},\nabla w) = 0 , \label{eq:2Dcontinue1}\\
		& M_{i}\left(n\frac{\partial \mathbf{u}}{\partial t} ,\mathbf{v}\right)+M_i( n\mathbf{u} \cdot \nabla{\mathbf{u}} ,\mathbf{v})=(- n \nabla \mu - s_i \nabla T_i-s_e\nabla T_e,\mathbf{v})+ (J\times \mathbf{B} ,\mathbf{v})\nonumber\\
		&\quad+ M_{i} (n \mathbf{g} ,\mathbf{v})-a(\mathbf{u},\mathbf{v}),  \label{eq:2Dcontinue2} \\
		&\left(\frac{\partial \epsilon_i}{\partial t},q\right) - ( \mathbf{u} (n\mu_i+s_iT_i),\nabla q)+\kappa_i(n\nabla T_i,\nabla q)-(\mathbf{u}\cdot(n\nabla \mu_i+s_i\nabla T_i),q)\nonumber\\
		&\quad+ (\pmb{\sigma}_{i} :\nabla \mathbf{u},q)=0, \label{eq:2Dcontinue31}\\
		& \left(\frac{\partial \epsilon_{e}}{\partial t} ,g\right)-( \mathbf{u} (n\mu_e+s_eT_e) ,\nabla g)+ \kappa_e (n\nabla T_e,\nabla g)-(\mathbf{u}\cdot(n\nabla \mu_e+s_e\nabla T_e) ,g)\nonumber\\
		&\quad-\left( \nabla\big(\frac{J}{n e}\big) : \pmb{\sigma}_{e},g\right)-\left(\frac{J}{n e} \cdot  R,g\right)=0, \label{eq:2Dcontinue32}
		\\
		& \left(\frac{\partial \mathbf{B}}{\partial t}, \mathbf{C}\right)-\mu_0(\mathbf{u}\times\mathbf{B},  G)+\mu_0\xi_R(J,G)+\mu_0\left(\pmb{\sigma}_{e},\nabla\big(\frac{G}{ne}\big)\right)\nonumber\\
		&\quad+\mathbf{b}((\mathbf{C},G),E)-\mathbf{b}((\mathbf{B},J),F)=0 . \label{eq:2Dcontinue4}
	\end{align}
\end{subequations}

The fully discrete scheme is given by the following variational equation. Define the backward difference operator as $d_t u^{k+1} = (u^{k+1} - u^k) / \Delta t$. For each $k = 0, 1, \ldots, N-1$, the solution $\left(n_h^{k+1},\mathbf{u}_h^{k+1},T_{ih}^{k+1},T_{eh}^{k+1},E_h^{k+1},\mathbf{B}_h^{k+1},J_h^{k+1}\right) \in \mathbf{X}_h$ is required to satisfy the following for all test functions $\left(w_h,\mathbf{v}_h,q_h,g_h,F_h,\mathbf{C}_h,G_h\right) \in \mathbf{X}_h$:
\begin{subequations}
	\begin{align}
		& \textbf{Mass equation}\nonumber\\[-2pt]
		&(d_t n_h^{k+1},w_h) - (n_h^{k+1}\mathbf{u}^k_{\star h},\nabla w_h) = 0, \label{eq:2Ddiscrete1}\\
		& \textbf{Momentum equation} \nonumber\\[-2pt]
		& M_i(n_h^k d_t\mathbf{u}_h^{k+1},\mathbf{v}_h)
		+ M_i(n_h^{k+1}\mathbf{u}^k_{\star h}\cdot\nabla\mathbf{u}_h^{k+1},\mathbf{v}_h)
		= (J_h^{k+1}\times\mathbf{B}_h^{k+1},\mathbf{v}_h) \label{eq:2Ddiscrete22}\\
		& \quad - (n_h^{k+1}\nabla\mu_h^{k+1} + s_{ih}^k\nabla T_{ih}^{k+1} + s_{eh}^k\nabla T_{eh}^{k+1},\mathbf{v}_h)
		+ M_i(n_h^{k+1}\mathbf{g},\mathbf{v}_h)
		- a(\mathbf{u}_h^{k+1},\mathbf{v}_h), \nonumber\\
		& \textbf{Ion internal energy equation} \nonumber\\[-2pt]
		& (d_t\epsilon_{ih}^{k+1},q_h)
		- (\mathbf{u}^k_{\star h}(n_h^{k+1}\mu_{ih}^{k+1}+s_{ih}^kT_{ih}^{k+1}),\nabla q_h)
		+ \kappa_i(n_h^k\nabla T_{ih}^{k+1},\nabla q_h) \label{eq:2Ddiscrete31}\\
		& \quad + (\pmb{\sigma}_{ih}^{k+1}:\nabla\mathbf{u}_h^{k+1},q_h)
		- (\mathbf{u}^k_{\star h}\cdot(n_h^{k+1}\nabla\mu_{ih}^{k+1}+s_{ih}^k\nabla T_{ih}^{k+1}),q_h) \nonumber\\
		& \quad - \tfrac{M_i}{2\Delta t}(n_h^k(|\mathbf{u}_h^{k+1}-\mathbf{u}^k_{\star h}|^2+|\mathbf{u}^k_{\star h}-\mathbf{u}_h^k|^2),q_h)=0, \nonumber\\
		& \textbf{Electron internal energy equation} \nonumber\\[-2pt]
		& (d_t\epsilon_{eh}^{k+1},g_h)
		- (\mathbf{u}^k_{\star h}(n_h^{k+1}\mu_{eh}^{k+1}+s_{eh}^kT_{eh}^{k+1}),\nabla g_h)
		+ \kappa_e(n_h^k\nabla T_{eh}^{k+1},\nabla g_h) \label{eq:2Ddiscrete32}\\
		& \quad - (\mathbf{u}^k_{\star h}\cdot(n_h^{k+1}\nabla\mu_{eh}^{k+1}+s_{eh}^k\nabla T_{eh}^{k+1}),g_h)
		- (\tfrac{1}{2\mu_0\Delta t}|\mathbf{B}_h^{k+1}-\mathbf{B}_h^k|^2,g_h) \nonumber\\
		& \quad - (\tfrac{J_h^{k+1}}{n_h^{k+1}e}\cdot R_h^{k+1},g_h)
		- (\nabla(\tfrac{J_h^{k+1}}{n_h^{k+1}e}):\pmb{\sigma}_{eh}^{k+1},g_h)=0, \nonumber\\
		& \textbf{Maxwell equation} \nonumber\\[-2pt]
		& (d_t\mathbf{B}_h^{k+1},\mathbf{C}_h)
		- \mu_0(\mathbf{u}_h^{k+1}\times\mathbf{B}_h^{k+1},G_h)
		+ \mu_0\xi_R(J_h^{k+1},G_h)
		+ \mu_0(\pmb{\sigma}_{eh}^{k+1},\nabla(\tfrac{G_h}{n_h^{k+1}e})) \nonumber\\
		& \quad + \mathbf{b}((\mathbf{C}_h,G_h),E_h^{k+1})
		- \mathbf{b}((\mathbf{B}_h^{k+1},J_h^{k+1}),F_h) = 0. \label{eq:2Ddiscrete4}
	\end{align}
\end{subequations}
The auxiliary quantities are defined by
\begin{equation*}
	\begin{aligned}
		& \mathbf{u}^k_{\star h} = \mathbf{u}_h^k - \tfrac{\Delta t}{n_h^k M_i}(n_h^{k+1}\nabla\mu_h^{k+1}
		+ s_{ih}^k\nabla T_{ih}^{k+1} + s_{eh}^k\nabla T_{eh}^{k+1}),\\
		& \mu_h^{k+1} = \mu_{ih}^{k+1} + \mu_{eh}^{k+1},\quad
		\mu_{\alpha h}^{k+1} = k_B T_{\alpha h}^{k+1}(\ln n_h^{k+1}+1)
		- \tfrac{k_B}{r-1}(T_{\alpha h}^{k+1})^2,\\
		& s_{\alpha h}^k = -k_B n_h^k \ln n_h^k + \tfrac{2k_B}{r-1}T_{\alpha h}^k n_h^k,\quad
		\epsilon_{\alpha h}^{k+1} = \tfrac{k_B}{r-1}(T_{\alpha h}^{k+1})^2 n_h^{k+1},\\
		& \pmb{\sigma}_{eh}^{k+1} = \xi_e\xi_R\nabla(\tfrac{J_h^{k+1}}{n_h^{k+1}e}),\quad
		R_h^{k+1} = \xi_R n_h^{k+1} e J_h^{k+1}.
	\end{aligned}
\end{equation*}

\subsection{Well-posedness of the EBJ System}
\label{sec:inf-sup} 
The well-posedness analysis for the overall decoupled system reduces to that of the Maxwell subsystem \eqref{eq:2Ddiscrete4}, since all other subproblems are standard linear elliptic equations whose existence and uniqueness are assured.

Let $\boldsymbol{\xi}=(\mathbf{B},J)$ and $\boldsymbol{\eta}=(\mathbf{C},G)\in \mathbf{X}^{\boldsymbol{\xi}}_h$.
We define the bilinear functional $\mathbf{a}(\cdot,\cdot)$ on $\mathbf{X}^{\boldsymbol{\xi}}_h\times \mathbf{X}^{\boldsymbol{\xi}}_h$  by:
$$\mathbf{a}(\boldsymbol{\xi},\boldsymbol{\eta})=\frac{1}{\Delta t}(\mathbf{B},\mathbf{C})-\mu_0(\mathbf{u}^{-}\times \mathbf{B},G)+\xi_R\mu_0(J,G)+\mu_0\Big(\pmb{\sigma}_e,\nabla\big(\frac{G}{n^-e}\big)\Big),$$
where $n^-$ and $\mathbf{u}^-$ are known from the previous iteration step.
Then, equation \eqref{eq:2Ddiscrete4} can be reformulated into the standard mixed (saddle-point) form:
\begin{ques}\label{andian_ques}
	For any $\boldsymbol{h}\in \mathbf{X}^{\boldsymbol{\xi},\ast}_h$ and $g\in V^\ast_h$, find $(\boldsymbol{\xi},E)\in \mathbf{X}^{\boldsymbol{\xi}}_h\times V_h$, such that for any $(\boldsymbol{\eta},F)\in  \mathbf{X}^{\boldsymbol{\xi}}_h\times V_h$
	$$
	\left\{\begin{array}{l}
		\boldsymbol{a}(\boldsymbol{\xi}, \boldsymbol{\eta})+\boldsymbol{b}(\boldsymbol{\eta}, E)=\langle\boldsymbol{h}, \boldsymbol{\eta}\rangle, \\
		\boldsymbol{b}(\boldsymbol{\xi}, F)=\langle g, F\rangle.
	\end{array}\right.
	$$
\end{ques}

According to Brezzi's theorem\cite{Brezzi1974,boffi2013mixed}, Problem \ref{andian_ques} admits a unique solution if the following conditions hold:
(\romannumeral 1) $\boldsymbol{a}(\cdot, \cdot)$ and $\boldsymbol{b}(\cdot, \cdot)$ are bounded;
(\romannumeral 2) the inf-sup condition holds for $\boldsymbol{a}(\cdot, \cdot)$ in the kernel of $\boldsymbol{b}(\cdot, \cdot)$;
(\romannumeral 3) the inf-sup condition holds for $\boldsymbol{b}(\cdot, \cdot)$.
We define the norms:
\begin{equation*}
	\begin{aligned}
		||\boldsymbol{\xi}||^2_\mathbf{X}=||\mathbf{B}||^2_\mathbf{W}+||J||^2_{1,n},\quad ||J||_{1,n}^2=||J||^2+||\nabla (\tfrac{J}{n^-})||^2.
	\end{aligned}
\end{equation*}
Since the discrete field $\mathbf{B}_h$ satisfies $\nabla\cdot \mathbf{B}_h=0$, the norms $\|\cdot\|$ and $\|\cdot\|_\mathbf{W}$ are equivalent in our case \cite{Hu2017}.
\begin{lemma}[Equivalence of the weighted norm and the $H^1$ norm]
	Let $J \in V_h$ , and let the given function $n^- \in Q_h$ satisfy: (\romannumeral 1) $n^-\geq n_0>0;$
	(\romannumeral 2)$\|n^-\|_{L^{\infty}(\Omega)}\leq K$ and $\|\nabla n^-\|_{L^{\infty}(\Omega)} \leq K.$
	Then there exist constants $C_1, C_2>0$, depending only on $n_0$ and $K$, such that:
	$$
	C_1\|J\|_1 \leq\|J\|_{1, n} \leq C_2\|J\|_1,
	$$
	where $\|J\|_1:=\left(\|J\|^2+\|\nabla J\|^2\right)^{\frac{1}{2}}$ is the standard $H^1$ norm.
\end{lemma}

\begin{proof}
	Applying the chain rule yields: $
	\nabla\left(\frac{J}{n^-}\right)=\frac{\nabla J}{n^-}-\frac{J}{(n^-)^2}  \nabla n^-,$ we obtain $	\left\|\nabla\left(\frac{J}{n^-}\right)\right\| \leq \frac{1}{n_0}\|\nabla J\|+\frac{K}{n_0^2}\|J\| $, that is $\|J\|_{1, n} \leq C_2\|J\|_1.$
	Conversely, rearranging gives: $\nabla J=n^- \nabla\left(\frac{J}{n^-}\right)+\frac{J }{n^-}\nabla n^-  $, which implies $\|J\|_1 \leq C_1\|J\|_{1, n}.$
	This completes the proof of the lemma.
\end{proof}

\begin{lemma}[Boundedness of bilinear forms]
	There exists a constant C such that
	$$
	\boldsymbol{a}(\boldsymbol{\xi}, \boldsymbol{\eta}) \leq C\|\boldsymbol{\xi}\|_\mathbf{X}\|\boldsymbol{\eta}\|_\mathbf{X}, \quad \boldsymbol{b}(\boldsymbol{\eta}, F) \leq C\|\boldsymbol{\eta}\|_\mathbf{X}\|F\|_1 ,
	$$
	where C depends on $\Omega,\mu_0,\xi_R,\|\mathbf{u}^-\|_{\infty},\Delta t,$ but is independent of the mesh size $h$.
	
\end{lemma}

\begin{proof}
	By Cauchy-Schwarz inequality, we have
	\begin{align*}
		& (\mathbf{B},\mathbf{C})\leq \|\mathbf{B}\|_{\mathbf{W}}\|\mathbf{C}\|_{\mathbf{W}}, \quad (\mathbf{u}^-\times \mathbf{B},G)\leq \|\mathbf{u}^-\|_{\infty}\|\mathbf{B}\|_\mathbf{W}\|G\|_{1,n}, \\
		&(J,G)\leq \|J\|_{1,n} \|G\|_{1,n},\quad(\pmb{\sigma}_e,\nabla(\frac{G}{n^-e}))\leq\|J\|_{1,n}\|\mathbf{G}\|_{1,n}, \\
		& (G,F)\leq \|G\|_{1,n}\|F\|_1,\quad (\mathbf{C},\nabla \times F)\leq \|\mathbf{C}\|_\mathbf{W}\|F\|_1.
	\end{align*}
	Thus, each term is uniformly bounded, completing the proof.
	
\end{proof}

\begin{lemma}[Inf-sup condition for $\boldsymbol{a}(\cdot,\cdot)$]
	If the time step $\Delta t$ satisfies $\Delta t\leq \frac{min\{\xi_R\mu_0,\xi_e\xi_R\mu_0\}}{\|\mathbf{u}^-\|^2_{\infty}}$, then $\boldsymbol{a}(\cdot, \cdot)$ satisfies the inf-sup condition. Specifically, there exists $\gamma>0$ such that
	$$
	\begin{gathered}
		\inf _{\mathbf{0} \neq \boldsymbol{\xi} \in \mathbf{X}_h^{\xi, 0}} \sup _{0 \neq \boldsymbol{\eta} \in \mathbf{X}_h^{\xi, 0}} \frac{\boldsymbol{a}(\boldsymbol{\xi}, \boldsymbol{\eta})}{\|\boldsymbol{\xi}\|_ {\mathbf{X}}\|\boldsymbol{\eta}\|_{\mathbf{X}}} \geq \gamma, \quad
		\inf _{\mathbf{0} \neq \boldsymbol{\eta} \in \mathbf{X}_h^{\xi, 0}} \sup _{0 \neq \boldsymbol{\xi} \in \mathbf{X}_h^{\xi, 0}} \frac{\boldsymbol{a}(\boldsymbol{\xi}, \boldsymbol{\eta})}{\|\boldsymbol{\xi}\|_ {\mathbf{X}}\|\boldsymbol{\eta}\|_{\mathbf{X}} }\geq \gamma,
	\end{gathered}
	$$
	where $\gamma$  is independent of the mesh size $h$.
	
\end{lemma}

\begin{proof}
	We only prove the ﬁrst inf-sup condition. The second one follows by the same argument.
	Taking $\mathbf{C}=\mathbf{B},G=J$, we get
	\begin{align*}
		\boldsymbol{a}(\boldsymbol{\xi}, \boldsymbol{\eta})= & \Delta t^{-1}\|\boldsymbol{B}\|^2-\mu_0(\mathbf{u}^-\times \mathbf{B},J)+\xi_R\mu_0\|J\|^2+\xi_e\xi_R\mu_0\|\nabla (\tfrac{J}{n^-})\|^2.
	\end{align*}
	By the Cauchy-Schwarz inequality and Young’s inequality $ab\leq \epsilon a^2+\frac{1}{4\epsilon} b^2$(where $\epsilon>0$), we have:
	\begin{align*}
		|(\mathbf{u}^-\times \mathbf{B},J)|
		\leq  \|\mathbf{u}^-\|_{\infty}\|\mathbf{B}\|_\mathbf{W}\|J\|_{1,n}\leq \epsilon_1\|J\|^2_{1,n}+\frac{\|\mathbf{u}^-\|^2_{\infty}}{4\epsilon_1}\|\mathbf{B}\|^2_\mathbf{W}.
	\end{align*}
	Choosing $\epsilon_1=\frac{1}{2}\min\{\xi_R\mu_0,\xi_e\xi_R\mu_0\}$ and
	$\Delta t\leq \frac{2\epsilon_1}{\|\mathbf{u}^-\|^2_{\infty}},$
	we derive
	$$
	\boldsymbol{a}(\boldsymbol{\xi}, \boldsymbol{\eta}) \geq \frac{1}{2\Delta t}\|\mathbf{B}\|^2_\mathbf{W}+\frac{1}{2}\min\{\xi_R\mu_0,\xi_e\xi_R\mu_0\}\|J\|^2_{1,n}.
	$$
	Thus, there exists $\gamma_1$ such that for any $\boldsymbol{\xi} \in \mathbf{X}_h^{\boldsymbol{\xi}, 0}$, there exists $\boldsymbol{\eta} \in \mathbf{X}_h^{\xi, 0}$ satisfying 
	$$
	\boldsymbol{a}(\boldsymbol{\xi}, \boldsymbol{\eta}) \geq \gamma_1\|\boldsymbol{\xi}\|_\mathbf{X}^2.
	$$
	Since $\|\boldsymbol{\eta}\|_{\mathbf{X}} \leq C\|\boldsymbol{\xi}\|_{\mathbf{X}}$ , the conclusion follows.
	
\end{proof}

\begin{lemma}[Inf-sup condition for $\boldsymbol{b}(\cdot,\cdot)$]
	The bilinear form $\boldsymbol{b}(\cdot, \cdot)$ satisfies the inf-sup condition, that is, there exists a constant $\beta>0$, such that
	\begin{align}\label{eq:infsup}
		\inf _{F \in V_h} \sup _{\boldsymbol{\eta} \in \mathbf{X}^{\boldsymbol{\xi}}_h} \frac{\boldsymbol{b}(\boldsymbol{\eta}, F)}{\|\boldsymbol{\eta}\|_{\mathbf{X}}\left\|F\right\|_1} \geq \beta,
	\end{align}
	where $\beta$ is independent of the mesh size $h$.
\end{lemma}

\begin{proof}
	For any $F\in V_h$, since $\nabla \times V_h\subset \mathbf{W}_h$, we  take $\mathbf{C}=\nabla \times F\in \mathbf{W}_h$, and  $G=-\frac{1}{\mu_0}F\in V_h.$ Then 
	$$
	\boldsymbol{b}(\boldsymbol{\eta},F) =\|F\|^2+\|\nabla \times F\|^2\geq C\|F\|_1^2,
	$$
	where we use the fact that $\|\nabla \times F\|^2=\|\partial_y F\|^2+\|\partial_x F\|^2=\|\nabla F\|^2.$ Moreover, $\|\boldsymbol{\eta}\|^2_{\mathbf{X}}=\|\nabla \times F\|_{\mathbf{W}}^2+\|F\|^2_{1,n}\leq C\|F\|^2_1$. Thus, the inf-sup condition holds with $\beta$ independent of $h$.
\end{proof}

\subsection{Error estimation}
To derive error estimates, we begin by positing several regularity and boundedness assumptions. We assume that the continuous problem \eqref{eq:2Dcontinue1}--\eqref{eq:2Dcontinue4} admits a unique solution obeying \eqref{eq:assume_conti}, and that the discrete problem \eqref{eq:2Ddiscrete1}--\eqref{eq:2Ddiscrete4} similarly possesses a unique solution, with analogous conditions given by \eqref{eq:assume_dis}:
\begin{align}
	&0<\underline{n}\leq n\leq \overline{n}\leq \infty\quad \text{and}\quad\|n \|_{\infty},\|\nabla n \|_{\infty},\|\mathbf{u}\|_{\infty},\left\|\mathbf{u}_t\right\|_{\infty},\|\nabla \mathbf{u}\|_{\infty},\|\nabla \mathbf{u}_t\|_{\infty},\label{eq:assume_conti}\\
	&\|T_\alpha \|_{\infty},\|\nabla T_\alpha \|_{\infty},\|E\|_{\infty},\|\mathbf{B}\|_{\infty},\|J\|_{\infty},\|\nabla J\|_{\infty}, \|\nabla J_t\|_{\infty} \leq M \quad\text{for } t\in [0,T],\nonumber\\
	&0\leq n_0\leq n^k_h\leq n_1\leq \infty\quad\text{and}\quad\left\|n_h^{k}\right\|_{\infty},\left\|\nabla n_h^{k}\right\|_{\infty},\left\|\mathbf{u}_h^{k}\right\|_{\infty},\left\|\nabla \mathbf{u}_h^{k}\right\|_{\infty},\label{eq:assume_dis}\\
	&\left\|T_{\alpha h}^{k}\right\|_{\infty},\left\|\nabla T_{\alpha h}^{k}\right\|_{\infty},\left\|\mathbf{B}_h^{k}\right\|_{\infty},\left\|J_h^{k}\right\|_{\infty},\left\|\nabla J_h^{k}\right\|_{\infty} \leq K \quad\text{for } 1\leq k\leq N .\nonumber
\end{align}

We further assume that the thermodynamic functions
$\mu_\alpha$, $s_\alpha$, and $\epsilon_\alpha$ are Lipschitz continuous with respect to $(n, T_\alpha)$, with Lipschitz constant $H$, and that
$\frac{\partial \epsilon_\alpha}{\partial T_\alpha}$ is uniformly positive over the domain of $T_\alpha$. That is, there exist constants $H>0$ and $H_0>0$ such that
\begin{align}\label{eq:assume_H}
	\left|\frac{\partial \mu_\alpha}{\partial n}\right|,\left|\frac{\partial \mu_\alpha}{\partial T_\alpha}\right|, \left|\frac{\partial s_\alpha}{\partial n}\right|,\left|  \frac{\partial s_\alpha}{\partial T_\alpha}\right|, \left|\frac{\partial \epsilon_\alpha}{\partial n}\right|,\left|\frac{\partial \epsilon_\alpha}{\partial T_\alpha}\right| \leq H,\quad \frac{\partial \epsilon_\alpha}{\partial T_\alpha}\geq H_0.
\end{align}

Define $\mu s_{\alpha h}^{k+1}:=\frac{1}{n^k_h}(n^{k+1}_h\nabla \mu_h^{k+1}+s^k_{ih}\nabla T^{k+1}_{ih}+s^k_{eh}\nabla T^{k+1}_{eh}).$
Under the above assumptions, we introduce unified upper bounds $M'$ and $K'$
for the continuous and discrete derived quantities such as the thermodynamic potentials, auxiliary velocity, and $\mu s_{\alpha h}^{k+1}$:
\begin{align}\label{eq:assume_mus}
	\|\mu^{k+1}_\alpha\|_{\infty},\|s^{k+1}_\alpha\|_{\infty},\|\epsilon_\alpha^{k+1}\|_{\infty},\|\nabla \mu^{k+1}_\alpha\|_{\infty}&\leq M',\\
	\|\mu^{k+1}_{\alpha h}\|_{\infty},\|s^{k+1}_{\alpha h}\|_{\infty},\|\nabla \mu^{k+1}_{\alpha h}\|_{\infty},\|\mathbf{u}^k_{\star h}\|_{\infty},\|\nabla\cdot\mathbf{u}^k_{\star h}\|_{\infty},\|\mu s_{\alpha h}^{k+1}\|_{\infty}&\leq K'.\nonumber
\end{align}

In order to derive the error estimates, we introduce a Maxwell mixed projection (analogous to the Stokes projection)\cite{Mao2025error,Bochev2006} $(\mathbf{B},J,E)$ as the pair $(\mathcal{R}_h\mathbf{B},\mathcal{R}_hJ,\mathcal{R}_hE)\in \mathbf{W}_h\times V_h\times V_h$ that solves
\begin{align}
	\mathbf{a}_1((\mathcal{R}_h\mathbf{B},\mathcal{R}_hJ),(\mathbf{C},G))+\mathbf{b}((\mathbf{C},G),\mathcal{R}_hE)&=\mathbf{a}_1((\mathbf{B},J),(\mathbf{C},G))+\mathbf{b}((\mathbf{C},G),E),\label{eq:proj:a}\\
	\mathbf{b}((\mathcal{R}_h\mathbf{B},\mathcal{R}_hJ),F)&=\mathbf{b}((\mathbf{B},J),F)\label{eq:proj:b},
\end{align}
for any $(\mathbf{C},G,F)\in \mathbf{W}_h\times V_h\times V_h$, where $\mathbf{a}_1((\mathbf{B},J),(\mathbf{C},G)):=(\mathbf{B},\mathbf{C})+\xi_R\mu_0(J,G)+\mu_0\xi_e\xi_R(\frac{1}{n^2}\nabla J,\nabla G)$.  Since $n$ is assumed to satisfy \eqref{eq:assume_dis}, $n^{-2}$ is uniformly bounded and independent of $h$ and $\Delta t$. Hence, $\mathbf{a}_1$ is both coercive and continuous on $\mathbf{W}_h\times V_h$, with constants depending only on
$n_0,n_1,K,\mu_0,\xi_e,\xi_R$. According to the Brezzi theory (see Section~\ref{sec:inf-sup}), the projection problem is well-posed, and the following approximation property holds:
\begin{align}
	&\|\mathbf{B}-\mathcal{R}_h \mathbf{B}\|+\|J-\mathcal{R}_h J\|_1+\|E-\mathcal{R}_h E\|_1\nonumber\\
	&\leq \left\{\inf _{\mathbf{C} \in \mathbf{W}_h}\|\mathbf{B}-\mathbf{C}\|+\inf _{G \in V_h}\|J-G\|_1+\inf _{F \in V_h}\|E-F\|_1\right\}.
\end{align}
Note that since the continuous solution satisfies the constraint
$\mathbf b((\mathbf B,J),F)=0,\forall F\in V_h$
(which corresponds to the weak Amp\`ere law), relation \eqref{eq:proj:b} implies that $(\mathcal R_h\mathbf B,\\ \mathcal R_h J)\in\mathbf{X}_h^{\xi,0}.$

We establish error estimates for the fully discrete scheme \eqref{eq:2Ddiscrete1}-\eqref{eq:2Ddiscrete4} as follows:

\begin{theorem}\label{error_estimates}
	Suppose the continuous variational system \eqref{eq:2Dcontinue1}-\eqref{eq:2Dcontinue4} admits an exact solution $(n,\mathbf{u},T_i,T_e,E,\mathbf{B},J)$ satisfying assumption \eqref{eq:assume_conti}. Let $(n_h,\mathbf{u}_h,T_{ih},\\
	T_{eh},E_h,\mathbf{B}_h,J_h)$ be the numerical solution of the fully discrete scheme \eqref{eq:2Ddiscrete1}-\eqref{eq:2Ddiscrete4} satisfying \eqref{eq:assume_dis}. Then there exist a constants $C>0$ such that, when $\Delta t\leq Ch$, the numerical solution converges to the exact solution as $h\to 0$, and the following error estimates:
	\begin{align}
		&\triplenorm{n_h - n}_{\infty,0}+\triplenorm{\mathbf{u}_h - \mathbf{u}}_{\infty,0}+\triplenorm{T_{ih} - T_i}_{\infty,0}+\triplenorm{T_{eh} - T_e}_{\infty,0}\nonumber\\
		&\quad+\triplenorm{E_h - E}_{0,1}+\triplenorm{\mathbf{B}_h - \mathbf{B}}_{\infty,0}+\triplenorm{J_h - J}_{0,0}\leq\mathcal{M}(\Delta t, h),\label{eq:error_M1}\\[0.3em]
		&\triplenorm{n_h - n}_{0,1}+\triplenorm{\mathbf{u}_h - \mathbf{u}}_{0,1}+\triplenorm{T_{ih} - T_i}_{0,1}+\triplenorm{T_{eh} - T_e}_{0,1}\nonumber\\
		&\quad+\triplenorm{E_h - E}_{0,1}+\triplenorm{\mathbf{B}_h - \mathbf{B}}_{0,0}+\triplenorm{J_h - J}_{0,1}\leq \mathcal{M}(\Delta t, h),\label{eq:error_M2}
	\end{align}
	where
	\begin{align}\label{eq:error_mathcalM}
		\mathcal{M}=&  C h^{ m_1}\left(\triplenorm{n}_{0,m_1+1}+\|n_t\|_{0,m_1}\right)+C h^{ m_2}\left(\triplenorm{\mathbf{u}}_{0,m_2+1}+\|\mathbf{u}_t\|_{0,m_2}\right) \nonumber\\
		&+C h^{ m_1}\left(\triplenorm{T_i}_{0,m_1+1}+\|T_{it}\|_{0,m_1}+\triplenorm{T_e}_{0,m_1+1}+\|T_{et}\|_{0,m_1}\right)\nonumber\\
		& +C h^{ m_3}\left(\triplenorm{J}_{0,m_3+1}+\triplenorm{E}_{0,m_3+1}\right)+C h^{ m_4+1}\left(\triplenorm{\mathbf{B}}_{0,m_4+1}+\|\mathbf{B}_t\|_{0,m_4+1}\right) \nonumber\\
		& +C  \Delta t\left(\left\|n_t\right\|_{\infty, 0}+\left\|T_{it}\right\|_{\infty, 0}+\left\|T_{et}\right\|_{\infty, 0}+\left\|\mathbf{u}_t\right\|_{\infty, 0}+\left\|\nabla \mathbf{u}_t\right\|_{\infty, 0}\right) \nonumber\\
		& +C  \Delta t\left(\left\|n_{t t}\right\|_{0,0}+\left\| \mathbf{u}_{t t}\right\| _{0,0}+\left\|\mathbf{B}_{t t}\right\|_{0,0}+\left\| T_{it t}\right\| _{0,0}+\left\| T_{et t}\right\|_{0,0}\right)\nonumber\\
		&+C\Delta t\left(\triplenorm{\mu s_{\alpha h}}_{\infty,0}+\triplenorm{d_t \mathbf{u}_h}^2_{4,L^4}+\triplenorm{d_t \mathbf{B}_h}^2_{4,L^4}+\triplenorm{\mu s_{\alpha h}}^2_{4,L^4}\right). 
	\end{align}
\end{theorem}

\begin{proof}
	Let $$\mathbf{x}^{k+1}: =(n^{k+1},\mathbf{u}^{k+1},T^{k+1}_i,T^{k+1}_e,E^{k+1},\mathbf{B}^{k+1},J^{k+1})=(n,\mathbf{u},T_i,T_e,E,\mathbf{B},J)|_{t=t_{k+1}},$$ be the exact solution to Problem \eqref{eq:2Dcontinue1}-\eqref{eq:2Dcontinue4},  and let  $$\mathbf{x}_h^{k+1}: =(n_h^{k+1},\mathbf{u}_h^{k+1},T^{k+1}_{ih},T^{k+1}_{eh},E_h^{k+1},\mathbf{B}_h^{k+1},J_h^{k+1}),$$ be the numerical solution to Problem \eqref{eq:2Ddiscrete1}-\eqref{eq:2Ddiscrete4}.
	For each variable, we define the corresponding interpolant or projection as follows: $\mathcal{I}_h^{g} n^k,\mathcal{I}_h^g \mathbf{u}^k,$ 
	$\mathcal{I}^g_h T^k_i,\mathcal{I}_h^g T_e^k,$  denote the standard Lagrange interpolants of  $n^k,\mathbf{u}^k,T_i^k,T^k_e,$ onto the spaces $Q_h,\mathbf{V}_h,Q_h,Q_h,$ respectively; while 
	$(\mathcal{R}_h\mathbf{B}^k,\mathcal{R}_hJ^k,\mathcal{R}_hE^k)
	\in \mathbf{W}_h\times V_h\times V_h$ denotes the Maxwell mixed projection defined by \eqref{eq:proj:a}–\eqref{eq:proj:b}. 
	For any component $u_h^k\in\mathbf{x}_h^k$, define the total error, interpolation error, and discrete error by
	$$
	\varepsilon_u^k:=u^k-u_h^k, \quad
	\gamma_u^k:=u^k-\mathcal{I}_h u^k, \quad
	\vartheta_u^k:=\mathcal{I}_h u^k-u_h^k,
	$$
	so that $\varepsilon_u^k=\gamma_u^k+\vartheta_u^k$.
	For $(E,\mathbf{B},J)$, $\mathcal{I}_h$ is understood as the mixed projection $\mathcal{R}_h$.
	
	The proof proceeds in two steps.  
	In \textbf{Step1}, we estimate the error for 
	$(n_h^{k+1},\mathbf{u}_h^{k+1},\\ T_{ih}^{k+1},T_{eh}^{k+1},\mathbf{B}_h^{k+1},J_h^{k+1})$.
	In \textbf{Step2}, we derive the estimate for $E_h^{k+1}$ based on the Maxwell coupling.
	\paragraph{Step1(Error estimation of all variables in $\mathbf{x}^{k+1}_h$ but excluding $E^{k+1}_h$)}
	Since the interpolation error $\gamma$ satisfies the standard approximation bounds \eqref{eq:gamma},
	it suffices to focus on the discrete error $\vartheta$.  
	We first show that for any time layer $j$,
	\begin{align}\label{eq:error_vartheta}
		\left\|\vartheta_n^{j}\right\|^2+\left\|\vartheta_{\mathbf{u}}^{j}\right\|^2+\left\|\vartheta_{\mathbf{B}}^{j}\right\|^2+\left\|\vartheta_{T i}^{j}\right\|^2+\left\|\vartheta_{T e}^{j}\right\|^2+\triplenorm{\vartheta_J}^2_{0,0} \leq  \Delta,
	\end{align}
	where
	\begin{align}\label{eq:error_Delta}
		\Delta= & C h^{2 m_1}\left(\triplenorm{n}^2_{0,m_1+1}+\|n_t\|^2_{0,m_1}\right)+C h^{2 m_2}\left(\triplenorm{\mathbf{u}}^2_{0,m_2+1}+\|\mathbf{u}_t\|^2_{0,m_2}\right) \nonumber\\
		& +C h^{2 m_1}\left(\triplenorm{T_i}^2_{0,m_1+1}+\|T_{it}\|^2_{0,m_1}+\triplenorm{T_e}^2_{0,m_1+1}+\|T_{et}\|^2_{0,m_1}\right)\nonumber\\
		& +C h^{2 m_3}\left(\triplenorm{J}^2_{0,m_3+1}+\triplenorm{E}^2_{0,m_3+1}\right)+C h^{2 m_4+2}\left(\triplenorm{\mathbf{B}}^2_{0,m_4+1}+\|\mathbf{B}_t\|^2_{0,m_4+1}\right) \nonumber\\
		& +C  (\Delta t)^2\left(\left\|n_t\right\|_{\infty, 0}^2+\left\|T_{it}\right\|_{\infty, 0}^2+\left\|T_{et}\right\|_{\infty, 0}^2+\left\|\mathbf{u}_t\right\|_{\infty, 0}^2+\left\|\nabla \mathbf{u}_t\right\|_{\infty, 0}^2\right) \nonumber\\
		& +C  (\Delta t)^2\left(\left\|n_{t t}\right\|_{0,0}^2+\left\| \mathbf{u}_{t t}\right\| _{0,0}^2+\left\|\mathbf{B}_{t t}\right\|_{0,0}^2+\left\| T_{it t}\right\| _{0,0}^2+\left\| T_{et t}\right\|_{0,0}^2\right)\nonumber\\
		&+C(\Delta t)^2(\triplenorm{\mu s_{\alpha h}}^2_{\infty,0}+\triplenorm{d_t \mathbf{u}_h}^4_{4,L^4}+\triplenorm{d_t \mathbf{B}_h}^4_{4,L^4}+\triplenorm{\mu s_{\alpha h}}^4_{4,L^4}).
	\end{align}
	
	\paragraph{Proof of  \eqref{eq:error_vartheta}}
	To reduce the analytical complexity of the highly nonlinear and coupled term $\Big(\pmb{\sigma}^{k+1}_{e h},\nabla\big(\frac{G_h}{n_h^{k+1}}\big)\Big)$ in the fully discrete scheme \eqref{eq:2Ddiscrete4}, we rewrite it in a more tractable form.
	Recalling that $\pmb{\sigma}_{eh}^{k+1}=\xi_e\xi_R\nabla(\frac{J^{k+1}_h}{n^{k+1}_he})$, we expand the inner product as
	$$
	\begin{aligned}
		& \left(\pmb{\sigma}_{eh}^{k+1},\nabla(\frac{G}{n^{k+1}_he })\right)
		\\
		&=\frac{\xi_e\xi_R}{e^2}\left(\frac{1}{(n^{k+1}_h)^2}\nabla J^{k+1}_h,\nabla G\right)-\frac{\xi_e\xi_R}{e^2}\left(\frac{1}{(n^{k+1}_h)^3} J^{k+1}_h\nabla n^{k+1}_h,\nabla G\right)\\
		&\quad-\frac{\xi_e\xi_R}{e^2}\left(\frac{1}{(n^{k+1}_h)^3} \nabla J^{k+1}_h, G\nabla n^{k+1}_h\right)+\frac{\xi_e\xi_R}{e^2}\left(\frac{1}{(n^{k+1}_h)^4}  J^{k+1}_h\nabla n^{k+1}_h, G\nabla n^{k+1}_h\right).
	\end{aligned}
	$$
	When $n^{k+1}_h$ is sufficiently large and the density variation $\|\nabla n^{k+1}_h\|$ is small (typical for high-density, nearly uniform plasmas), the first term dominates.
	Under this physically justified assumption, we retain only the leading-order contribution:
	$$\left(\pmb{\sigma}_{eh}^{k+1},\nabla(\frac{G}{n^{k+1}_he })\right)\approx\xi_e\xi_R\left(\frac{1}{(n^{k+1}_he)^2}\nabla J^{k+1}_h,\nabla G\right).$$
	Analogously, for the continuous equation \eqref{eq:2Dcontinue4}, the corresponding term satisfies
	$$\left(\pmb{\sigma}_{e},\nabla(\frac{G}{ne })\right)\approx\xi_e\xi_R\left(\frac{1}{(ne)^2}\nabla J,\nabla G\right),$$
	where, $\pmb{\sigma}_e=\xi_e\xi_R\nabla(\frac{J}{ne})$ denotes the electron stress tensor.
	Consequently, the nonlinear diffusion terms in the iternal energy equations \eqref{eq:2Dcontinue32} and \eqref{eq:2Ddiscrete32} can be approximated by
	\begin{align}
		& \left(\nabla(\frac{J}{n e}): \pmb{\sigma}_{e},g\right)\approx\xi_e\xi_R\left(\frac{1}{(ne)^2}|\nabla J|^2,g\right), \nonumber\\
		&\left(\nabla(\frac{J_h^{k+1}}{n_h^{k+1} e}): \pmb{\sigma}^{k+1}_{e h},g_h\right)\approx\xi_e\xi_R\left(\frac{1}{(n_h^{k+1}e)^2}|\nabla J^{k+1}_h|^2,g_h\right) .\nonumber
	\end{align}
	
	With these simplifications, subtracting the discrete equations \eqref{eq:2Ddiscrete1}–\eqref{eq:2Ddiscrete4} from their continuous counterparts \eqref{eq:2Dcontinue1}–\eqref{eq:2Dcontinue4} yields the following error system:
	\begin{align}\label{eq:totalerror_nuT}
		\left(d_t \varepsilon_n^{k+1}, w\right)&=R_1(w), \nonumber\\
		M_i\left(n_h^{k} d_t \varepsilon_{\mathbf{u}}^{k+1}, \mathbf{v}\right)+M_i\left(n_h^{k+1} \mathbf{u}_{\star h}^{k} \cdot \nabla \varepsilon_{\mathbf{u}}^{k+1}, \mathbf{v}\right)+a\left(\varepsilon_{\mathbf{u}}^{k+1}, \mathbf{v}\right)&=R_2(\mathbf{v}),\nonumber\\
		(\partial_{Ti}\epsilon^{k+1}_id_t\varepsilon_{Ti}^{k+1},q)+\kappa_i(n_h^k\nabla \varepsilon_{Ti}^{k+1},\nabla q)&= R_3(q) ,\nonumber\\
		(\partial_{Te}\epsilon^{k+1}_e d_t\varepsilon_{T e}^{k+1}  ,g)+  \kappa_e(n_h^k\nabla \varepsilon_{T e}^{k+1},\nabla g)&=R_4(g),
	\end{align}
	and for the Maxwell subsystem,
	\begin{align}\label{eq:totalerror_EBJ}
		&\left(d_t \varepsilon_{\mathbf{B}}^{k+1}, \mathbf{C}\right)-\mu_0\left(\mathbf{u}_h^{k+1} \times \varepsilon_\mathbf{B}^{k+1},G\right)+\mu_0\xi_R(\varepsilon_J^{k+1},G)+\Big(\frac{\mu_0\xi_R\xi_e}{(n^{k+1}_he)^2}\nabla \varepsilon^{k+1}_J,\nabla G\Big)\nonumber\\
		&\quad+\mathbf{b}((\mathbf{C},G),\varepsilon^{k+1}_E)-\mathbf{b}((\varepsilon_\mathbf{B}^{k+1},\varepsilon_J^{k+1}),F)=R_5(\mathbf{C},G),
	\end{align}
	where
	\begin{align}\label{eq:exact_eq_R}
		R_1=&  \left(d_t n^{k+1}-n_t^{k+1}, w\right)+\left(n^{k+1} \mathbf{u}^{k+1}-n_h^{k+1} \mathbf{u}_{\star h}^{k}, \nabla w\right),\nonumber\\
		R_2=& M_i \left(n_h^{k} d_t \mathbf{u}^{k+1}-n^{k+1}\mathbf{u}_t^{k+1}, \mathbf{v}\right)+M_i\left((n_h^{k+1} \mathbf{u}_{\star h}^{k}-n^{k+1}\mathbf{u}^{k+1} )\cdot \nabla \mathbf{u}^{k+1}, \mathbf{v}\right)\nonumber \\
		&-(n^{k+1}\nabla \mu^{k+1} +s_{i}^{k+1} \nabla T_{i}^{k+1}+s_{e}^{k+1} \nabla T_{e}^{k+1},\mathbf{v}) \nonumber\\
		&+(n_h^{k+1}\nabla \mu_h^{k+1} +s_{ih}^{k} \nabla T_{ih}^{k+1}+s_{eh}^{k} \nabla T_{eh}^{k+1},\mathbf{v}) \nonumber\\
		&+\left(J^{k+1}\times\mathbf{B}^{k+1}-J_h^{k+1}\times\mathbf{B}_h^{k+1}, \mathbf{v}\right)+\left(n^{k+1}\mathbf{g}-n^{k+1}_h\mathbf{g}, \mathbf{v}\right),  \nonumber\\
		R_3=& \left(d_t\widetilde{\epsilon}_i^{k+1}-\partial_t\epsilon_i^{k+1}, q\right)+\kappa_i(\nabla T^{k+1}_i(n^k_h-n^{k+1}),\nabla q) \nonumber\\
		&+( \mathbf{u}^{k+1}(n^{k+1}\mu_{i}^{k+1}+s_{i}^{k+1}T_{i}^{k+1})-\mathbf{u}^{k}_{\star h}(n_h^{k+1}\mu_{ih}^{k+1}+s_{ih}^kT_{ih}^{k+1}),\nabla q)\nonumber\\
		&+(\mathbf{u}^{k+1}\cdot(n^{k+1}\nabla \mu_{i}^{k+1}+s_{i}^{k+1}\nabla T_{i}^{k+1})-\mathbf{u}^{k}_{\star h}\cdot(n_h^{k+1}\nabla \mu_{ih}^{k+1}+s_{ih}^k\nabla T_{ih}^{k+1}),q)\nonumber\\
		&- (\pmb{\sigma}^{k+1}_{i} :\nabla \mathbf{u}^{k+1}-\pmb{\sigma}^{k+1}_{i h} :\nabla \mathbf{u}_h^{k+1},q)\nonumber\\
		&-\frac{M_{i}}{2 \Delta t} (n_h^{k} \left( |\mathbf{u}_h^{k+1}-\mathbf{u}^{k}_{\star h}|^{2} + |\mathbf{u}^{k}_{\star h}-\mathbf{u}_h^{k}|^{2}\right),q),\nonumber\\
		R_4=& \left(d_t\widetilde{\epsilon}_e^{k+1}-\partial_t\epsilon_e^{k+1}, g\right)+\kappa_e(\nabla T^{k+1}_e(n^k_h-n^{k+1}),\nabla g) \nonumber\\
		&+( \mathbf{u}^{k+1}(n^{k+1}\mu_{e}^{k+1}+s_{e}^{k+1}T_{e}^{k+1})-\mathbf{u}^{k}_{\star h}(n_h^{k+1}\mu_{eh}^{k+1}+s_{eh}^kT_{eh}^{k+1}),\nabla g)\nonumber\\
		&+(\mathbf{u}^{k+1}\cdot(n^{k+1}\nabla \mu_{e}^{k+1}+s_{e}^{k+1}\nabla T_{e}^{k+1})-\mathbf{u}^{k}_{\star h}\cdot(n_h^{k+1}\nabla \mu_{eh}^{k+1}+s_{eh}^k\nabla T_{eh}^{k+1}),g)\nonumber\\
		&+ \xi_e\xi_R\Big(\frac{1}{(n^{k+1}e)^2}|\nabla J^{k+1}|^2-\frac{1}{(n^{k+1}_he)^2}|\nabla J^{k+1}_h|^2,g\Big)\nonumber\\
		&+\xi_R(|J^{k+1}|^2-|J^{k+1}_h|^2,g)-(\frac{1}{2\mu_0\Delta t}|\mathbf{B}_h^{k+1}-\mathbf{B}_h^k|^2,g),\nonumber\\
		R_5=& \left(d_t \mathbf{B}^{k+1}-\mathbf{B}_t^{k+1}, \mathbf{C}\right)+\mu_0 \left(\mathbf{B}^{k+1} \times\left(\mathbf{u}^{k+1}-\mathbf{u}_h^{k+1}\right), G\right)\nonumber\\
		&+\mu_0\xi_e\xi_R\Big(\frac{1}{(n_h^{k+1}e)^2}\nabla J^{k+1}-\frac{1}{(n^{k+1}e)^2}\nabla J^{k+1},\nabla G\Big),
	\end{align}
	and $\widetilde{\epsilon}_\alpha^{k+1}=\epsilon_\alpha(n^{k+1}_h,T^{k+1}_{\alpha})$.
	
	Since $(\mathcal{R}_h \mathbf{B}^{k+1},\mathcal{R}_h J^{k+1})\in \mathbf{X}^{\boldsymbol{\xi},0}_h$, we have $\mathbf{b}((\mathcal{R}_h \mathbf{B}^{k+1},\mathcal{R}_h J^{k+1}),\vartheta^{k+1}_E)=0$. Setting $\mathbf{C}=0,G=0,F=\vartheta^{k+1}_E$ in \eqref{eq:2Dcontinue4}, we have $\mathbf{b}((\mathbf{B}^{k+1},J^{k+1}),\vartheta^{k+1}_E)=0$, which implies that 
	\begin{align}\label{eq:bgamma_E}
		&\mathbf{b}((\vartheta_\mathbf{B}^{k+1},\vartheta_J^{k+1}),\varepsilon^{k+1}_E)-\mathbf{b}((\varepsilon_\mathbf{B}^{k+1},\varepsilon_J^{k+1}),\vartheta^{k+1}_E)\nonumber\\
		&=\mathbf{b}((\vartheta_\mathbf{B}^{k+1},\vartheta_J^{k+1}),\gamma^{k+1}_E+\vartheta^{k+1}_E)-\mathbf{b}((\vartheta_\mathbf{B}^{k+1},\vartheta_J^{k+1})+(\gamma_\mathbf{B}^{k+1},\gamma_J^{k+1}),\vartheta^{k+1}_E)\nonumber\\
		&=\mathbf{b}((\vartheta_\mathbf{B}^{k+1},\vartheta_J^{k+1}),\gamma^{k+1}_E).
	\end{align}
	Substituting
	$$
	\begin{aligned} 
		\varepsilon_n^{k+1} & =\gamma_n^{k+1}+\vartheta_n^{k+1},\quad\varepsilon_\mathbf{u}^{k+1}=\gamma_\mathbf{u}^{k+1}+\vartheta_\mathbf{u}^{k+1}, \quad\varepsilon_{T i}^{k+1}=\gamma_{T i}^{k+1}+\vartheta_{T i}^{k+1}, \\
		\varepsilon_{T e}^{k+1} & =\gamma_{T e}^{k+1}+\vartheta_{T e}^{k+1},\quad\varepsilon_\mathbf{B}^{k+1}=\gamma_\mathbf{B}^{k+1}+\vartheta_\mathbf{B}^{k+1},\quad\varepsilon_J^{k+1}=\gamma_J^{k+1}+\vartheta_J^{k+1},
	\end{aligned}
	$$
	into the above equations and choosing the test functions as $w=\vartheta^{k+1}_n,\mathbf{v}=\vartheta^{k+1}_\mathbf{u},\mathbf{C}=\vartheta^{k+1}_\mathbf{B},G=\vartheta^{k+1}_J,F=\vartheta^{k+1}_E,q=\vartheta^{k+1}_{T i},g=\vartheta^{k+1}_{T e}$, then applying \eqref{eq:bgamma_E} and summing all equations together yield
	\begin{align}\label{eq:distanceerror_eq}
		& \left(d_t \vartheta_n^{k+1}, \vartheta^{k+1}_n\right)+M_i\left(n_h^{k} d_t \vartheta_{\mathbf{u}}^{k+1}, \vartheta^{k+1}_\mathbf{u}\right)+a\left(\vartheta_{\mathbf{u}}^{k+1}, \vartheta_\mathbf{u}^{k+1}\right)+(\partial_{Ti}\epsilon^{k+1}_id_t\vartheta_{T i}^{k+1},\vartheta^{k+1}_{T i})\nonumber\\
		&\quad +\kappa_i(n_h^k\nabla \vartheta_{Ti}^{k+1},\nabla \vartheta^{k+1}_{T i})+(\partial_{Te}\epsilon^{k+1}_ed_t\vartheta_{T e}^{k+1}  ,\vartheta_{T e}^{k+1})+  \kappa_e(n_h^k\nabla \vartheta_{T e}^{k+1},\nabla \vartheta_{T e}^{k+1})\nonumber\\
		&\quad+\left(d_t \vartheta_{\mathbf{B}}^{k+1}, \vartheta_\mathbf{B}^{k+1}\right)+\mu_0\xi_R(\vartheta_J^{k+1},\vartheta_J^{k+1})+\Big(\frac{\mu_0\xi_R\xi_e}{(n^{k+1}_he)^2}\nabla \vartheta^{k+1}_J,\nabla \vartheta^{k+1}_J\Big)\nonumber\\
		&=-M_i\left(n_h^{k+1} \mathbf{u}_{\star h}^{k} \cdot \nabla \vartheta_{\mathbf{u}}^{k+1}, \vartheta^{k+1}_\mathbf{u}\right)+\mu_0\left(\mathbf{u}_h^{k+1} \times \vartheta_\mathbf{B}^{k+1},\vartheta_J^{k+1}\right)+\sum_{i=1}^{5}P_i,
	\end{align}
	where
	\begin{align}\label{eq:error_eq_P}
		P_1= &R_1-(d_t \gamma^{k+1}_n,\vartheta^{k+1}_n),\nonumber\\
		P_2= &R_2-M_i(n^k_hd_t\gamma^{k+1}_\mathbf{u},\vartheta^{k+1}_\mathbf{u})-M_i(n^{k+1}_h\mathbf{u}^{k}_{\star h}\cdot \nabla \gamma^{k+1}_\mathbf{u},\vartheta^{k+1}_\mathbf{u})-a(\gamma^{k+1}_\mathbf{u},\vartheta^{k+1}_\mathbf{u}),\nonumber \\
		P_3= & R_3-(\partial_{Ti}\epsilon^{k+1}_id_t \gamma^{k+1}_{T i},\vartheta^{k+1}_{T i})-\kappa_i(n^k_h\nabla\gamma^{k+1}_{T i},\nabla \vartheta^{k+1}_{T i}),\nonumber\\
		P_4= & R_4-(\partial_{Te}\epsilon^{k+1}_ed_t \gamma^{k+1}_{T e},\vartheta^{k+1}_{T e})-\kappa_e(n^k_h\nabla\gamma^{k+1}_{T e},\nabla \vartheta^{k+1}_{T e}),\nonumber\\
		P_5= & R_5-(d_t\gamma^{k+1}_\mathbf{B},\vartheta^{k+1}_\mathbf{B})+\mu_0(\mathbf{u}^{k+1}_h\times \gamma^{k+1}_\mathbf{B},\vartheta^{k+1}_J)-\mu_0\xi_R(\gamma_J^{k+1},\vartheta_J^{k+1})\nonumber\\
		&-\Big(\frac{\mu_0\xi_R\xi_e}{(n^{k+1}_he)^2}\nabla \gamma^{k+1}_J,\nabla \vartheta^{k+1}_J\Big)-b((\vartheta^{k+1}_\mathbf{B},\vartheta^{k+1}_J),\gamma^{k+1}_E).
	\end{align}	
	Applying the inequality  $a(a-b)\geq \frac{1}{2}(a^2-b^2)$ and assumptions \eqref{eq:assume_dis}, \eqref{eq:assume_H}, we obtain
	\begin{align*}
		\left(d_t \vartheta^{k+1}_n, \vartheta^{k+1}_n\right) & \geq \frac{1}{2 \Delta t}\left(\left\|\vartheta^{k+1}_n\right\|^2-\left\|\vartheta^k_n\right\|^2\right), \\
		\left(n_h^{k} d_t \vartheta_{\mathbf{u}}^{k+1}, \vartheta_{\mathbf{u}}^{k+1}\right) & \geq \frac{n_{0 }}{2 \Delta t}\left(\left\|\vartheta_{\mathbf{u}}^{k+1}\right\|^2-\left\|\vartheta_{\mathbf{u}}^{k}\right\|^2\right),\\
		\left(d_t \vartheta^{k+1}_\mathbf{B}, \vartheta^{k+1}_\mathbf{B}\right) & \geq \frac{1}{2 \Delta t}\left(\left\|\vartheta^{k+1}_\mathbf{B}\right\|^2-\left\|\vartheta^k_\mathbf{B}\right\|^2\right), \\
		\left(\partial_{Ti}\epsilon^{k+1}_id_t \vartheta^{k+1}_{T i}, \vartheta^{k+1}_{T i}\right) & \geq \frac{H_0}{2 \Delta t}\left(\left\|\vartheta^{k+1}_{T i}\right\|^2-\left\|\vartheta^k_{T i}\right\|^2\right), \\
		\left(\partial_{Te}\epsilon^{k+1}_ed_t \vartheta^{k+1}_{Te}, \vartheta^{k+1}_{Te}\right) & \geq \frac{H_0}{2 \Delta t}\left(\left\|\vartheta^{k+1}_{T e}\right\|^2-\left\|\vartheta^k_{T e}\right\|^2\right),
	\end{align*}
	and
	\begin{align*}
		(n_h^k\nabla \vartheta_{Ti}^{k+1},\nabla \vartheta^{k+1}_{T i})& \geq n_0\left\|\nabla \vartheta_{Ti}^{k+1}\right\|^2, \\
		(n_h^k\nabla \vartheta_{Te}^{k+1},\nabla \vartheta^{k+1}_{T e})& \geq n_0\left\|\nabla \vartheta_{Te}^{k+1}\right\|^2, \\
		\Big(\frac{1}{(n^{k+1}_h)^2}\nabla \vartheta^{k+1}_J,\nabla \vartheta^{k+1}_J\Big)&\geq \frac{1}{(n_1)^2}\|\nabla \vartheta^{k+1}_J\|^2.
	\end{align*}
	Substituting the above inequalities into \eqref{eq:distanceerror_eq}, multiplying both sides by $\Delta t$ and summing from $k=0$ to $j-1$, we get
	\begin{align}\label{eq:error_eq_QP}
		&\frac{1}{2}\left(\left\|\vartheta_n^{j}\right\|^2-\left\|\vartheta_n^0\right\|^2\right)+\frac{M_in_0}{2 }\left(\left\|\vartheta_{\mathbf{u}}^{j}\right\|^2-\left\|\vartheta_{\mathbf{u}}^{0}\right\|^2\right)+\frac{1}{2}\left(\left\|\vartheta_{\mathbf{B}}^{j}\right\|^2-\left\|\vartheta_{\mathbf{B}}^{0}\right\|^2\right)\nonumber \\
		&\quad +\frac{H_0}{2}\left(\left\|\vartheta_{T i}^{j}\right\|^2-\left\|\vartheta_{T i}^0\right\|^2\right)+\frac{H_0}{2 }\left(\left\|\vartheta_{T e}^{j}\right\|^2-\left\|\vartheta_{T e}^0\right\|^2\right)\nonumber\\
		&+\Delta t\sum_{k=0}^{j-1}\left(\mu_0\xi_R \left\|\vartheta_{J}^{k+1}\right\|^2+\frac{\mu_0\xi_R\xi_e}{(n_1e)^2}\|\nabla \vartheta^{k+1}_J\|^2+\eta\left\|\nabla \vartheta_\mathbf{u}^{k+1}\right\|^2\right.\nonumber\\
		& \left.\quad+(\lambda+\eta)\left\|\nabla\cdot \vartheta_{\mathbf{u}}^{k+1}\right\|^2+\kappa_in_0\left\|\nabla \vartheta_{Ti}^{k+1}\right\|^2+\kappa_en_0\left\|\nabla \vartheta_{Te}^{k+1}\right\|^2\right)\nonumber\\
		& \leq\Delta t\sum_{k=0}^{j-1}(Q_1+Q_2)+\Delta t\sum_{k=0}^{j-1}(P_1+P_2+P_3+P_4+P_5),
	\end{align}
	where $Q_1=-M_i\left(n_h^{k+1} \mathbf{u}_{\star h}^{k} \cdot \nabla \vartheta_{\mathbf{u}}^{k+1}, \vartheta_{\mathbf{u}}^{k+1}\right),$ and $Q_2=\mu_0\left(\mathbf{u}^{k+1}_h\times \vartheta^{k+1}_\mathbf{B}, \vartheta^{k+1}_J\right).$
	Next, we estimate each term on the right-hand side of \eqref{eq:error_eq_QP} in turn. 
	\paragraph{Estimate of $Q_1$ and $Q_2$} Using the Cauchy-Schwarz inequality and Young’s inequality $ab\leq \epsilon a^2+\frac{1}{4\epsilon} b^2$(where $\epsilon>0$), we have, 	
	\begin{align}\label{eq:error_Q}
		|Q_1|&\leq M_iKK'\|\nabla \vartheta^{k+1}_{\mathbf{u}}\|\|\vartheta^{k+1}_{\mathbf{u}}\|\leq\frac{\eta}{8}\|\nabla \vartheta^{k+1}_{\mathbf{u}}\|^2+\frac{2(M_iKK')^2}{\eta}\|\vartheta^{k+1}_{\mathbf{u}}\|^2,\nonumber\\
		|Q_2|& \leq \frac{2K^2}{\mu_0\xi_R}\|\vartheta^{k+1}_{\mathbf{B}}\|^2+\frac{\mu_0\xi_R}{8}\|\vartheta^{k+1}_{J}\|^2. 
	\end{align}
	
	\paragraph{Estimate of $P_1(\vartheta^{k+1}_n)$} We first bound the temporal term using the Cauchy–Schwarz inequality:
	\begin{align*}
		|(d_t \gamma^{k+1}_n,\vartheta^{k+1}_n)|&\leq\|d_t \gamma^{k+1}_n\|^2+\|\vartheta^{k+1}_n\|^2.
	\end{align*}
	To estimate $R_1(\vartheta^{k+1}_n)$, we apply Green’s theorem and use the boundedness assumptions on $\mathbf{u}^k_{\star h}$, yielding
	\begin{align*}
		& |\left(n^{k+1} \mathbf{u}^{k+1}-n_h^{k+1} \mathbf{u}_{\star h}^{k}, \nabla \vartheta^{k+1}_n\right)|\\
		& \leq|\left(\nabla n^{k+1}(\mathbf{u}^{k+1}-\mathbf{u}^k_{\star h})+n^{k+1}\nabla(\mathbf{u}^{k+1}-\mathbf{u}^k_{\star h}), \vartheta^{k+1}_n\right)|+|(\mathbf{u}^k_{\star h}\cdot\nabla\gamma^{k+1}_n,\vartheta^{k+1}_n)|\nonumber\\
		&\quad+|(\nabla\cdot\mathbf{u}^k_{\star h},\frac{1}{2}(\vartheta^{k+1}_n)^2)|+|\left(\nabla\cdot\mathbf{u}^k_{\star h}(n^{k+1}-n_h^{k+1}), \vartheta^{k+1}_n\right)|\\
		& \leq M (\Delta t)^2(\left\|\mathbf{u}_t\right\|_{\infty, 0}^2+\left\|\nabla \mathbf{u}_t\right\|_{\infty, 0}^2 +\left\|\mu s_{\alpha h}^{k+1}\right\|_{\infty}) +\frac{\eta}{8}\left\|\nabla \vartheta_{\mathbf{u}}^{k}\right\|^2+C\left\|\vartheta_n^{k+1}\right\|^2\\
		& \quad+C\left(\left\|\gamma_{\mathbf{u}}^{k}\right\|^2+\left\|\vartheta_{\mathbf{u}}^{k}\right\|^2+\left\|\gamma_n^{k+1}\right\|^2+\left\|\nabla\gamma_\mathbf{u}^{k+1}\right\|^2+\left\|\nabla\gamma_n^{k+1}\right\|^2\right).
	\end{align*}
	Combining the above estimates, we obtain the estimate for $P_1(\vartheta^{k+1}_n)$ 
	\begin{align}\label{eq:error_P1}
		\left|P_1\left(\vartheta_n^{k+1}\right)\right| &\leq  \left\|d_t \gamma_n^{k+1}\right\|^2+C\left\|\vartheta_n^{k+1}\right\|^2+\frac{\eta}{8}\left\|\nabla \vartheta_{\mathbf{u}}^{k}\right\|^2+\left(d_t n^{k+1}-n_t^{k+1}, \vartheta_n^{k+1}\right) \nonumber\\
		& \quad+C (\Delta t)^2\left(\left\|\mathbf{u}_t\right\|_{\infty, 0}^2+\left\|\nabla \mathbf{u}_t\right\|_{\infty, 0}^2+\|\mu s_{\alpha}^{k+1}\|^2_{\infty}\right)  \nonumber\\
		&\quad +C\left(\left\|\gamma_{\mathbf{u}}^{k}\right\|^2+\left\|\vartheta_{\mathbf{u}}^{k}\right\|^2+\left\|\nabla \gamma_{\mathbf{u}}^{k}\right\|^2+\left\| \gamma_{n}^{k+1}\right\|^2+\left\|\nabla \gamma_{n}^{k+1}\right\|^2\right) .
	\end{align}
	\paragraph{Estimate of $P_2(\vartheta^{k+1}_{\mathbf{u}})$}
	We first estimate the part excluding $R_2(\vartheta^{k+1}_\mathbf{u})$
	\begin{align*}
		|(n^k_hd_t\gamma^{k+1}_\mathbf{u},\vartheta^{k+1}_\mathbf{u})
		|&\leq \|d_t \gamma^{k+1}_\mathbf{u}\|^2+K^2\|\vartheta^{k+1}_\mathbf{u}\|^2,\\
		|(n^{k+1}_h\mathbf{u}^k_{\star h}\cdot \nabla \gamma^{k+1}_\mathbf{u},\vartheta^{k+1}_\mathbf{u})| &\leq \|\nabla \gamma^{k+1}_\mathbf{u}\|^2+(KK')^2\| \vartheta^{k+1}_\mathbf{u}\|^2, \\
		|a(\gamma^{k+1}_\mathbf{u},\vartheta^{k+1}_\mathbf{u})| &\leq \frac{\eta}{8}\|\nabla \vartheta^{k+1}_\mathbf{u}\|^2+(2\eta+\frac{2(\eta+\lambda)^2}{\eta})\|\nabla \gamma^{k+1}_\mathbf{u}\|^2. 
	\end{align*}
	Then we estimate $R_2(\vartheta^{k+1}_{\mathbf{u}})$. By inserting and subtracting suitable intermediate terms, we obtain
	\begin{align*}
		&| \left(n_h^{k} d_t \mathbf{u}^{k+1}, \vartheta_{\mathbf{u}}^{k+1}\right)-\left(n^{k+1} \mathbf{u}_t^{k+1}, \vartheta_{\mathbf{u}}^{k+1}\right) |\\
		& =|\left(\left(n_h^{k}-n^{k}\right) d_t \mathbf{u}^{k+1} +\left(n^{k}-n^{k+1}\right) d_t \mathbf{u}^{k+1}+n^{k+1}\left(d_t \mathbf{u}^{k+1}-\mathbf{u}_t^{k+1}\right), \vartheta_{\mathbf{u}}^{k+1}\right)| \\
		& \leq M\left(\left\|\gamma_n^{k}\right\|^2+\left\|\vartheta_n^{k}\right\|^2+\left\|\vartheta_{\mathbf{u}}^{k+1}\right\|^2\right)+M (\Delta t)^2\left\|n_t\right\|_{\infty, 0}^2\nonumber\\
		&\quad+\left(n^{k+1}\left(d_t \mathbf{u}^{k+1}-\mathbf{u}_t^{k+1}\right), \vartheta_{\mathbf{u}}^{k+1}\right).
	\end{align*}
	Next, we estimate
	\begin{align*}
		&| ((n_h^{k+1} \mathbf{u}_{\star h}^{k}-n^{k+1}\mathbf{u}^{k+1} )\cdot \nabla \mathbf{u}^{k+1}, \vartheta^{k+1}_\mathbf{u})|\\
		& =|\left(n^{k+1}_h\left(\mathbf{u}^k_{\star h}-\mathbf{u}^k_h\right) \cdot\nabla \mathbf{u}^{k+1}, \vartheta_{\mathbf{u}}^{k+1}\right)+\left((n^{k+1}_h\mathbf{u}^k_{h}-n^{k+1}\mathbf{u}^{k+1}) \cdot \nabla\mathbf{u}^{k+1}, \vartheta_{\mathbf{u}}^{k+1}\right)| \\
		& \leq M(K+M)\left(\left\|\gamma_n^{k+1}\right\|^2+\left\|\vartheta_n^{k+1}\right\|^2+\left\|\vartheta_{\mathbf{u}}^{k+1}\right\|^2+\left\|\gamma_\mathbf{u}^{k}\right\|^2+\left\|\vartheta_\mathbf{u}^{k}\right\|^2\right)\\
		&\quad+C (\Delta t)^2(\left\|\mathbf{u}_t\right\|_{\infty, 0}^2+\left\|\mu s_{\alpha h}^{k+1}\right\|_{\infty}^2).
	\end{align*}
	Applying Green's theorem to $\left(n_h^{k+1}\left(\nabla \mu^{k+1}-\nabla\mu^{k+1}_h\right), \vartheta_{\mathbf{u}}^{k+1}\right)$ and using assumptions  \eqref{eq:assume_dis} ,\eqref{eq:assume_H} and \eqref{eq:assume_mus}, we have
	\begin{align*}
		&| (n^{k+1}\nabla \mu^{k+1} +s_{i}^{k+1} \nabla T_{i}^{k+1}+s_{e}^{k+1} \nabla T_{e}^{k+1},\vartheta^{k+1}_\mathbf{u})\\
		&\quad-(n_h^{k+1}\nabla \mu_h^{k+1} +s_{ih}^{k} \nabla T_{ih}^{k+1}+s_{eh}^{k} \nabla T_{eh}^{k+1},\vartheta^{k+1}_\mathbf{u})|\\
		& \leq|\left(\left(n^{k+1}-n_h^{k+1}\right) \nabla\mu^{k+1}, \vartheta_{\mathbf{u}}^{k+1}\right)+\left(n_h^{k+1}\left(\nabla \mu^{k+1}-\nabla\mu^{k+1}_h\right), \vartheta_{\mathbf{u}}^{k+1}\right)| \\
		&\quad+|\left(\left(s_i^{k+1}-s_i^{k}+s_i^{k}-s_{ih}^{k}\right)\nabla T_i^{k+1}, \vartheta_{\mathbf{u}}^{k+1}\right)+\left(s_{ih}^{k}\left(\nabla T_i^{k+1}-\nabla T_{ih}^{k+1}\right), \vartheta_{\mathbf{u}}^{k+1}\right)| \\
		&\quad+|\left(\left(s_e^{k+1}-s_e^{k}+s_e^{k}-s_{eh}^{k}\right)\nabla T_e^{k+1}, \vartheta_{\mathbf{u}}^{k+1}\right)+\left(s_{eh}^{k}\left(\nabla T_e^{k+1}-\nabla T_{eh}^{k+1}\right), \vartheta_{\mathbf{u}}^{k+1}\right)| \\
		& \leq ((KH)^2+M')(\Delta t)^2\left(\left\|n_t\right\|_{\infty, 0}^2+\left\|T_{it}\right\|_{\infty, 0}^2+\left\|T_{et}\right\|_{\infty, 0}^2\right)\\
		&\quad+C\left(\|\gamma^k_n\|^2+\|\gamma^{k+1}_n\|^2+\|\nabla \gamma^{k+1}_n\|^2+\|\nabla \gamma^{k+1}_{Ti}\|^2+\|\nabla \gamma^{k+1}_{Te}\|^2\right)\\
		&\quad+C\left(\|\vartheta^k_n\|^2+\|\vartheta^{k+1}_n\|^2+\|\vartheta^k_{Ti}\|^2+\|\vartheta^k_{Te}\|^2+\|\vartheta^k_\mathbf{u}\|^2\right)\\
		&\quad+\frac{\eta}{8}\|\nabla \vartheta^{k+1}_\mathbf{u}\|^2+\frac{\kappa_in_0}{8}\|\nabla \vartheta^{k+1}_{Ti}\|^2+\frac{\kappa_en_0}{8}\|\nabla \vartheta^{k+1}_{Te}\|^2.
	\end{align*}
	For the magnetic coupling, we have
	\begin{align*}
		&|\left(J^{k+1}\times\mathbf{B}^{k+1}-J_h^{k+1}\times\mathbf{B}_h^{k+1}, \vartheta^{k+1}_\mathbf{u}\right)|\\
		& =|\left(\left(J^{k+1}-J_h^{k+1}\right)\times \mathbf{B}^{k+1}, \vartheta_{\mathbf{u}}^{k+1}\right)+\left(J^{k+1}_h\left(\mathbf{B}^{k+1}-\mathbf{B}_h^{k+1}\right) , \vartheta_{\mathbf{u}}^{k+1}\right)| \\
		& \leq \left(\left\|\gamma_J^{k+1}\right\|^2+\frac{\mu_0\xi_R}{8}\left\|\vartheta_J^{k+1}\right\|^2+\left\|\gamma_{\mathbf{B}}^{k+1}\right\|^2+\left\|\vartheta_{\mathbf{B}}^{k+1}\right\|^2\right)+(\frac{2M^2}{\mu_0\xi_R}+K^2)\|\vartheta^{k+1}_\mathbf{u}\|^2.
	\end{align*}
	Finally, the gravity contribution satisfies
	\begin{align*}
		|\left(n^{k+1}\mathbf{g}-n^{k+1}_h\mathbf{g}, \vartheta^{k+1}_\mathbf{u}\right) | \leq M\left(\left\|\gamma_n^{k+1}\right\|^2+\left\|\vartheta_n^{k+1}\right\|^2+\left\|\vartheta_{\mathbf{u}}^{k+1}\right\|^2\right).
	\end{align*}
	Collecting all bounds above, we obtain
	\begin{align}\label{eq:error_P2}
		\left|P_2\left(\vartheta_{\mathbf{u}}^{k+1}\right)\right|&\leq  \frac{\eta}{4}\left\|\nabla \vartheta_{\mathbf{u}}^{k+1}\right\|^2+\frac{\kappa_in_0}{8}\left\|\nabla \vartheta_{Ti}^{k+1}\right\|^2+\frac{\kappa_en_0}{8}\left\|\nabla \vartheta_{Te}^{k+1}\right\|^2+\frac{\mu_0\xi_R}{8}\left\| \vartheta_{J}^{k+1}\right\|^2\nonumber\\
		& \quad+C\left(\left\|d_t \gamma_{\mathbf{u}}^{k+1}\right\|^2+\left\|\nabla \gamma_{\mathbf{u}}^{k+1}\right\|^2\right)+\left(n^{k+1}\left(d_t \mathbf{u}^{k+1}-\mathbf{u}_t^{k+1}\right), \vartheta_{\mathbf{u}}^{k+1}\right)\nonumber\\
		&\quad+C (\Delta t)^2\left(\left\|n_t\right\|_{\infty, 0}^2+\left\|\mathbf{u}_t\right\|_{\infty, 0}^2+\left\|\mu s_{\alpha h}^{k+1}\right\|_{\infty}^2+\left\|T_{it}\right\|_{\infty, 0}^2+\left\|T_{et}\right\|_{\infty, 0}^2\right)\nonumber \\
		&\quad +C\left(\left\|\vartheta_{\mathbf{B}}^{k+1}\right\|^2+\left\|\gamma_{\mathbf{B}}^{k+1}\right\|^2+\left\|\gamma_J^{k+1}\right\|^2+\left\|\gamma_n^{k+1}\right\|^2+\left\|\vartheta_n^{k+1}\right\|^2\right)\nonumber \\
		&\quad +C\left(\left\|\gamma_n^{k}\right\|^2+\left\|\gamma_\mathbf{u}^{k}\right\|^2+\left\|\nabla\gamma_{n}^{k+1}\right\|^2+\left\|\nabla\gamma_{Ti}^{k+1}\right\|^2+\left\|\nabla\gamma_{Te}^{k+1}\right\|^2\right) \nonumber\\
		&\quad +C\left(\left\|\vartheta_n^{k}\right\|^2+\left\|\vartheta_{\mathbf{u}}^{k}\right\|^2+\left\|\vartheta_{\mathbf{u}}^{k+1}\right\|^2+\left\|\vartheta_{Ti}^{k}\right\|^2+\left\|\vartheta_{Te}^{k}\right\|^2\right).
	\end{align}
	
	\paragraph{Estimate of $P_3(\vartheta^{k+1}_{Ti})$}
	We start by estimating the terms excluding $R_3(\vartheta^{k+1}_{Ti})$. By the Cauchy–Schwarz and Young inequalities, it follows that
	\begin{align}\label{eq:P3R3_1}
		|(\partial_{Ti}\epsilon^{k+1}_id_t \gamma^{k+1}_{T i},\vartheta^{k+1}_{T i})|&\leq \|d_t \gamma^{k+1}_{Ti}\|^2+H^2\|\vartheta^{k+1}_{Ti}\|^2,\\
		|\kappa_i(n^k_h\nabla\gamma^{k+1}_{T i},\nabla \vartheta^{k+1}_{T i}) | &\leq \frac{2\kappa_iK^2}{n_0}\|\nabla \gamma^{k+1}_{Ti}\|+\frac{\kappa_in_0}{8}\|\nabla\vartheta^{k+1}_{Ti}\|^2.
	\end{align}
	Next, we estimate the residual term $R_3(\vartheta^{k+1}_{Ti})$. By decomposing the internal energy derivative, we have
	\begin{align}\label{eq:epsilon_err}
		&|\left(d_t\widetilde{\epsilon}_i^{k+1}-\partial_t\epsilon_i^{k+1}, \vartheta^{k+1}_{Ti}\right) |\\
		& =|\left(d_t\widetilde{\epsilon}_i^{k+1}-\partial_t\widetilde{\epsilon}_i^{k+1}, \vartheta^{k+1}_{Ti}\right)+\left(\partial_t\widetilde{\epsilon}_i^{k+1}-\partial_t\epsilon_i^{k+1}, \vartheta^{k+1}_{Ti}\right)| \nonumber\\
		& \leq H\left(d_t T_i^{k+1}-\partial_t T_i^{k+1}, \vartheta^{k+1}_{Ti}\right)+H(\partial_t\varepsilon^{k+1}_n, \vartheta^{k+1}_{Ti})\nonumber\\
		& \leq H\left(d_t T_i^{k+1}-\partial_t T_i^{k+1}, \vartheta^{k+1}_{Ti}\right)+H\|\vartheta^{k+1}_{Ti}\|^2+H\|\partial_t\gamma^{k+1}_n\|^2.\nonumber
	\end{align}
	Since the discrete scheme uses $T_i$ as the primary variable while the evolution equation is expressed in terms of the internal energy $\epsilon_i$, the above decomposition separates the time discretization and density effects. For the coupling term involving plasma density, we have
	\begin{align}\label{eq:P3R3_2}
		&|\kappa_i(\nabla T^{k+1}_i(n^k_h-n^{k+1}),\nabla \vartheta^{k+1}_{Ti})|\\
		& \leq \frac{2\kappa_iM^2}{n_0}(\left\|\gamma_n^{k}\right\|^2+\left\|\vartheta_n^{k}\right\|^2)+\frac{\kappa_in_0}{8}\|\nabla\vartheta^{k+1}_{Ti}\|+C(\Delta t)^2\|n_t\|^2_{\infty,0}.\nonumber
	\end{align}
	Applying assumptions \eqref{eq:assume_dis}–\eqref{eq:assume_mus}, we estimate the nonlinear convective terms as follows:
	\begin{align}\label{eq:P3R3_3}
		&|(\mathbf{u}^{k+1}(n^{k+1}\mu_{i}^{k+1}+s_{i}^{k+1}T_{i}^{k+1})-\mathbf{u}^{k}_{\star h}(n_h^{k+1}\mu_{ih}^{k+1}+s_{ih}^kT_{ih}^{k+1}),\nabla \vartheta^{k+1}_{Ti})|\\
		& =|\left(\left(\mathbf{u}^{k+1}-\mathbf{u}^{k}+\mathbf{u}^k-\mathbf{u}^k_{\star h}\right)(n^{k+1}\mu_i^{k+1}+s_i^{k+1}T^{k+1}_i), \nabla\vartheta_{Ti}^{k+1}\right)\nonumber\\
		&\quad+\left(\mathbf{u}^k_{\star h}[\left(n^{k+1}-n_h^{k+1}\right)\mu_i^{k+1}+n^{k+1}_h(\mu_i^{k+1}-\mu^{k+1}_{ih})],\nabla \vartheta_{Ti}^{k+1}\right)\nonumber\\
		&\quad+\left(\mathbf{u}^k_{\star h}[\left(s_i^{k+1}-s_i^{k}+s_i^{k}-s_{ih}^{k}\right)T_i^{k+1}+s^k_{ih}(T_i^{k+1}-T^{k+1}_{ih})],\nabla \vartheta_{Ti}^{k+1}\right)| \nonumber\\
		& \leq\frac{\kappa_in_0}{8}\|\nabla\vartheta^{k+1}_{Ti}\|^2+ C(\Delta t)^2(\left\|n_t\right\|_{\infty, 0}^2+\left\|T_{it}\right\|_{\infty, 0}^2+\left\|\mathbf{u}_{t}\right\|_{\infty, 0}^2+\left\|\mu s_{\alpha h}^{k+1}\right\|_{\infty}^2)\nonumber\\
		&\quad+C(\|\vartheta^{k+1}_n\|^2+\|\gamma^{k+1}_n\|^2+\|\vartheta^{k+1}_{Ti}\|^2+\|\gamma^{k+1}_{Ti}\|^2)\nonumber\\
		&\quad+C(\|\vartheta^k_n\|^2+\|\gamma^k_n\|^2+\|\vartheta^k_{Ti}\|^2+\|\gamma^k_{Ti}\|^2).\nonumber
	\end{align}	
	Similarly, and using Green’s theorem for
	$\left(\mathbf{u}^k_{\star h}\cdot n_h^{k+1}\left(\nabla \mu_i^{k+1}-\nabla\mu^{k+1}_{ih}\right), \vartheta_{Ti}^{k+1}\right)$, we get
	\begin{align}\label{eq:P3R3_4}
		&|(\mathbf{u}^{k+1}\cdot(n^{k+1}\nabla \mu_{i}^{k+1}+s_{i}^{k+1}\nabla T_{i}^{k+1})-\mathbf{u}^{k}_{\star h}\cdot(n_h^{k+1}\nabla \mu_{ih}^{k+1}+s_{ih}^k\nabla T_{ih}^{k+1}),\vartheta^{k+1}_{Ti})|\nonumber\\
		& =|\left(\left(\mathbf{u}^{k+1}-\mathbf{u}^{k}+\mathbf{u}^k-\mathbf{u}^k_{\star h}\right)\cdot(n^{k+1}\nabla\mu_i^{k+1}+s_i^{k+1}\nabla T^{k+1}_i), \vartheta_{Ti}^{k+1}\right)\nonumber\\
		&\quad+|\left(\mathbf{u}^k_{\star h}\cdot[\left(n^{k+1}-n_h^{k+1}\right)\nabla\mu_i^{k+1}+n^{k+1}_h(\nabla\mu_i^{k+1}-\nabla\mu^{k+1}_{ih})], \vartheta_{Ti}^{k+1}\right)| \nonumber\\
		&\quad+|\left(\mathbf{u}^k_{\star h}\cdot[\left(s_i^{k+1}-s_i^{k}+s_i^{k}-s_{ih}^{k}\right)\nabla T_i^{k+1}+s^k_{ih}(\nabla T_i^{k+1}-\nabla T^{k+1}_{ih})],\vartheta_{Ti}^{k+1}\right)| \nonumber\\
		& \leq\frac{\kappa_in_0}{8}\|\nabla\vartheta^{k+1}_{Ti}\|^2+C(\|\vartheta^{k+1}_{Ti}\|^2+\|\vartheta^k_n\|^2+\|\vartheta^{k+1}_n\|^2+\|\vartheta^k_{Ti}\|^2)\nonumber\\
		&\quad+C(\|\nabla \gamma^{k+1}_n\|^2+\|\nabla\gamma^{k+1}_{Ti}\|^2+\|\gamma^k_n\|^2+\|\gamma^{k+1}_n\|^2+\|\gamma^k_{Ti}\|^2)\nonumber\\
		&\quad+ C(\Delta t)^2(\left\|n_t\right\|_{\infty, 0}^2+\left\|T_{it}\right\|_{\infty, 0}^2+\left\|\mathbf{u}_{t}\right\|_{\infty, 0}^2+\left\|\mu s_{\alpha h}^{k+1}\right\|_{\infty}^2).
	\end{align}
	The remaining stress and discrete kinetic energy correction term satisfy
	\begin{align*}
		&| (\pmb{\sigma}^{k+1}_{i} :\nabla \mathbf{u}^{k+1}-\pmb{\sigma}^{k+1}_{i h} :\nabla \mathbf{u}_h^{k+1},\vartheta^{k+1}_{Ti})|\\
		& =| ((\pmb{\sigma}^{k+1}_{i}-\pmb{\sigma}^{k+1}_{i h}) :\nabla \mathbf{u}^{k+1}+\pmb{\sigma}^{k+1}_{i h} :(\nabla \mathbf{u}^{k+1}-\nabla \mathbf{u}_h^{k+1}),\vartheta^{k+1}_{Ti})| \\
		& \leq C\|\vartheta^{k+1}_{Ti}\|+(K^2+M^2)\|\nabla \gamma^{k+1}_\mathbf{u}\|^2+\frac{\eta}{8}\|\nabla \vartheta^{k+1}_\mathbf{u}\|^2,
	\end{align*}
	\begin{align*}
		&| \frac{M_{i}}{2 \Delta t} (n_h^{k} \left( |\mathbf{u}_h^{k+1}-\mathbf{u}^{k}_{\ast h}|^{2} + |\mathbf{u}^{k}_{\ast h}-\mathbf{u}_h^{k}|^{2}\right),\vartheta^{k+1}_{Ti})|\\
		& \leq M_{i}|\frac{\Delta t}{2 } (n^k_h |d_t\mathbf{u}_h^{k+1}|^{2} ,\vartheta^{k+1}_{Ti})|+ M_{i}|\Delta t(n_h^{k} |\mu s_{\alpha h}^{k+1}|^{2} ,\vartheta^{k+1}_{Ti})|\\
		& \leq C\|\vartheta^{k+1}_{Ti}\|^2+\frac{(K\Delta t)^2}{4}\|d_t \mathbf{u}^{k+1}_h\|^4_{L^4}+(K\Delta t)^2\|\mu s^{k+1}_{\alpha h}\|^4_{L^4}.
	\end{align*}
	Collecting all the estimates above, we obtain
	\begin{align}\label{eq:error_P3}
		\left|P_3\left(\vartheta_{Ti}^{k+1}\right)\right| &\leq \frac{\kappa_in_0}{2}\|\nabla\vartheta^{k+1}_{Ti}\|^2+\frac{\eta}{8}\|\nabla \vartheta^{k+1}_\mathbf{u}\|^2+C\|d_t \gamma^{k+1}_{Ti}\|^2\nonumber \\
		&\quad +C (\Delta t)^2\left(\left\|n_t\right\|_{\infty, 0}^2+\left\|T_{it}\right\|_{\infty, 0}^2+\left\|\mathbf{u}_t\right\|_{\infty, 0}^2+\|\mu s_{\alpha h}^{k+1}\|^2_{\infty}\right)\nonumber\\
		&\quad +C\left(\left\|\vartheta_{n}^{k+1}\right\|^2+\left\|\vartheta_{Ti}^{k+1}\right\|^2+\left\|\vartheta_{n}^{k}\right\|^2+\left\|\vartheta_{Ti}^{k}\right\|^2\right)\nonumber\\
		&\quad+C(\|\gamma^k_n\|^2+\|\gamma^{k+1}_{Ti}\|^2+\|\gamma^{k+1}_n\|^2+\|\gamma^k_{Ti}\|^2)\nonumber\\
		&\quad+C(\|\nabla \gamma^{k+1}_n\|^2+\|\nabla\gamma^{k+1}_{Ti}\|^2+\|\nabla \gamma^{k+1}_\mathbf{u}\|^2+\|\partial_t\gamma^{k+1}_n\|^2)\nonumber\\
		&\quad+C(\Delta t)^2(\|d_t \mathbf{u}^{k+1}_h\|^4_{L^4}+\|\mu s^{k+1}_{\alpha h}\|^4_{L^4})\nonumber\\
		&\quad+H\left(d_t T_i^{k+1}-\partial_t T_i^{k+1}, \vartheta^{k+1}_{Ti}\right).
	\end{align}
	
	\paragraph{Estimate of $P_4(\vartheta^{k+1}_{Te})$}
	For $P_4(\vartheta^{k+1}_{Te})$ and $R_4(\vartheta^{k+1}_{Te})$, we only need to estimate the last three items in $R_4$, since the remaining parts share the same structure as $P_3$ and $R_3$ (see Eqs.~\eqref{eq:P3R3_1}–\eqref{eq:P3R3_4}). 
	We first handle the gradient-related term:
	\begin{align*}
		&\Big| \Big(\frac{1}{(n^{k+1}e)^2}|\nabla J^{k+1}|^2-\frac{1}{(n^{k+1}_he)^2}|\nabla J^{k+1}_h|^2,\vartheta^{k+1}_{Te}\Big)\Big| \\
		&\leq \frac{1}{e^2}\Big|\Big(\frac{(\nabla J^{k+1}+\nabla J^{k+1}_h)}{(n^{k+1})^2}\nabla(J^{k+1}-J^{k+1}_h),\vartheta^{k+1}_{Te}\Big)\Big|\\
		&\quad+\frac{1}{e^2}\Big|\Big(\frac{|\nabla J^{k+1}_h|^2(n^{k+1}+n^{k+1}_h)}{(n^{k+1}n^{k+1}_h)^2}(n^{k+1}-n^{k+1}_h),\vartheta^{k+1}_{Te}\Big)\Big|\\
		&\leq \frac{\mu_0}{8(n_1e)^2}\|\nabla \vartheta_J^{k+1}\|^2
		+C(\|\vartheta^{k+1}_{Ti}\|^2+\|\vartheta^{k+1}_{n}\|^2+\|\gamma^{k+1}_{n}\|^2+\|\nabla\gamma^{k+1}_{J}\|^2).
	\end{align*}
	Next, we estimate the nonlinear current term:
	\begin{align*}
		| (|J^{k+1}|^2-|J^{k+1}_h|^2,\vartheta^{k+1}_{Te})|
		&=| ((J^{k+1}-J^{k+1}_h)J^{k+1}+J^{k+1}_h(J^{k+1}-J^{k+1}_h),\vartheta^{k+1}_{Te})| \\
		&\leq C\|\vartheta^{k+1}_{Te}\|^2+C\|\gamma^{k+1}_J\|^2+\frac{\mu_0}{8}\|\vartheta^{k+1}_J\|^2.
	\end{align*}
	Finally, for the discrete magnetic energy correction term:
	\begin{align*}
		\Big|\frac{1}{2\mu_0\Delta t} (|\mathbf{B}_h^{k+1}-\mathbf{B}_h^k|^2,\vartheta^{k+1}_{Te})\Big|
		&=\frac{\Delta t}{2\mu_0}\big|\big(| d_t\mathbf{B}^{k+1}_h|^2,\vartheta^{k+1}_{Te}\big)\big| \\
		&\leq C\|\vartheta^{k+1}_{Te}\|^2+\frac{(\Delta t)^2}{4\mu^2_0}\| d_t\mathbf{B}^{k+1}_h\|^4_{L^4}.
	\end{align*}
	Collecting the above estimates yields
	\begin{align}\label{eq:error_P4}
		\left|P_4\left(\vartheta_{Te}^{k+1}\right)\right|& \leq \frac{\kappa_en_0}{2}\|\nabla\vartheta^{k+1}_{Te}\|^2+\frac{\mu_0\xi_R\xi_e}{8(n_1e)^2}\|\nabla \vartheta_J^{k+1}\|^2+\frac{\mu_0\xi_R}{8}\| \vartheta^{k+1}_J\|^2 \nonumber\\
		& \quad+C (\Delta t)^2\left(\left\|n_t\right\|_{\infty, 0}^2+\left\|T_{et}\right\|_{\infty, 0}^2+\left\|\mathbf{u}_t\right\|_{\infty, 0}^2+\|\mu s_{\alpha h}^{k+1}\|^2_{\infty}\right)\nonumber\\
		&\quad +C\left(\left\|\vartheta_{n}^{k+1}\right\|^2+\left\|\vartheta_{Te}^{k+1}\right\|^2+\left\|\vartheta_{n}^{k}\right\|^2+\left\|\vartheta_{Te}^{k}\right\|^2\right)+C\|d_t \gamma^{k+1}_{Ti}\|^2\nonumber\\
		&\quad+C(\|\gamma^k_n\|^2+\|\gamma^{k+1}_{Te}\|^2+\|\gamma^{k+1}_n\|^2+\|\gamma^k_{Te}\|^2+\|\gamma^{k+1}_J\|^2)\nonumber\\
		&\quad+C(\|\nabla \gamma^{k+1}_n\|^2+\|\nabla\gamma^{k}_{n}\|^2+\|\nabla\gamma^{k+1}_{Te}\|^2)\nonumber\\
		&\quad+C(\|\nabla \gamma^{k+1}_\mathbf{u}\|^2+\|\nabla\gamma^{k+1}_{J}\|^2+\|\partial_t\gamma^{k+1}_n\|^2)\nonumber\\
		&\quad+H\left(d_t T_e^{k+1}-\partial_t T_e^{k+1}, \vartheta^{k+1}_{Te}\right)+C(\Delta t)^2\| d_t\mathbf{B}^{k+1}_h\|^4_{L^4}.
	\end{align}
	
	\paragraph{Estimate of $P_5(\vartheta^{k+1}_{\mathbf{B}},\vartheta^{k+1}_{J})$}
	We next turn to the coupled electromagnetic terms.  
	For the part excluding $R_5(\vartheta^{k+1}_\mathbf{B},\vartheta^{k+1}_J)$,
	\begin{align*}
		|(d_t\gamma^{k+1}_\mathbf{B},\vartheta^{k+1}_\mathbf{B})|&\leq\|d_t \gamma^{k+1}_\mathbf{B}\|^2+\|\vartheta^{k+1}_\mathbf{B}\|^2,\\
		|(\mathbf{u}^{k+1}_h \times \gamma^{k+1}_{\mathbf{B}},\vartheta^{k+1}_J) | &\leq\frac{4K^2}{\mu_0\xi_R} \| \gamma^{k+1}_\mathbf{B}\|+\frac{\mu_0\xi_R}{16}\|\vartheta^{k+1}_{J}\|^2, \\
		|( \gamma^{k+1}_{J},\vartheta^{k+1}_J) | &\leq2 \| \gamma^{k+1}_J\|+\frac{1}{8}\|\vartheta^{k+1}_{J}\|^2, \\
		\Big|\Big(\frac{1}{(n^{k+1}_h)^2}\nabla \gamma^{k+1}_J,\nabla \vartheta^{k+1}_J\Big) \Big| &\leq\frac{2n^2_1 }{n^2_0}\| \nabla\gamma^{k+1}_J\|+\frac{1}{8n^2_1}\|\nabla \vartheta_J^{k+1}\|^2, \\
		|(\gamma^{k+1}_E,\vartheta^{k+1}_J)|&\leq\frac{4}{\mu_0\xi_R}\| \gamma^{k+1}_E\|^2+\frac{\mu_0\xi_R}{16}\|\vartheta^{k+1}_J\|^2,\\
		|(\nabla \times \gamma^{k+1}_E,\vartheta^{k+1}_\mathbf{B}) | &\leq \|\nabla\times \gamma^{k+1}_E\|+\|\vartheta^{k+1}_{\mathbf{B}}\|^2. \\
	\end{align*}
	For the residual terms $R_5(\vartheta^{k+1}_\mathbf{B},\vartheta^{k+1}_J)$, we similarly obtain	
	\begin{align*}
		|\left(\mathbf{B}^{k+1} \times\left(\mathbf{u}_h^{k+1}-\mathbf{u}^{k+1}\right),\vartheta^{k+1}_J\right)| \leq \frac{2M^2}{\mu_0\xi_R}(\|\gamma^{k+1}_\mathbf{u}\|^2+\|\vartheta^{k+1}_\mathbf{u}\|^2)+\frac{\mu_0\xi_R}{8}\|\vartheta^{k+1}_J\|^2,
	\end{align*}
	\begin{align*}
		&\Big|\Big(\frac{1}{(n^{k+1})^2}\nabla J^{k+1}-\frac{1}{(n_h^{k+1})^2}\nabla J^{k+1},\nabla \vartheta^{k+1}_J\Big)\Big|\\
		& =\Big| \Big(\nabla J^{k+1}\frac{(n^{k+1}+n^k_h)(n^{k+1}-n^{k+1}_h)}{(n^{k+1}n^{k+1}_h)^2},\nabla \vartheta^{k+1}_J\Big)\Big| \\
		& \leq2n^2_1\Big(\frac{M^2+K^2}{(\underline{n}n_0)^2}\Big)^2\left(\|\vartheta^{k+1}_n\|^2+\|\gamma^{k+1}_n\|^2\right)+\frac{1}{8n^2_1}\|\nabla \vartheta^{k+1}_J\|^2.
	\end{align*}
	Collecting all the estimates leads to
	\begin{align}\label{eq:error_P5}
		\left|P_5\right| &\leq \frac{3\mu_0\xi_R}{8}\left\|\vartheta_J^{k+1}\right\|^2+\frac{\mu_0\xi_e\xi_R}{4(n_1e)^2}\|\nabla \vartheta^{k+1}_J\|^2 +\left(d_t \mathbf{B}^{k+1}-\mathbf{B}_t^{k+1}, \vartheta_{\mathbf{B}}^{k+1}\right)\nonumber\\
		&\quad+C(\|\gamma^{k+1}_\mathbf{u}\|^2+\|\vartheta^{k+1}_\mathbf{u}\|^2+\| \nabla \gamma^{k+1}_J\|^2+\|\vartheta^{k+1}_n\|^2+\|\gamma^{k+1}_n\|^2+\|\nabla\gamma^{k+1}_E\|^2) \nonumber\\
		&\quad+ \left\|d_t \gamma_{\mathbf{B}}^{k+1}\right\|^2+C(\|\gamma^{k+1}_E\|^2+\|\gamma^{k+1}_\mathbf{B}\|^2+\|\gamma^{k+1}_J\|^2).
	\end{align}
	
	We choose the initial approximations as
	$$
	n_h^0=\mathcal{I}^g_h n^0, \mathbf{u}_h^0=\mathcal{I}^g_h\mathbf{u}^0, \mathbf{B}_h^0=\mathcal{R}_h \mathbf{B}^0, T_{ih}^0=\mathcal{I}^g_h T_i^0,T_{eh}^0=\mathcal{I}^g_h T_e^0
	,$$
	which guarantees
	$\vartheta^0_n=0,\vartheta^0_\mathbf{u}=0,\vartheta^0_\mathbf{B}=0,\vartheta^0_{Ti}=0,\vartheta^0_{Te}=0.$
	Substituting \eqref{eq:error_Q}-\eqref{eq:error_P5} into \eqref{eq:error_eq_QP}, we obtain
	\begin{align}\label{eq:error_eq_QP2}
		&\frac{1}{2}\left\|\vartheta_n^{j}\right\|^2+\frac{M_in_0}{2 }\left\|\vartheta_{\mathbf{u}}^{j}\right\|^2+\frac{1}{2}\left\|\vartheta_{\mathbf{B}}^{j}\right\|^2+\frac{H_0}{2 }\left\|\vartheta_{T i}^{j}\right\|^2+\frac{H_0}{2}\left\|\vartheta_{T e}^{j}\right\|^2\nonumber \\
		&\quad+\Delta t\sum_{k=0}^{j-1}\left(\frac{\mu_0\xi_R}{4} \left\|\vartheta_{J}^{k+1}\right\|^2 +\frac{5\mu_0\xi_R\xi_e}{8(n_1e)^2}\|\nabla \vartheta^{k+1}_J\|^2+\frac{3\eta}{8}\left\|\nabla \vartheta_\mathbf{u}^{k+1}\right\|^2\right.\nonumber\\
		&\left. \quad+(\lambda+\eta)\left\|\nabla\cdot \vartheta_{\mathbf{u}}^{k+1}\right\|^2+\frac{3\kappa_in_0}{8}\left\|\nabla \vartheta_{Ti}^{k+1}\right\|^2+\frac{3\kappa_en_0}{8}\left\|\nabla \vartheta_{Te}^{k+1}\right\|^2\right)\nonumber\\
		& \leq C\Delta t \sum_{k=0}^{j-1}\left(\left\| \vartheta_n^{k+1}\right\|^2+\left\| \vartheta_{\mathbf{u}}^{k+1}\right\|^2+\left\| \vartheta_{\mathbf{B}}^{k+1}\right\|^2+\left\| \vartheta_{Ti}^{k+1}\right\|^2 +\left\| \vartheta_{Te}^{k+1}\right\|^2\right) \nonumber\\
		& \quad+C T (\Delta t)^2\left(\left\|n_t\right\|_{\infty, 0}^2+\left\|T_{it}\right\|_{\infty, 0}^2+\left\|T_{et}\right\|_{\infty, 0}^2+\left\|\mathbf{u}_t\right\|_{\infty, 0}^2+\left\|\nabla \mathbf{u}_t\right\|_{\infty, 0}^2\right) \nonumber\\
		&\quad+C(\Delta t)^2(\triplenorm{\mu s_{\alpha h}}^2_{\infty,0}+\triplenorm{d_t \mathbf{u}_h}^4_{4,L^4}+\triplenorm{d_t \mathbf{B}_h}^4_{4,L^4}+\triplenorm{\mu s_{\alpha h}}^4_{4,L^4})\nonumber\\
		& \quad+C\Delta t\sum_{k=0}^{j-1}\left(\mathcal{G}^{k+1}_1+\mathcal{G}^{k+1}_2+\mathcal{G}^{k+1}_3+\mathcal{G}^{k+1}_4\right).
	\end{align}
	Here, the triple norm over discrete time levels is defined as
	$\triplenorm{d_t \mathbf{u}_h}_{4,L^4}:= (\sum_{k=0}^{j-1} \Delta t\\ \|d_t \mathbf{u}^{k+1}_h\|^4_{L^4})^{\frac{1}{4}}$
	and analogously for $\triplenorm{d_t \mathbf{B}_h}_{4,L^4},\triplenorm{\mu s_{\alpha h}}_{4,L^4}$. Moreover, we set
	\begin{align}
		\mathcal{G}_1^{k+1}:=&\left\| \gamma_n^{k+1}\right\|^2+\left\| \gamma_{\mathbf{u}}^{k+1}\right\|^2+\left\| \gamma_{\mathbf{B}}^{k+1}\right\|^2+\left\| \gamma_{E}^{k+1}\right\|^2+\left\| \gamma_{Ti}^{k+1}\right\|^2 +\left\| \gamma_{Te}^{k+1}\right\|^2+\left\| \gamma_{J}^{k+1}\right\|^2,\nonumber\\
		\mathcal{G}_2^{k+1}:=&\left\|\nabla \gamma_n^{k+1}\right\|^2+\left\|\nabla \gamma_{\mathbf{u}}^{k+1}\right\|^2+\left\|\nabla \gamma_{Ti}^{k+1}\right\|^2 +\left\|\nabla \gamma_{Te}^{k+1}\right\|^2\nonumber\\
		&+\left\|\nabla \gamma_{J}^{k+1}\right\|^2+\left\|\nabla \gamma_{E}^{k+1}\right\|^2+\left\|\partial_t \gamma_n^{k+1}\right\|^2,\nonumber\\
		\mathcal{G}_3^{k+1}:=&\left\|d_t \gamma_n^{k+1}\right\|^2+\left\|d_t \gamma_{\mathbf{u}}^{k+1}\right\|^2+\left\|d_t \gamma_{\mathbf{B}}^{k+1}\right\|^2+\left\|d_t \gamma_{Ti}^{k+1}\right\|^2+\left\|d_t \gamma_{Te}^{k+1}\right\|^2,\nonumber\\
		\mathcal{G}_4^{k+1}:=&\left(d_t n^{k+1}-n_t^{k+1}, \vartheta^{k+1}_n\right)+\left(n^{k+1}\left(d_t \mathbf{u}^{k+1}-\mathbf{u}_t^{k+1}\right),\vartheta^{k+1}_\mathbf{u}\right),\nonumber\\
		&+\sum_{\alpha=i,e}H\left(d_tT_\alpha^{k+1}-\partial_tT_\alpha^{k+1}, \vartheta^{k+1}_{T\alpha}\right)+\left(d_t \mathbf{B}^{k+1}-\mathbf{B}_t^{k+1}, \vartheta_{\mathbf{B}}^{k+1}\right) .
	\end{align}
	From the interpolation property \eqref{eq:gamma}, we can estimate $\Delta t\sum_{k=0}^{j-1}\mathcal{G}^{k+1}_1$ and $\Delta t\sum_{k=0}^{j-1}\\ \mathcal{G}^{k+1}_2$ as 
	\begin{align}\label{eq:G_1}
		\Delta t \sum_{k=0}^{j-1}\mathcal{G}^{k+1}_1 
		&\leq C\left(h^{2 m_1+2}\big(\triplenorm{n}^2_{0,m_1+1}+\triplenorm{T_i}^2_{0,m_1+1}+\triplenorm{T_e}^2_{0,m_1+1}\big)\right. \nonumber\\
		& \left.\quad+h^{2 m_3+2}\big(\triplenorm{E}^2_{0,m_3+1}+\triplenorm{J}^2_{0,m_3+1}\big)\right.\nonumber\\
		&\left.\quad+h^{2 m_2+2}\triplenorm{\mathbf{u}}^2_{0,m_2+1}+h^{2 m_4+2}\triplenorm{\mathbf{B}}^2_{0,m_4+1}\right),\\
		\Delta t \sum_{k=0}^{j-1}\mathcal{G}^{k+1}_2 
		&\leq C\left(h^{2 m_1}\big(\triplenorm{n}^2_{0,m_1+1}+\triplenorm{n_t}^2_{0,m_1}+\triplenorm{T_i}^2_{0,m_1+1}+\triplenorm{T_e}^2_{0,m_1+1}\big)\right. \nonumber\\
		& \left.\quad+h^{2 m_2}\triplenorm{\mathbf{u}}^2_{0,m_2+1}+h^{2 m_3}\big(\triplenorm{J}^2_{0,m_3+1}+\triplenorm{E}^2_{0,m_3+1}\big)\right).
	\end{align}
	Next, we estimate the term $\Delta t\sum_{k=0}^{j-1}\mathcal{G}^{k+1}_3$. Using the definition of $d_t \gamma_n^{k+1}$ and the approximation properties of the projection operator, we have
	\begin{align*}
		\sum_{k=0}^{j-1}\Delta t\left\|d_t \gamma_n^{k+1}\right\|^2 =&\sum_{k=0}^{j-1} \left\|\frac{1}{\Delta t} \int_{t_{k}}^{t_{k+1}} \frac{\partial \gamma_n}{\partial t} \mathrm{~d} t\right\|^2 \Delta t \nonumber\\
		& \leq \sum_{k=0}^{j-1} \frac{1}{(\Delta t)^2} \int_{t_{k}}^{t_{k+1}}\left\|\frac{\partial \gamma_n}{\partial t}\right\|^2 \cdot  \Delta t \cdot \Delta t \nonumber\\
		& \leq C h^{2 m_1}\left\|n_t\right\|_{0, m_1}^2.
	\end{align*}
	Analogous estimates hold for $\mathbf{u}, T_i$, and $T_e$:
	\begin{align*}
		\sum_{k=0}^{j-1} \Delta t\left\|d_t \gamma_\psi^{k+1}\right\|^2 \leq C h^{2 m_\psi}\left\|\psi_t\right\|_{0, m_\psi}^2, \quad \psi \in\left\{\mathbf{u}, T_i, T_e\right\}.
	\end{align*}
	Since the magnetic field evolution does not contain higher-order spatial derivatives, to retain optimal convergence we assume $\mathbf{B}_t\in L^2(0,T;H^{m_4+1}(\Omega))$ yielding the improved estimate \cite{Feng2010}:
	$$ \sum_{k=0}^{j-1} \Delta t\left\|d_t \gamma_{\mathbf{B}}^{k+1}\right\|^2 \leq C h^{2 m_4+2}\left\|\mathbf{B}_t\right\|_{0, m_4+1}^2.$$
	Therefore,
	\begin{align}\label{eq:G_3}
		\Delta t\sum_{k=0}^{j-1}\mathcal{G}^{k+1}_3\leq& Ch^{2m_1}\left(\|n_t\|^2_{0,m_1}+\|T_{it}\|^2_{0,m_1}+\|T_{et}\|^2_{0,m_1}\right)\nonumber\\
		&+Ch^{2m_2}\|\mathbf{u}_t\|^2_{0,m_2}+Ch^{2m_4+2}\|\mathbf{B}_t\|^2_{0,m_4}.
	\end{align}
	Next, we estimate $\Delta t\sum_{k=0}^{j-1}\mathcal{G}^{k+1}_4$. Using the integral form of Taylor’s theorem with remainder, $f(b)=f(a)+$ $f^{\prime}(a)(b-a)+\frac{1}{2} \int_a^b f^{\prime \prime}(x)(b-x) \mathrm{d} x$, we have
	$$
	d_t n^{k+1}-n_t^{k+1}=\frac{1}{2 \Delta t} \int_{t_{k}}^{t_{k+1}} n_{t t}(t, \cdot)\left(t_{k}-t\right) \mathrm{d} t.
	$$
	By the Cauchy–Schwarz inequality,
	\begin{align*}
		{\left[\int_{t_{k}}^{t_{k+1}} n_{t t}(t, \cdot)\left(t_{k}-t\right) \mathrm{d} t\right]^2 } & \leq \int_{t_{k}}^{t_{k+1}} n_{t t}^2(t, \cdot) \mathrm{d} t \int_{t_{k}}^{t_{k+1}}\left(t_{k}-t\right)^2 \mathrm{~d} t \\
		& =\frac{1}{3} (\Delta t)^3 \int_{t_{k}}^{t_{k+1}} n_{t t}^2(t, \cdot) \mathrm{d} t,
	\end{align*}
	and hence
	\begin{align*}
		\sum_{k=0}^{j-1} \Delta  t\left(d_t n^{k+1}-n_t^{k+1}, \vartheta_n^{k+1}\right) & \leq \sum_{k=0}^{j-1} \frac{1}{\sqrt{3}} (\Delta t)^{\frac{3}{2}}\left[\int_{\Omega} \int_{t_{k}}^{t_{k+1}} n_{t t}^2(t, \cdot) \mathrm{d} t \mathrm{~d} \boldsymbol{x}\right]^{\frac{1}{2}}\left\|\vartheta_n^{k+1}\right\|_0 \nonumber\\
		& \leq \sum_{k=0}^{j-1} \Delta t\left\|\vartheta_n^{k+1}\right\|^2+\frac{1}{3} (\Delta t)^2\left\|n_{t t}\right\|_{0,0}^2.
	\end{align*}
	Applying the same argument to $\mathbf{u},T_i,T_e$ and $\mathbf{B}$, and using assumption \eqref{eq:assume_dis}, we get
	\begin{align}\label{eq:G_4}
		\Delta t \sum_{k=0}^{j-1}\mathcal{G}^{k+1}_4 
		&\leq \sum_{k=0}^{j-1} \Delta t\bigg(\left\|\vartheta_n^{k+1}\right\|^2+2\overline{n}\left\|\vartheta_\mathbf{u}^{k+1}\right\|^2+\sum_{\alpha=i,e}H\left\|\vartheta_{T\alpha}^{k+1}\right\|^2+\left\|\vartheta_\mathbf{B}^{k+1}\right\|^2\bigg)\nonumber\\
		&\quad+\frac{1}{3} (\Delta t)^2\bigg(\left\|n_{t t}\right\|_{0,0}^2+\left\|\mathbf{u}_{t t}\right\|_{0,0}^2++\sum_{\alpha=i,e}H\left\|T_{\alpha t t}\right\|_{0,0}^2+\left\|\mathbf{B}_{t t}\right\|_{0,0}^2\bigg).
	\end{align}

	Substituting \eqref{eq:G_1}-\eqref{eq:G_4} into \eqref{eq:error_eq_QP2}, let $\alpha = \min\!\left(\tfrac{1}{2}, \tfrac{M_i n_0}{2}, \tfrac{H_0}{2}, \tfrac{\mu_0 \xi_R}{4}\right).$ Divide both sides by $\alpha$ for normalization and absorb $\alpha$ into the generic constant $C$. 
	Assuming the mesh size $h$ is sufficiently small such that $C h^2 \le 1$, we obtain
	\begin{align}\label{eq:error_eq_QP4}
		&\left\|\vartheta_n^{j}\right\|^2+\left\|\vartheta_{\mathbf{u}}^{j}\right\|^2+\left\|\vartheta_{\mathbf{B}}^{j}\right\|^2+\left\|\vartheta_{T i}^{j}\right\|^2+\left\|\vartheta_{T e}^{j}\right\|^2+\triplenorm{\vartheta_J}^2_{0,0}\nonumber \\
		& \quad+\Delta t\sum_{k=0}^{j-1}\left(\frac{5\mu_0\xi_R\xi_e}{8(n_1e)^2}\|\nabla \vartheta^{k+1}_J\|^2+\frac{3\eta}{8}\left\|\nabla \vartheta_\mathbf{u}^{k+1}\right\|^2+(\lambda+\eta)\left\|\nabla\cdot \vartheta_{\mathbf{u}}^{k+1}\right\|^2\right. \nonumber\\
		& \left.\quad+\frac{3\kappa_in_0}{8}\left\|\nabla \vartheta_{Ti}^{k+1}\right\|^2+\frac{3\kappa_en_0}{8}\left\|\nabla \vartheta_{Te}^{k+1}\right\|^2\right)\nonumber\\
		& \leq C\Delta t \sum_{k=0}^{j-1}\left(\left\| \vartheta_n^{k+1}\right\|^2+\left\| \vartheta_{\mathbf{u}}^{k+1}\right\|^2+\left\| \vartheta_{\mathbf{B}}^{k+1}\right\|^2+\left\| \vartheta_{Ti}^{k+1}\right\|^2 +\left\| \vartheta_{Te}^{k+1}\right\|^2\right) +\Delta,
	\end{align}
	where $C$ is a constant independent of $j,\Delta t$ and $h$, and $\Delta$ is defined in \eqref{eq:error_Delta}. Applying Gronwall’s lemma yields, for all $1\le j\le N$,
	\begin{align}\label{eq:error_eq_QP5}
		\left\|\vartheta_n^{j}\right\|^2+\left\|\vartheta_{\mathbf{u}}^{j}\right\|^2+\left\|\vartheta_{\mathbf{B}}^{j}\right\|^2+\left\|\vartheta_{T i}^{j}\right\|^2+\left\|\vartheta_{T e}^{j}\right\|^2+\triplenorm{\vartheta_J}^2_{0,0} \leq  \Delta.
	\end{align}
	Using \eqref{eq:error_eq_QP5} and the interpolation estimates \eqref{eq:gamma}, 
	we immediately obtain the total error estimates for all variables except $E$:
	\begin{align}\label{eq:error_all1}
		&\triplenorm{\varepsilon_n}^2_{\infty,0}+\triplenorm{\varepsilon_\mathbf{u}}^2_{\infty,0}+\triplenorm{\varepsilon_{Ti}}^2_{\infty,0}+\triplenorm{\varepsilon_{Te}}^2_{\infty,0}+\triplenorm{\varepsilon_\mathbf{B}}^2_{\infty,0}+\triplenorm{\varepsilon_J}^2_{0,0}\nonumber\\
		&\leq CT\Delta+C(h^{2m_1}+h^{2m_2}+h^{2m_3}+h^{2m_4+2}).
	\end{align}
	Moreover, from \eqref{eq:error_eq_QP4}–\eqref{eq:error_eq_QP5}, we have
	$$
	\triplenorm{\nabla \vartheta_{\mathbf{u}}}_{0,0}^2+\triplenorm{\nabla \vartheta_{J}}_{0,0}^2 +\triplenorm{\nabla \vartheta_{T_i}}_{0,0}^2+ \triplenorm{\nabla \vartheta_{T_e}}_{0,0}^2 \leq C T \Delta,
	$$
	and
	$$
	\triplenorm{ \vartheta_{n}}_{0,0}^2+ \triplenorm{ \vartheta_{\mathbf{u}}}_{0,0}^2 +\triplenorm{ \vartheta_{T_i}}_{0,0}^2 +\triplenorm{\vartheta_{T_e}}_{0,0}^2+\triplenorm{\vartheta_{\mathbf{B}}}_{0,0}^2+\triplenorm{\vartheta_J}_{0,0}^2\leq C T \Delta.
	$$
	Thus
	\begin{align}\label{eq:error_noE}
		\triplenorm{\varepsilon_n}^2_{0,0}+\triplenorm{\varepsilon_\mathbf{u}}^2_{0,1}+\triplenorm{\varepsilon_{Ti}}^2_{0,1}+\triplenorm{\varepsilon_{Te}}^2_{0,1}+\triplenorm{\varepsilon_\mathbf{B}}^2_{0,0}+\triplenorm{\varepsilon_J}^2_{0,1}\leq CT\Delta.
	\end{align}
	
	\paragraph{Step2 (Error estimate of $E^{k+1}_h$)} To estimate $\triplenorm{\varepsilon_E}_{0,1}^2$. we apply the discrete inf–sup condition \eqref{eq:infsup}, 
	together with the residual estimate $|R_5(\mathbf{C},G)|\le \sqrt{\Delta}\|(\mathbf{C},G)\|_{\mathbf{X}}$.  
	It then follows from \eqref{eq:totalerror_EBJ} that
	\begin{align}\label{eq:error_E}
		\|\vartheta_E^{k+1}\|_1\leq C(\|\gamma_E^{k+1}\|_1+\|d_t\varepsilon^{k+1}_\mathbf{B}\|+\|\varepsilon^{k+1}_\mathbf{B}\|+\|\varepsilon^{k+1}_J\|_1)+\sqrt{\Delta}.
	\end{align}
	Squaring \eqref{eq:error_E}, multiplying by $\Delta t$, and summing over $k=0,\dots,j-1$, we apply \eqref{eq:error_noE} to get
	\begin{align}\label{eq:error_E2}
		\triplenorm{\vartheta_E}_{0,1}^2\leq C\sum_{k=0}^{j-1}\Delta t\|d_t\vartheta^{k+1}_\mathbf{B}\|^2+CT\Delta.
	\end{align}
	
	We now estimate $\sum_{k=0}^{j-1}\Delta t\|d_t\vartheta^{k+1}_\mathbf{B}\|^2$. From Maxwell's mixed projection \eqref{eq:proj:a}-\eqref{eq:proj:b}, \eqref{eq:totalerror_EBJ} can be rewritten as
	\begin{align}\label{eq:projerror_EBJ}
		&\left(d_t \varepsilon_{\mathbf{B}}^{k+1}, \mathbf{C}\right)-\mu_0\left(\mathbf{u}_h^{k+1} \times \varepsilon_\mathbf{B}^{k+1},G\right)+\mu_0\xi_R(\vartheta_J^{k+1},G)+\Big(\frac{\mu_0\xi_R\xi_e}{(n^{k+1}_he)^2}\nabla \vartheta^{k+1}_J,\nabla G\Big)\nonumber\\
		&\quad-(\gamma^{k+1}_\mathbf{B},\mathbf{C})+\mathbf{b}((\mathbf{C},G),\vartheta^{k+1}_E)-\mathbf{b}((\vartheta_\mathbf{B}^{k+1},\vartheta_J^{k+1}),F)=R_5(\mathbf{C},G).
	\end{align}
	For any $F_h\in V_h,$ taking $\mathbf{C}=0,G=0$ in \eqref{eq:projerror_EBJ} yields $\mathbf{b}((\vartheta^{k+1}_\mathbf{B},\vartheta^{k+1}_J),F_h)=0$ for $0\leq k\leq j-1$. Moreover, by \eqref{eq:proj:b}, $\mathbf{b}((\mathbf{B}^0,J^0),F_h)=0$, and $(\mathbf{B}^0_h,J^0_h)\in \mathbf{X}^{\boldsymbol{\xi},0}_h$, we have $\mathbf{b}((\vartheta^0_\mathbf{B},\vartheta^0_J),F_h)=0$. Since $\vartheta_E^{k+1}\in V_h$, it follows that
	\begin{align}\label{eq:dtbE}
		\boldsymbol{b}((d_t\vartheta^{k+1}_\mathbf{B},d_t\vartheta^{k+1}_J),\vartheta^{k+1}_E)=0,\quad \forall  0\leq k\leq j-1.
	\end{align}
	Next, take $(\mathbf{C},G)=(2d_t\vartheta^{k+1}_\mathbf{B}\Delta t,2d_t\vartheta^{k+1}_J\Delta t)$ and $F=0$ in \eqref{eq:projerror_EBJ}. Using \eqref{eq:dtbE} and the identity $2a(a-b)=|a|^2-|b|^2+|a-b|^2$, and substituting the definition of  $R_5(\mathbf{C},G)$ from \eqref{eq:exact_eq_R}, we obtain
	\begin{align}\label{eq:projerror_EBJ2}
		&2\Delta t\|d_t\vartheta^{k+1}_\mathbf{B}\|^2+d_t\Big(\mu_0\xi_R\|\vartheta^{k+1}_J\|^2+\frac{\mu_0\xi_R\xi_e}{(n_1e)^2}\|\nabla \vartheta^{k+1}_J\|^2\Big)\Delta t\nonumber\\
		&\quad-\mu_0(\mathbf{u}^{k+1}\times \varepsilon^{k+1}_\mathbf{B},2d_t\vartheta^{k+1}_J\Delta t)-2\Delta t(\gamma^{k+1}_\mathbf{B},d_t\vartheta^{k+1}_\mathbf{B})+2\Delta t(d_t\gamma^{k+1}_\mathbf{B},d_t\vartheta^{k+1}_\mathbf{B})\nonumber\\
		&\leq \left(d_t \mathbf{B}^{k+1}-\mathbf{B}_t^{k+1}, \mathbf{C}\right)+\mu_0 \left(\mathbf{B}^{k+1} \times\left(\mathbf{u}^{k+1}-\mathbf{u}_h^{k+1}\right), G\right)\nonumber\\
		&\quad+\mu_0\xi_e\xi_R\Big(\frac{1}{(n_h^{k+1})^2}\nabla J^{k+1}-\frac{1}{(n^{k+1})^2}\nabla J^{k+1},\nabla G\Big).
	\end{align}
	Summing \eqref{eq:projerror_EBJ2} over $k=0,\dots,j-1$, and applying Young’s inequalities gives
	\begin{align*}
		|2(\gamma^{k+1}_\mathbf{B},d_t\vartheta^{k+1}_\mathbf{B})
		|&\leq \frac{1}{4}\|d_t \vartheta^{k+1}_{\mathbf{B}}\|^2+C\|\gamma^{k+1}_\mathbf{B}\|^2,\\
		|(d_t\gamma^{k+1}_\mathbf{B},d_t\vartheta^{k+1}_\mathbf{B})| &\leq \frac{1}{4}\|d_t \vartheta^{k+1}_{\mathbf{B}}\|^2+C\|d_t\gamma^{k+1}_\mathbf{B}\|^2, \\
		\sum_{k=0}^{j-1}|\left(d_t \mathbf{B}^{k+1}-\mathbf{B}_t^{k+1},d_t\vartheta^{k+1}_{\mathbf{B}}\right)| &\leq\sum_{k=0}^{j-1} \frac{1}{4}\Delta t\|d_t \vartheta^{k+1}_{\mathbf{B}}\|^2+C(\Delta t)^2\|\mathbf{B}_{tt}\|_{0,0}^2. \\
	\end{align*}
	Define  $$A^{k+1}_1:=\mathbf{u}^{k+1}\times \varepsilon^{k+1}_\mathbf{B}+\mathbf{B}^{k+1}\times \varepsilon^{k+1}_\mathbf{u}, \quad A^{k+1}_2:=\frac{1}{(n_h^{k+1})^2}\nabla J^{k+1}-\frac{1}{(n^{k+1})^2}\nabla J^{k+1}.$$ 
	Then we arrive at
	\begin{align}\label{eq:projerror_EBJ3}
		&\frac{5}{4}\sum_{k=0}^{j-1}\Delta t\|d_t\vartheta^{k+1}_\mathbf{B}\|^2+\mu_0\xi_R\|\vartheta^{j}_J\|^2+\frac{\mu_0\xi_R\xi_e}{(n_1e)^2}\|\nabla \vartheta^{j}_J\|^2\nonumber\\
		&\leq \Delta +\mu_0\sum_{k=0}^{j-1}(A_1^{k+1},\vartheta^{k+1}_J-\vartheta^k_J)+\frac{\mu_0\xi_e\xi_R}{e^2}\sum_{k=0}^{j-1}(A_2^{k+1},\nabla\vartheta^{k+1}_J-\nabla\vartheta^k_J).
	\end{align}
	To estimate the last two terms, we use the discrete summation identity $\sum_{k=0}^{j-1}(a^{k+1},b^{k+1}-b^k)=-\sum_{k=0}^{j-1}(a^{k+1}-a^k,b^k)+(a^j,b^j)-(a^0,b^0)$, and the initial condition $\vartheta^{0}_J=0$, which yields
	\begin{align} \label{eq:Adis1}
		\sum_{k=0}^{j-1}(A_1^{k+1},\vartheta^{k+1}_J-\vartheta^k_J)\leq& \sum_{k=0}^{j-1}\Delta t|(d_tA^{k+1}_1,\vartheta^{k}_J)|+|(A^{j}_1,\vartheta^j_J)|,\nonumber\\
		\sum_{k=0}^{j-1}(A_2^{k+1},\nabla\vartheta^{k+1}_J-\nabla\vartheta^k_J)\leq& \sum_{k=0}^{j-1}\Delta t|(d_tA^{k+1}_2,\nabla\vartheta^{k}_J)|+|(A^{j}_2,\nabla\vartheta^j_J)|.
	\end{align}
	Applying Young’s inequality, we have
	\begin{align} \label{eq:Adis2}
		|(A^{j}_1,\vartheta^j_J)|\leq& \frac{\xi_R}{2}\|\vartheta^j_J\|^2+C(\|\varepsilon^j_{\mathbf{B}}\|^2+\|\varepsilon^j_{\mathbf{u}}\|^2),\nonumber\\
		|(A^{j}_2,\nabla\vartheta^j_J)|\leq& \frac{1}{2n^2_1}\|\nabla \vartheta^j_J\|^2+C\|\varepsilon^j_n\|^2,\nonumber\\
		|(d_tA^{k+1}_1,\vartheta^{k}_J)|\leq& \frac{1}{4\mu_0}\|d_t\vartheta^{k+1}_\mathbf{B}\|^2+\|d_t\varepsilon^{k+1}_{\mathbf{u}}\|^2\nonumber\\
		&+C(\|d_t\gamma^{k+1}_\mathbf{B}\|^2+\|\varepsilon^{k+1}_{\mathbf{B}}\|^2+\|\varepsilon^{k+1}_{\mathbf{u}}\|^2+\|\vartheta^{k+1}_J\|^2),\nonumber\\
		\big|(d_t A^{k+1}_2,\nabla\vartheta_J^k)\big|
		\leq& \|d_t\varepsilon_n^{k+1}\|^2
		+ C(\|\varepsilon_n^{k+1}\|^2 + \|\nabla\vartheta_J^k\|^2).
	\end{align}
	Substituting \eqref{eq:Adis2} into \eqref{eq:Adis1}, and then back into \eqref{eq:projerror_EBJ3}, we obtain
	\begin{align}\label{eq:projerror_EBJ43}
		&\sum_{k=0}^{j-1}\Delta t\|d_t\vartheta^{k+1}_\mathbf{B}\|^2+\frac{\mu_0\xi_R}{2}\|\vartheta^{j}_J\|^2+\frac{\mu_0\xi_R\xi_e}{2(n_1e)^2}\|\nabla \vartheta^{j}_J\|^2\nonumber\\
		&\leq \Delta +C(\triplenorm{d_t\varepsilon_\mathbf{u}}^2_{0,0}+\triplenorm{d_t\varepsilon_n}^2_{0,0})+C(\|\varepsilon^{j}_{\mathbf{B}}\|^2+\|\varepsilon^{j}_{\mathbf{u}}\|^2+\|\varepsilon^{j}_n\|^2)\nonumber\\
		&\quad+C(\triplenorm{\varepsilon_\mathbf{B}}^2_{0,0}+\triplenorm{\varepsilon_\mathbf{u}}^2_{0,0}+\triplenorm{\varepsilon_n}^2_{0,0}+\triplenorm{\varepsilon_J}^2_{0,1}).
	\end{align}
	The estimates for $\triplenorm{d_t\varepsilon_\mathbf{u}}^2_{0,0}$ and $\triplenorm{d_t\varepsilon_n}^2_{0,0}$
	follow analogously from their respective error equations, namely $(\triplenorm{d_t\varepsilon_\mathbf{u}}^2_{0,0}+\triplenorm{d_t\varepsilon_n}^2_{0,0})\leq\Delta$, which are simpler than the $\mathbf{B}$–component case. Combining \eqref{eq:error_all1} and \eqref{eq:error_noE}, we deduce
	\begin{align}\label{eq:projerror_EBJ5}
		\sum_{k=0}^{j-1}\Delta t\|d_t\vartheta^{k+1}_\mathbf{B}\|^2+\frac{\mu_0\xi_R}{2}\|\vartheta^{j}_J\|^2+\frac{\mu_0\xi_R\xi_e}{2(n_1e)^2}\|\nabla \vartheta^{j}_J\|^2\leq CT\Delta .
	\end{align}
	
	Substituting this into \eqref{eq:error_E2}, we finally obtain the estimate for $E$: $$\triplenorm{\varepsilon_E}^2_{0,1}\leq CT\Delta.$$ This yields the first inequality \eqref{eq:error_M1} in Theorem~\ref{error_estimates}.
	Using \eqref{eq:error_noE}, we further obtain the second estimate \eqref{eq:error_M2} in Theorem~\ref{error_estimates}.
	
\end{proof}

\section{Numerical experiment}
\label{sec:experiments}
To validate the accuracy, thermodynamic consistency, and physical fidelity of the proposed thermodynamically consistent scheme, a series of two-dimensional numerical experiments are performed. All simulations are carried out using the iFEM package~\cite{ifem}. In Example~1, the order of convergence is assessed via the method of manufactured solutions; Example~2 investigates energy conservation, entropy production, and long-term stability in an isolated system; and Example~3 illustrates the ability of the two-fluid MHD model to capture magnetic reconnection phenomena.

\subsection{Example 1: Convergence Test}
We consider the computational domain  $\Omega = [0, 1] \times [0, 1]$ with the following exact solutions:
$$
\left\{
\begin{aligned}
	n(x, y, t) &= (e^{-t}+1)(\cos (\pi x) \cos (\pi y)+1.6), \\
	T_i(x, y, t) &= (e^{-t}+1)(\cos (\pi x) \cos (\pi y)+1.1), \\
	T_e(x, y, t) &= (e^{-t}+1)(\cos (\pi x) \cos (\pi y)+1.2), \\
	\mathbf{u}(x, y, t) &= (e^{-t}+1)
	\begin{bmatrix}
		(x^2 - x)\sin(\pi x) \\ (y^2 - y)\sin(\pi y)
	\end{bmatrix}, \\
	E(x, y, t) &= e^{-t}(x^2 - x)(y^2 - y), \\
	\mathbf{B}(x, y, t) &= e^{-t}
	\begin{bmatrix}
		(x^2 - x)(2y - 1) \\ (y^2 - y)(-2x + 1)
	\end{bmatrix}, \\
	J(x, y, t) &= -2e^{-t}(x^2 - x + y^2 - y).
\end{aligned}
\right.
$$
Parameters are set as follows: $r=\frac{5}{3}$, $\eta=\lambda=\mu_0=k_B=M_w=\xi_e=\xi_R=1$, and $\Theta_\alpha=\kappa_\alpha n$ with $\kappa_i=\kappa_e=1$.  
Homogeneous Neumann boundary conditions are applied to $n$ and $T_\alpha$, while Dirichlet conditions are imposed on the remaining variables.
For spatial discretization, uniform meshes are used, with:
\begin{itemize}
	\item Piecewise linear (P1) elements for $n,\,\mathbf{u},\,T_i,\,T_e,\,E,$ and $J$;
	\item RT0 elements for $\mathbf{B}$ to ensure the divergence-free constraint $\nabla\cdot\mathbf{B}=0$.
\end{itemize}

The $L^2$ errors and corresponding convergence rates for the plasma variables $(n, \mathbf{u}, T_i, T_e)$ and the electromagnetic variables $(\mathbf{E}, \mathbf{B}, \mathbf{J})$, along with $\nabla \cdot \mathbf{B}$, are summarized in Tables~\ref{tab:plasma_error} and~\ref{tab:maxwell_error}. The computations use a time step of $\Delta t = h^2/4$, and the errors are evaluated at $T = 1.0$. The numerical error for a variable $\phi$, denoted by $\varepsilon_\phi$, is computed as $\varepsilon_\phi = \phi - \phi_h$.
\begin{table}[htbp]
	\centering
	\caption{\small $L^2$ errors and convergence for plasma variables $(n,\mathbf{u},T_i,T_e)$}
	\label{tab:plasma_error}
	\begin{tabular}{@{}ccccccccc@{}}
		\toprule
		\ $ h $ & $ ||\varepsilon_n|| $ & order & $ ||\varepsilon_{\mathbf{u}}|| $& order & $ ||\varepsilon_{Ti}|| $ & order & $||\varepsilon_{Te}|| $ & order \\ \midrule
		1/4 & 7.74e-01 &   & 1.03e-01 &   & 5.02e-01 &  & 2.36e-01 &  \\
		1/8 & 2.31e-01 & 1.74 & 3.51e-02 & 1.55 & 1.83e-01 & 1.45 & 8.97e-02 & 1.40 \\
		1/16 & 6.10e-02 & 1.92 & 9.30e-03 & 1.92 & 5.26e-02 & 1.80 & 2.66e-02 & 1.75 \\
		1/32 & 1.57e-02 & 1.95 & 2.35e-03 & 1.98 & 1.36e-02 & 1.94 & 7.01e-03 & 1.92 \\ 
		\bottomrule
	\end{tabular}
\end{table}

\begin{table}[htbp]
	\centering
	\caption{\small $L^2$ errors and convergence for electromagnetic variables $(E,\mathbf{B},J,\nabla\cdot\mathbf{B})$}
	\label{tab:maxwell_error}
	\begin{tabular}{@{}cccccccc@{}}
		\toprule
		\ $h$ & $||\varepsilon_E||$ & order & $||\varepsilon_{\mathbf{B}}||$ & order & $||\varepsilon_J||$ & order & $||\nabla\cdot\mathbf{B}_h||$ \\ \midrule
		1/4  & 2.20e-03 &  & 2.17e-02 &  & 3.26e-02 &  & 1.39e-15 \\
		1/8  & 6.97e-04 & 1.66 & 1.12e-02 & 0.96 & 1.46e-02 & 1.15 & 6.47e-15 \\
		1/16 & 1.71e-04 & 2.03 & 5.62e-03 & 0.99 & 4.58e-03 & 1.68 & 3.10e-14 \\
		1/32 & 3.80e-05 & 2.17 & 2.80e-03 & 1.00 & 1.41e-03 & 1.69 & 1.88e-13 \\ 
		\bottomrule
	\end{tabular}
\end{table}

The results show convergence rates of approximately second order for the scalar variables and first order for $\mathbf{B}$, which closely matches the theoretical expectations for P1–RT0 discretizations. The observed mild reduction in the convergence rate of $J$ (approximately 1.69) is likely due to its strong nonlinear coupling with the density $n$, notably through the electron stress tensor term in the generalized Ohm’s law. Nevertheless, the overall accuracy is deemed satisfactory.

\subsection{Example 2: Isolated Systems}
A numerical simulation is performed for an ideal isolated plasma system with no mass or heat exchange and without gravitational effects, in order to verify its energy conservation, entropy production, and long-term stability.

The domain is taken to be $\Omega = [0,1] \times [0,1]$, using the same parameter set as in Example~1. Homogeneous Neumann boundary conditions are enforced for $n$ and $T_\alpha$ ($\alpha=i,e$), whereas homogeneous Dirichlet conditions are used for the remaining variables. The initial conditions are as follows:
$$
\left\{
\begin{aligned}
	n_0(x,y) &= 3(\cos(\pi x)\cos(\pi y)+1.5),\\
	T_{i0}(x,y) &= \tfrac{1}{2}(\cos(\pi x)\cos(\pi y)+3),\\
	T_{e0}(x,y) &= \tfrac{1}{4}(\cos(\pi x)\cos(\pi y)+2),\\
	\mathbf{B}_0(x,y) &= 
	\begin{bmatrix}
		(x^2-x)(2y-1)\\
		(y^2-y)(-2x+1)
	\end{bmatrix},\\
	\mathbf{u}_0(x,y) &= \mathbf{0}, \quad E_0(x,y)=0.
\end{aligned}
\right.
$$
We take $h=\frac{1}{32}$, $\Delta t = \frac{1}{4}h^2$, $T=10$.

Figure~\ref{fig:energy_variance} illustrates the time evolution of energy variances, defined as the differences from their initial values: total energy $\mathcal{E}^k-\mathcal{E}^0$, internal energy $\mathcal{U}^k-\mathcal{U}^0$, kinetic energy $\mathcal{H}^k-\mathcal{H}^0$, and electromagnetic energy $\mathcal{B}^k-\mathcal{B}^0$ (for $k=0,1,2,\dots$). In the absence of external sources (i.e., the right-hand side of equation~\eqref{eq:epsilonlaw_semi} is zero), the kinetic, internal, and electromagnetic energies undergo mutual conversion. The total energy, however, remains constant, thereby validating compliance with the first law of thermodynamics.
\begin{figure}[htbp]
	\centering
	\includegraphics[width=0.8\textwidth]{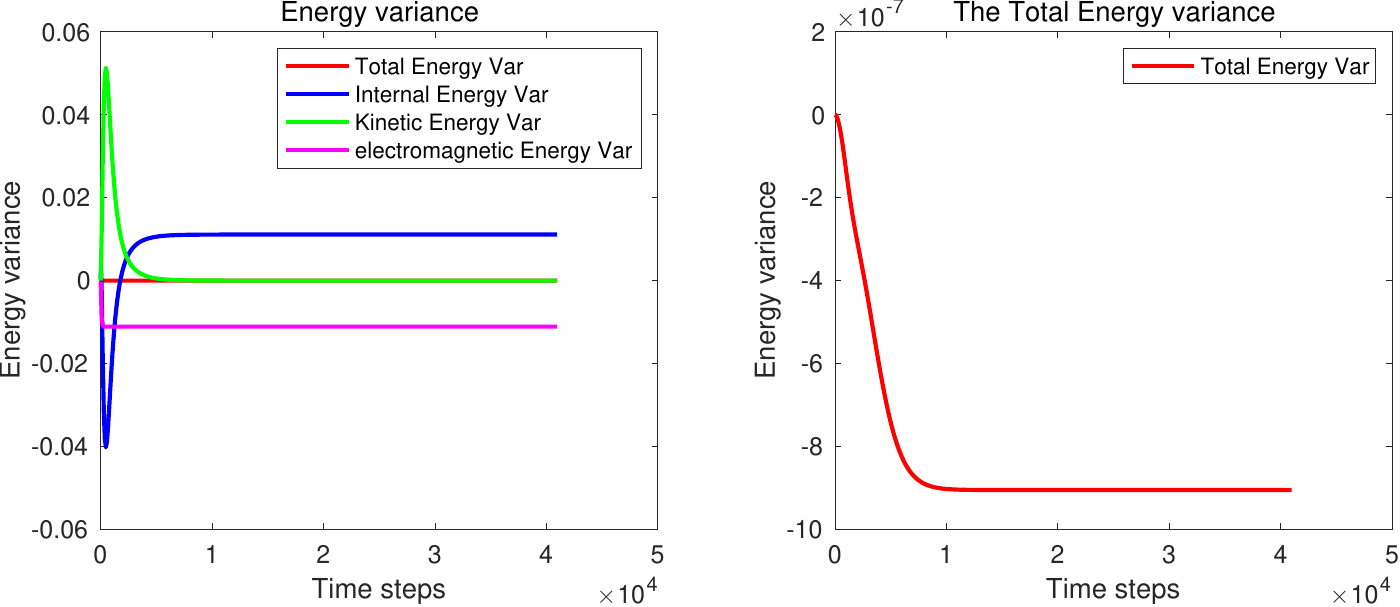}
	\caption{Energy variance over time.}
	\label{fig:energy_variance}
\end{figure}

Figures~\ref{fig:entropy_i}–\ref{fig:entropy_e} present the evolution of ion and electron entropies ($\mathcal{S}_i$ and $\mathcal{S}_e$). 
Both increase monotonically and reach steady states. 
The total entropy $\mathcal{S}=\mathcal{S}_i+\mathcal{S}_e$ also increases, verifying that the scheme satisfies the second law of thermodynamics.
\begin{figure}[htbp]
	\centering
	\begin{subfigure}{0.3\textwidth}
		\centering
		\includegraphics[width=\textwidth]{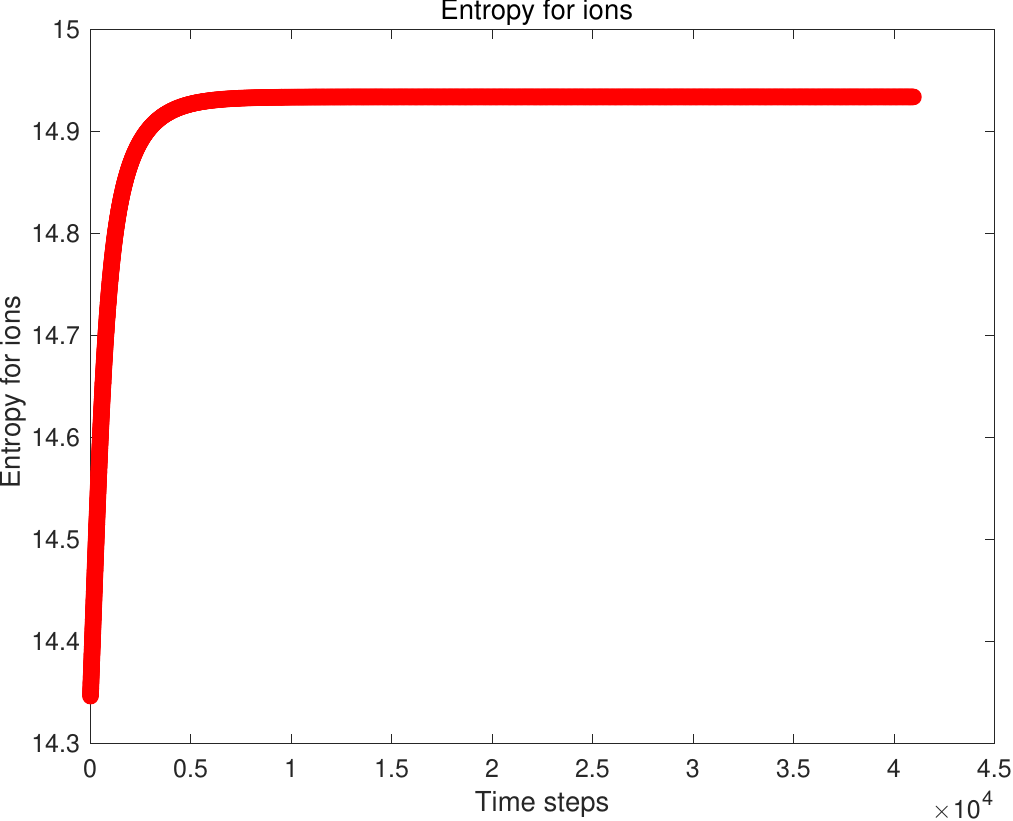}
		\caption{$\mathcal{S}_i$}
		\label{fig:entropy_i}
	\end{subfigure}
	\begin{subfigure}{0.3\textwidth}
		\centering
		\includegraphics[width=\textwidth]{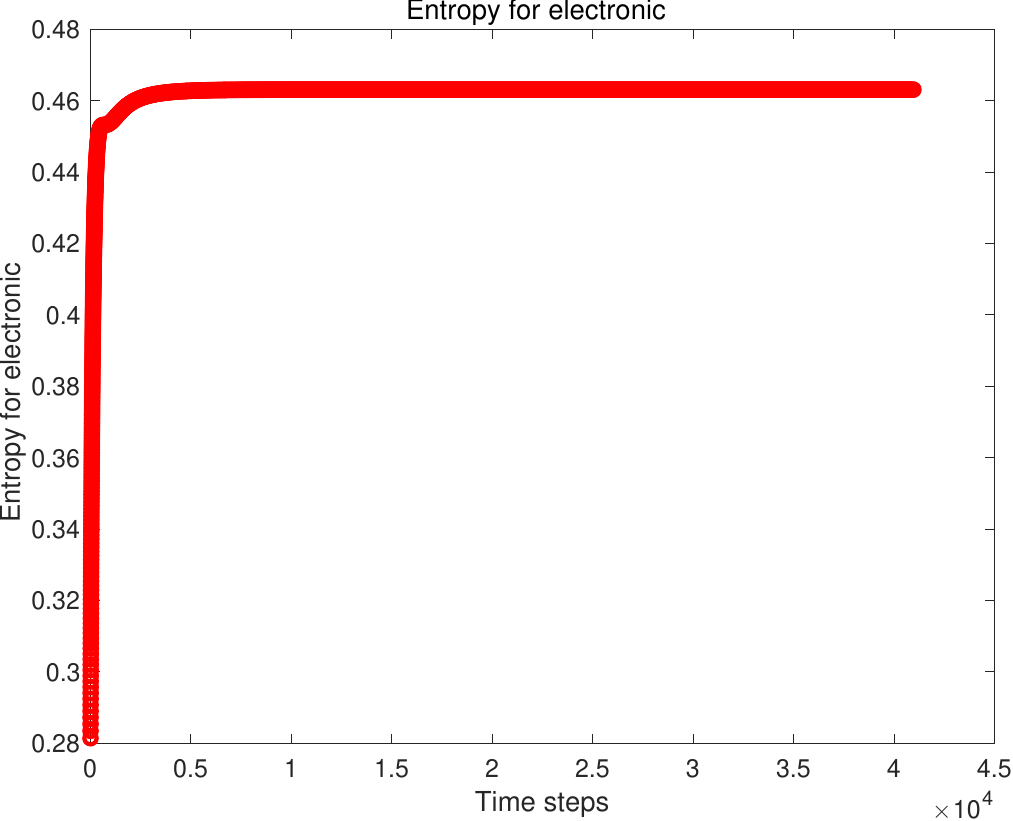}
		\caption{$\mathcal{S}_e$}
		\label{fig:entropy_e}
	\end{subfigure}
	\begin{subfigure}{0.3\textwidth}
		\centering
		\includegraphics[width=\textwidth]{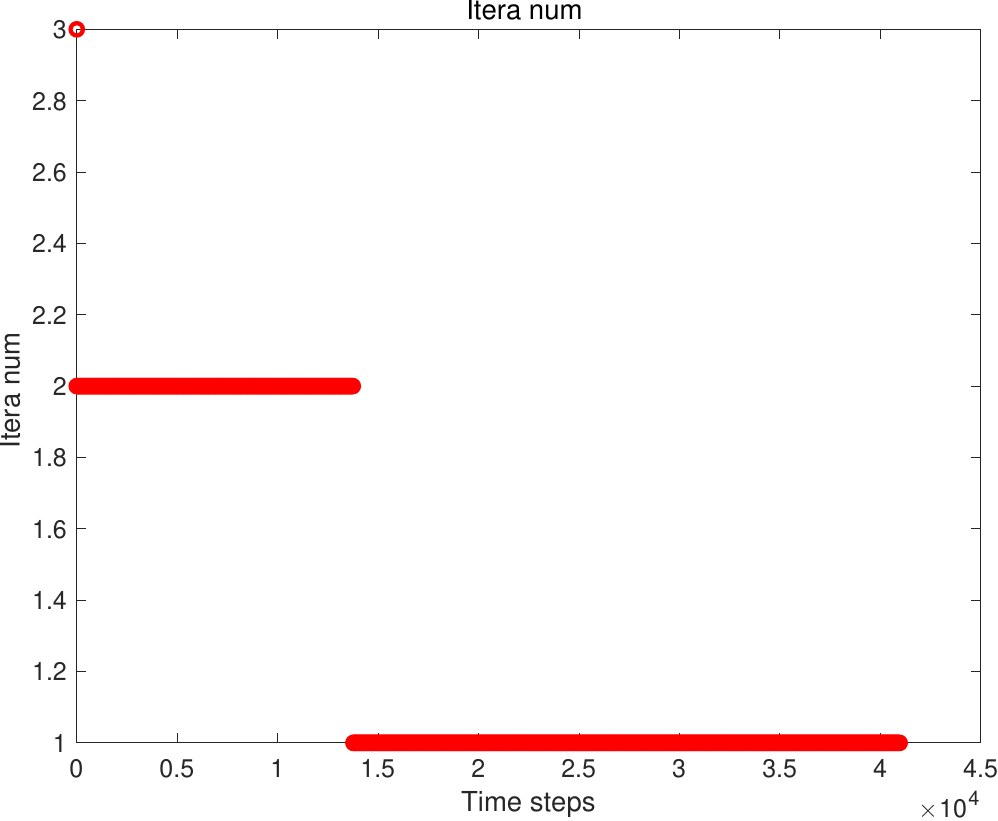}
		\caption{$iterations$}
		\label{fig:iter}
	\end{subfigure}
	\caption{\small Evolution of ion entropy (a), electron entropy (b), and iteration numbers (c).}
	\label{fig:entropy}
\end{figure}
The iterative solver stops when the relative 2-norm variation of all variables falls below $10^{-6}$. 
Figure~\ref{fig:iter} demonstrates rapid convergence behavior at each time step, which supports the computational efficiency of the method for long-term simulations.

\subsection{Example 3: Magnetic reconnection}
This example validates the capability of the proposed two-fluid MHD model to capture key physical mechanisms through numerical simulation of the classical magnetic reconnection problem. 
Magnetic reconnection, a fundamental nonlinear process in plasma physics, widely occurs in astrophysical (e.g., solar flares, geomagnetic storms) and laboratory plasmas. 
It is characterized by magnetic topology rearrangement, rapid magnetic energy release, and current sheet formation and dissipation.

The simulation adopts a 2D configuration following standard setups in magnetic reconnection studies~\cite{ZLi2019,Jardin2007,Cassak2007}. 
The initial magnetic field is an anti-parallel (Harris sheet-type) configuration with equilibrium background. 
A small perturbation is introduced to trigger instability and initiate reconnection. 
The initial equilibrium conditions are
$$
\left\{
\begin{array}{l}
	n_0(x, y)=\operatorname{sech}^2(2y)+0.2, \\
	T_{i0}(x, y)=0.5,\quad T_{e0}(x, y)=0.5, \\
	\mathbf{B}_0(x, y)=
	\begin{bmatrix}
		\tanh(2y) \\[2pt] 0
	\end{bmatrix}, \\
	\mathbf{u}_0(x, y)=\mathbf{0},\quad E_0(x, y)=0.
\end{array}
\right.
$$
At $t=0$, a small perturbation is imposed on $\mathbf{B}$:
$$\boldsymbol{\delta} = \varepsilon
\begin{bmatrix}
	-k_y \cos(k_x x)\sin(k_y y) \\[2pt]
	k_x \sin(k_x x)\cos(k_y y)
\end{bmatrix}.$$
The computational domain is the rectangular region $[-L_x/2, L_x/2] \times [-L_y/2, L_y/2]$, with periodic boundary conditions in the $x$-direction and ideal conducting walls at $y = \pm L_y/2$. The parameters are specified as follows: $k_x = 2\pi/L_x$, $k_y = \pi/L_y$, $L_x = 12.8$, $L_y = 6.4$, and a perturbation amplitude of $\varepsilon = 0.1$. The physical dissipation coefficients are set to $\xi_R = 0.004$ for resistivity, $\xi = \eta = 0.0025$ for viscosity, and $\kappa_i = \kappa_e = 0.02$ for thermal conductivity. Hyper-resistivity is included via $\xi_e = 0.1h^2$, where $h$ denotes the characteristic mesh size. The spatial discretization employs a uniform triangular mesh with $65 \times 33$ nodes, yielding $h = L_y/32 = 0.2$, and the time step is $\Delta t = h^2/2 = 0.02$.

\begin{figure}[htbp]
	\centering
	\includegraphics[width=0.38\textwidth]{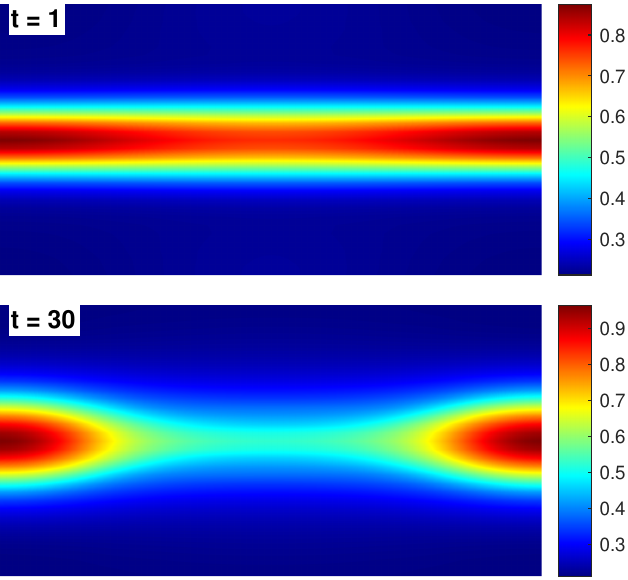}
	\includegraphics[width=0.38\textwidth]{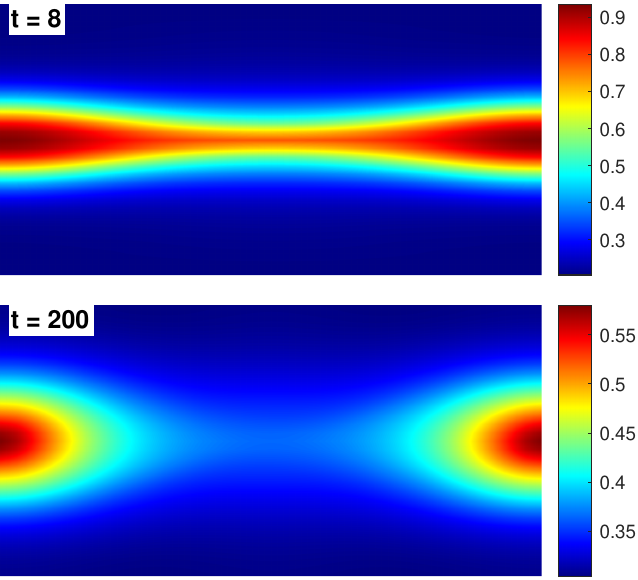}
	\caption{Plasma number density $n$ at selected times.} 
	\label{fig:n_trisurf}
\end{figure}

Figure~\ref{fig:n_trisurf} illustrates the evolution of the plasma number density from a near-uniform, central concentration. Following the onset of magnetic reconnection and associated energy release, the plasma is displaced laterally, leading to a significant decrease in central density by $t = 30 $ and the emergence of two high-density jets at the edges. By 
$t = 200$, a clear rarefied band forms centrally, with dense structures remaining on both sides, thereby capturing the magnetic-reconnection-driven acceleration and ejection of plasma from the central region.

\begin{figure}[htbp]
	\centering
	\includegraphics[width=0.32\textwidth]{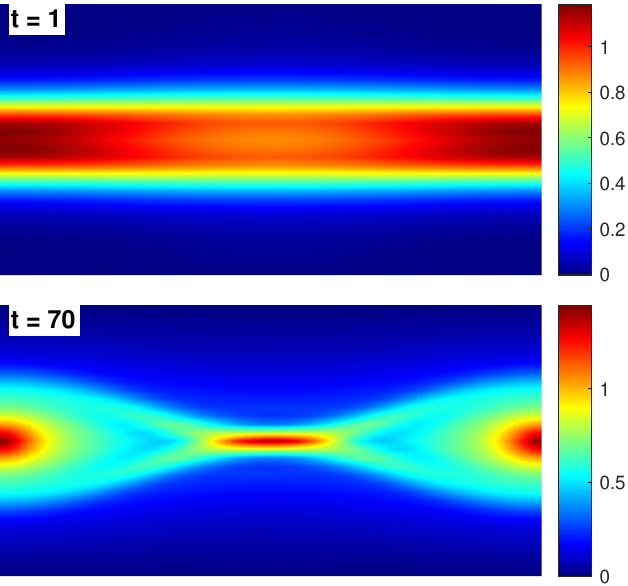}
	\includegraphics[width=0.32\textwidth]{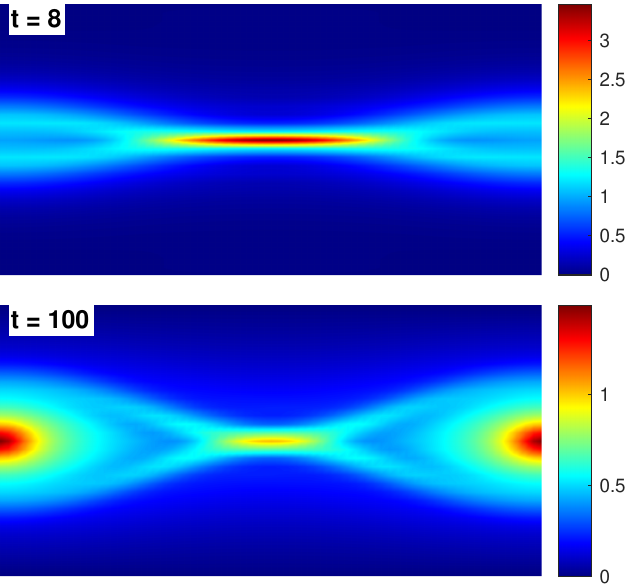}
	\includegraphics[width=0.32\textwidth]{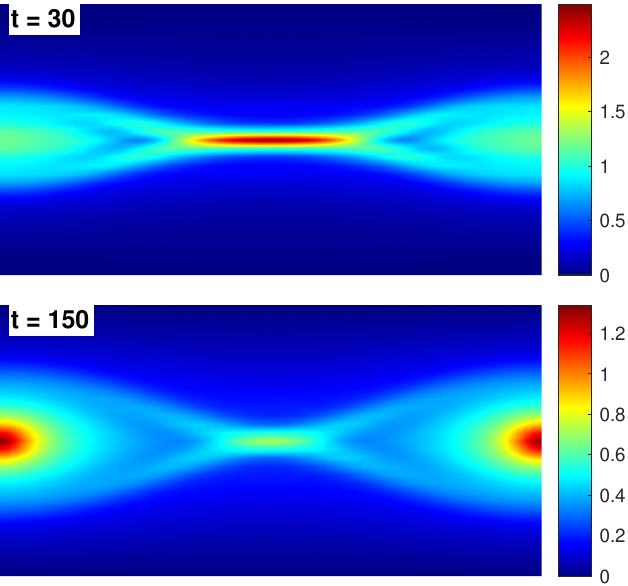}
	\caption{Current density $J$ at selected times.} 
	\label{fig:J_trisurf}
\end{figure}

The evolution of the current sheet, depicted in Figure \ref{fig:J_trisurf} via $-J$ contours, captures the key stages of magnetic reconnection. The initial central concentration ($t=1$) undergoes thinning and intensification, followed by dissipation and diffusion ($t=70 \sim 150$) that annihilates the central sheet and broadens the current distribution. This clear transition from a singular sheet to bifurcated structures directly demonstrates the breaking and reconnection of magnetic field lines.
\begin{figure}[htbp]
	\centering
	\includegraphics[width=0.98\textwidth]{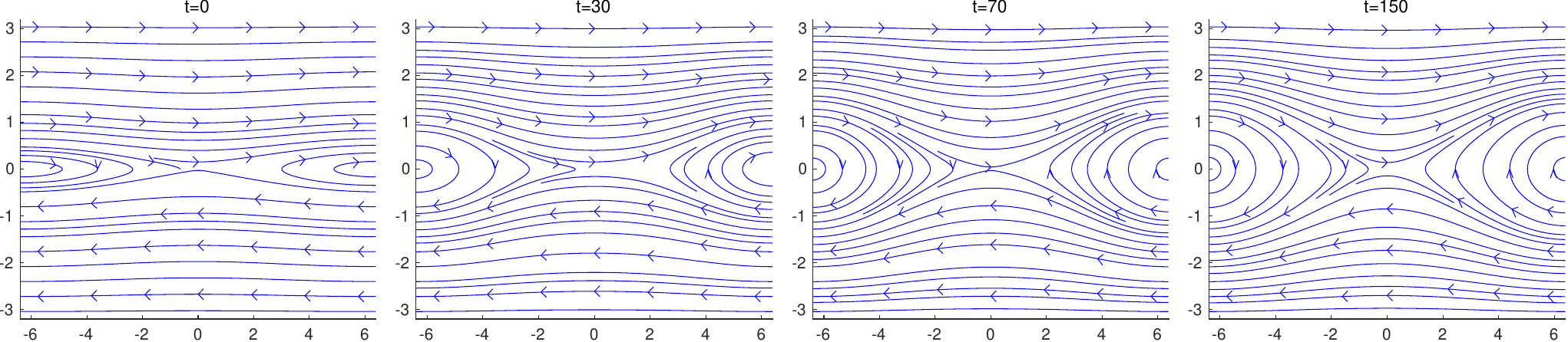}
	\caption{\small Evolution of magnetic field streamlines.}
	\label{fig:B_stream}
\end{figure}

Figure~\ref{fig:B_stream} depicts the temporal evolution of magnetic field streamlines. As the system evolves, the initial disturbance amplifies, leading to a pronounced vertical bending of the field lines by $t = 30$. Subsequently, a well-defined X-point emerges by $t = 70$, signalling the breaking and reconnection of magnetic flux at this location. Ultimately, by 
$t =150$, the topology transitions into two symmetric closed magnetic islands (O-points), marking the evolution from open to closed field line configurations accompanied by a corresponding redistribution of magnetic energy. This sequence of X–O topological evolution is characteristic of classical magnetic reconnection.

\begin{figure}[htbp]
	\centering
	\includegraphics[width=0.98\textwidth]{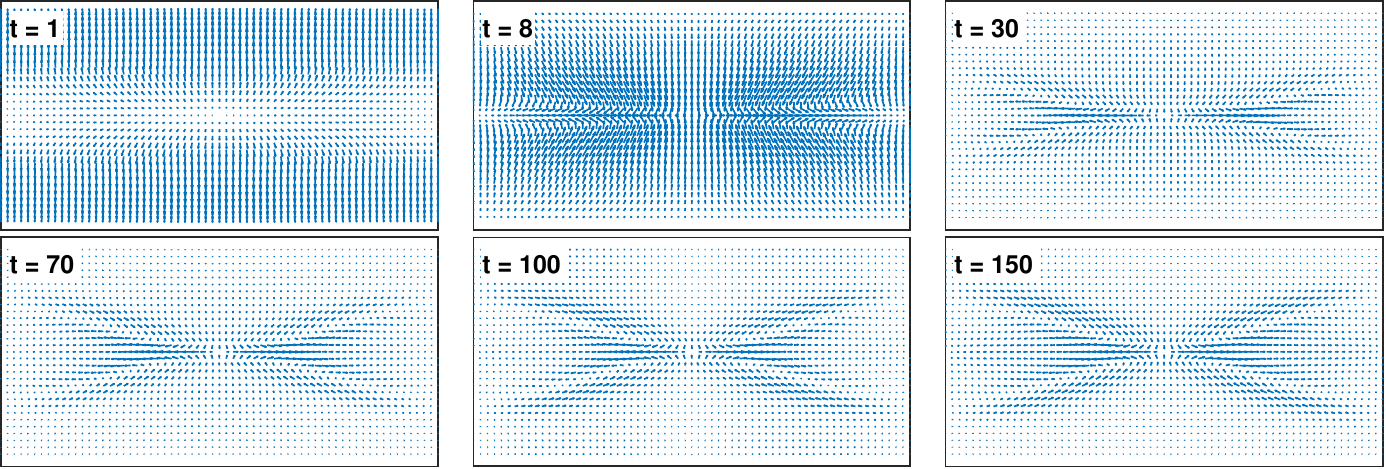}
	\caption{Ion velocity over time.}
	\label{fig:u_quiver}
\end{figure}

The evolution of the ion velocity field in Figure~\ref{fig:u_quiver} captures the outflow acceleration intrinsic to magnetic reconnection. The emergence of bidirectional jets during the peak reconnection phase 
($t = 30 \sim 100$) is a hallmark signature, demonstrating direct energy transfer from the magnetic field to the plasma. The subsequent decay of these flows coincides with the saturation of the reconnection process, completing the energy conversion cycle.

\begin{figure}[htbp]
	\centering
	\includegraphics[width=0.32\textwidth]{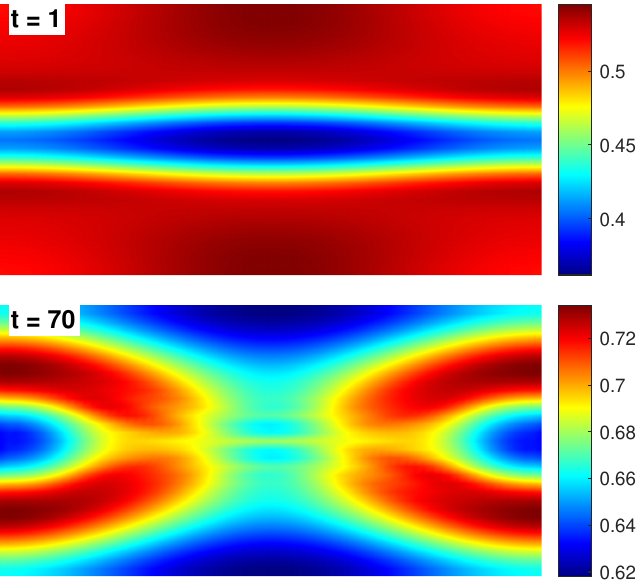}
	\includegraphics[width=0.32\textwidth]{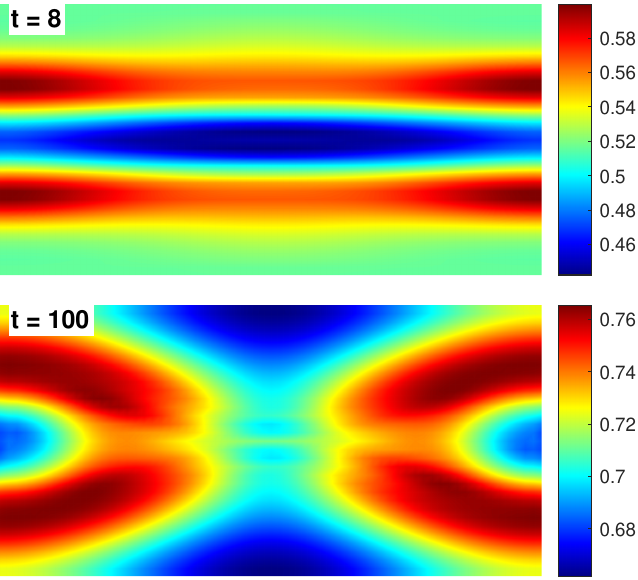}
	\includegraphics[width=0.32\textwidth]{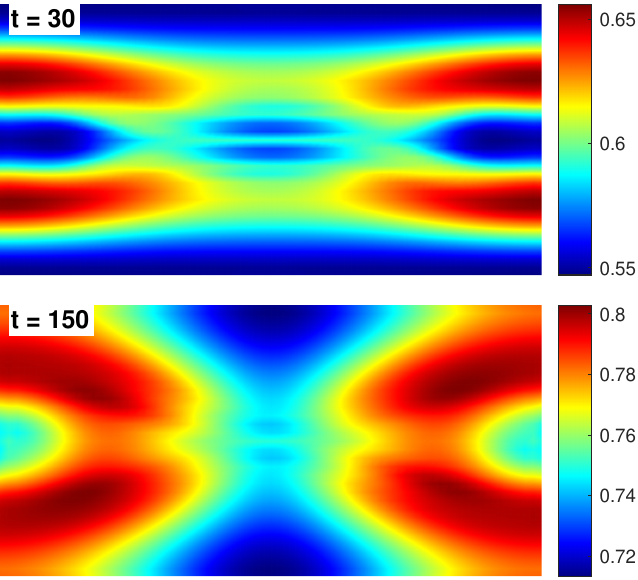}
	\caption{Temperature of ions $T_i$ at selected times.} 
	\label{fig:Ti_trisurf}
\end{figure}

\begin{figure}[htbp]
	\centering
	\includegraphics[width=0.32\textwidth]{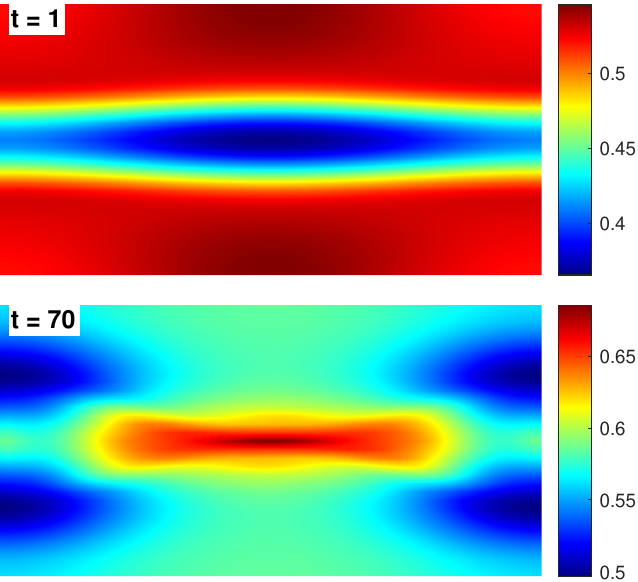}
	\includegraphics[width=0.32\textwidth]{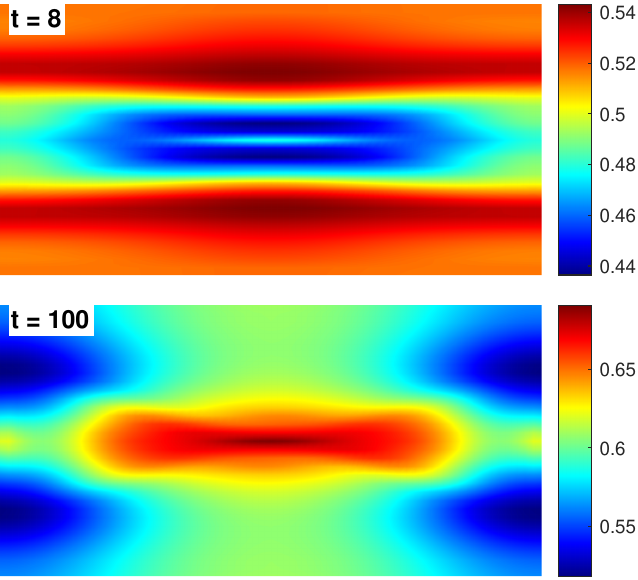}
	\includegraphics[width=0.32\textwidth]{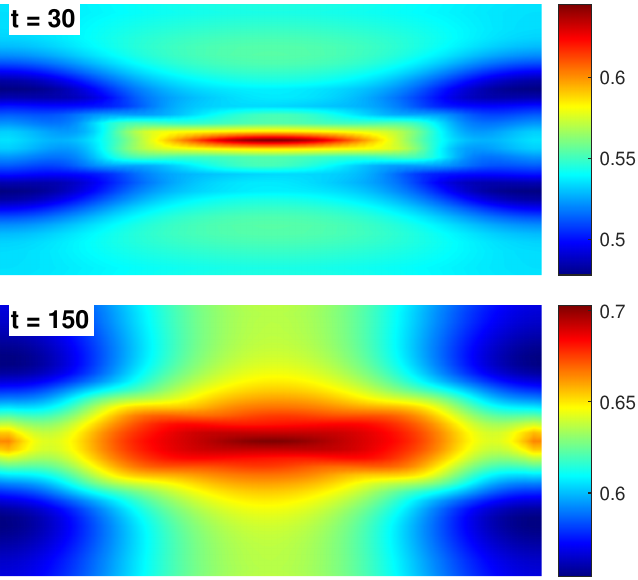}
	\caption{Temperature of electrons $T_e$ at selected times.} 
	\label{fig:Te_trisurf}
\end{figure}
Temperature evolution further supports the magnetic energy release and conversion mechanisms. 
Figures~\ref{fig:Ti_trisurf}--\ref{fig:Te_trisurf} show the ion and electron temperatures, respectively. 
Following reconnection onset ($t>30$), the electron temperature rises rapidly near the reconnection site, forming localized high-temperature regions where magnetic energy is converted into heat. 
In the later stage, ion temperature continues to rise, with a symmetrically distributed high-temperature region appearing near the X-shaped structure where the velocity field converges—demonstrating localized heating from fluid compression.
Throughout the reconnection process, electron temperature responds rapidly and locally, while ion temperature responds more slowly and broadly—exhibiting a distinct two-fluid characteristic.

The conservation of total energy and the production of entropy are also assessed. As shown in Figures~\ref{fig:real_energy}--\ref{fig:real_entropy}, the total energy remains constant, whereas the entropies of both ions and electrons exhibit a monotonic increase and converge rapidly.

\begin{figure}[htbp]
	\centering
	\includegraphics[width=0.4\textwidth]{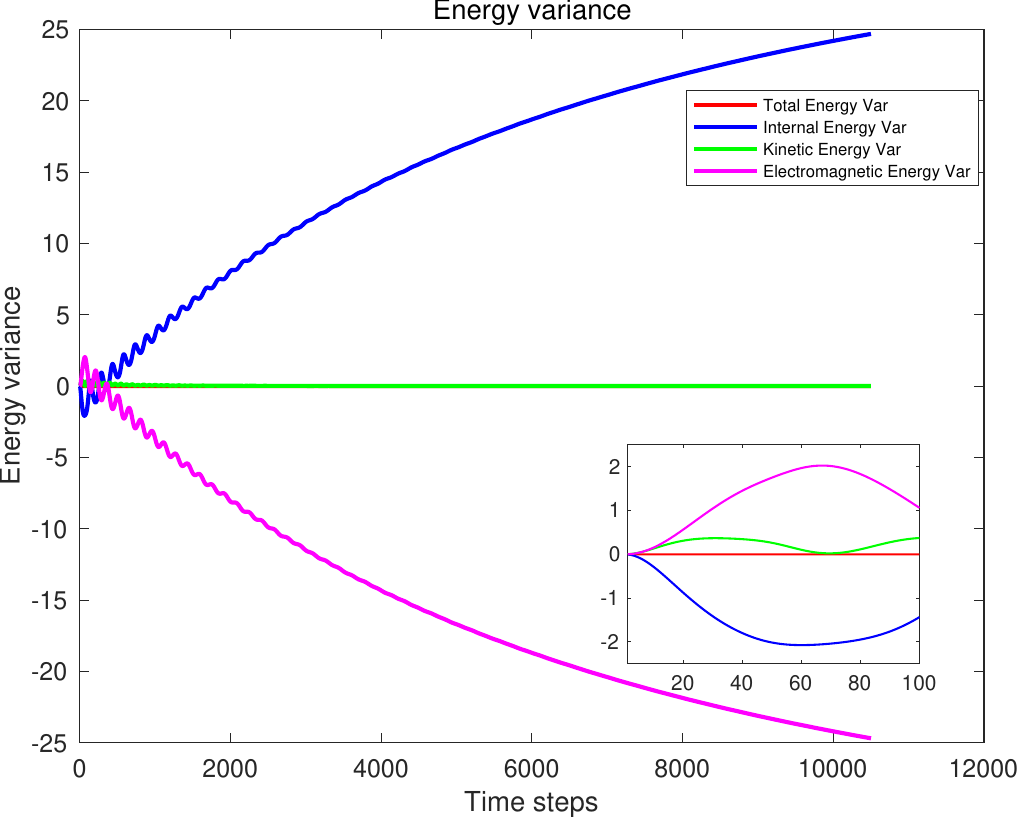}
	\hspace{5mm}
	\includegraphics[width=0.4\textwidth]{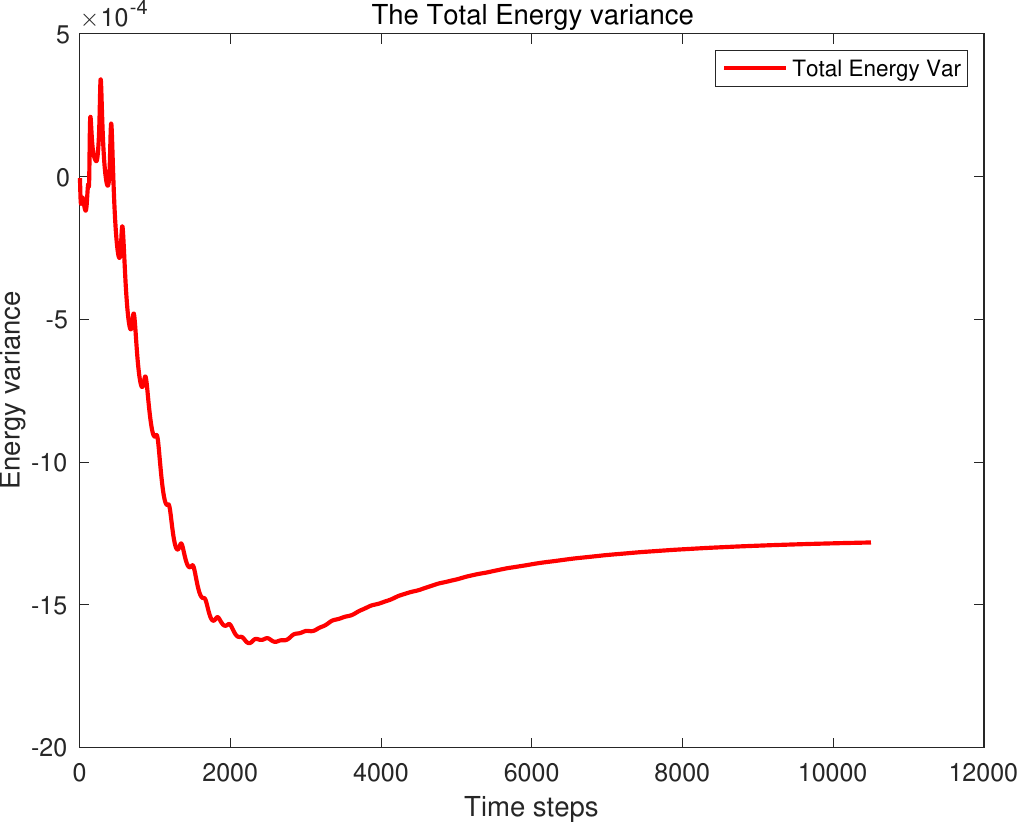}
	\caption{Energy variance over time.} 
	\label{fig:real_energy}
\end{figure}

\begin{figure}[htbp]
	\centering
	\begin{subfigure}{0.3\textwidth}
		\centering
		\includegraphics[width=\textwidth]{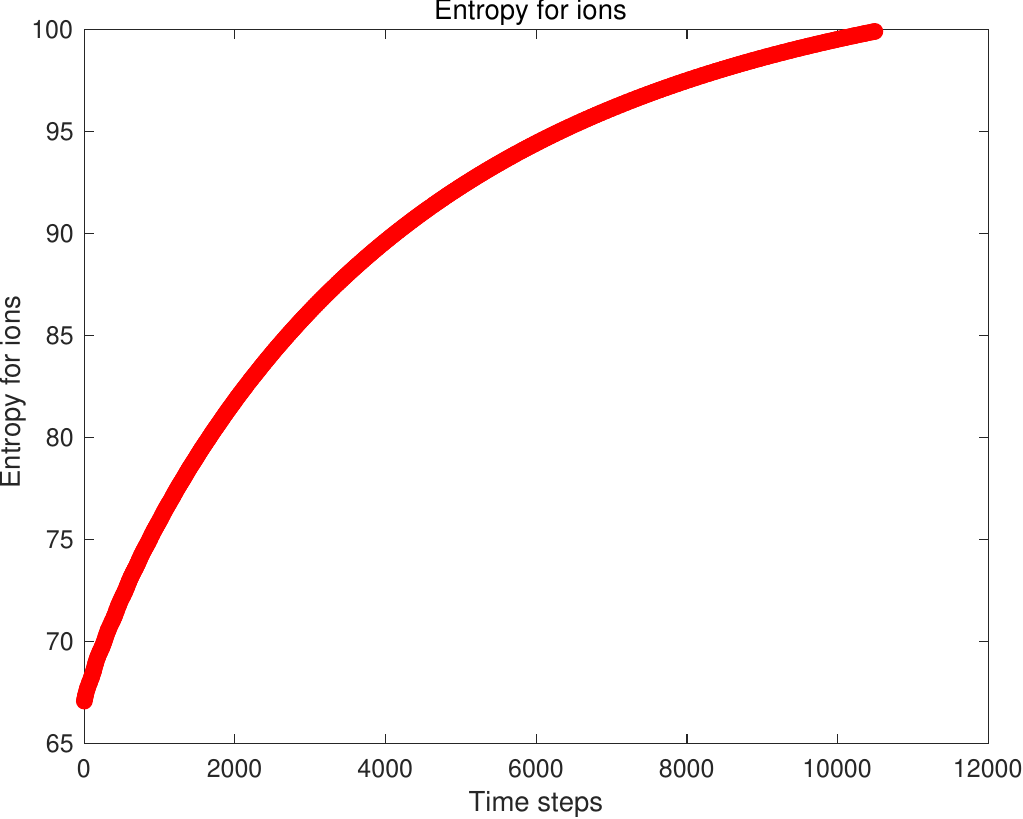}
		\caption{$\mathcal{S}_i$}
		\label{fig:real_entropy_i}
	\end{subfigure}
	\begin{subfigure}{0.3\textwidth}
		\centering
		\includegraphics[width=\textwidth]{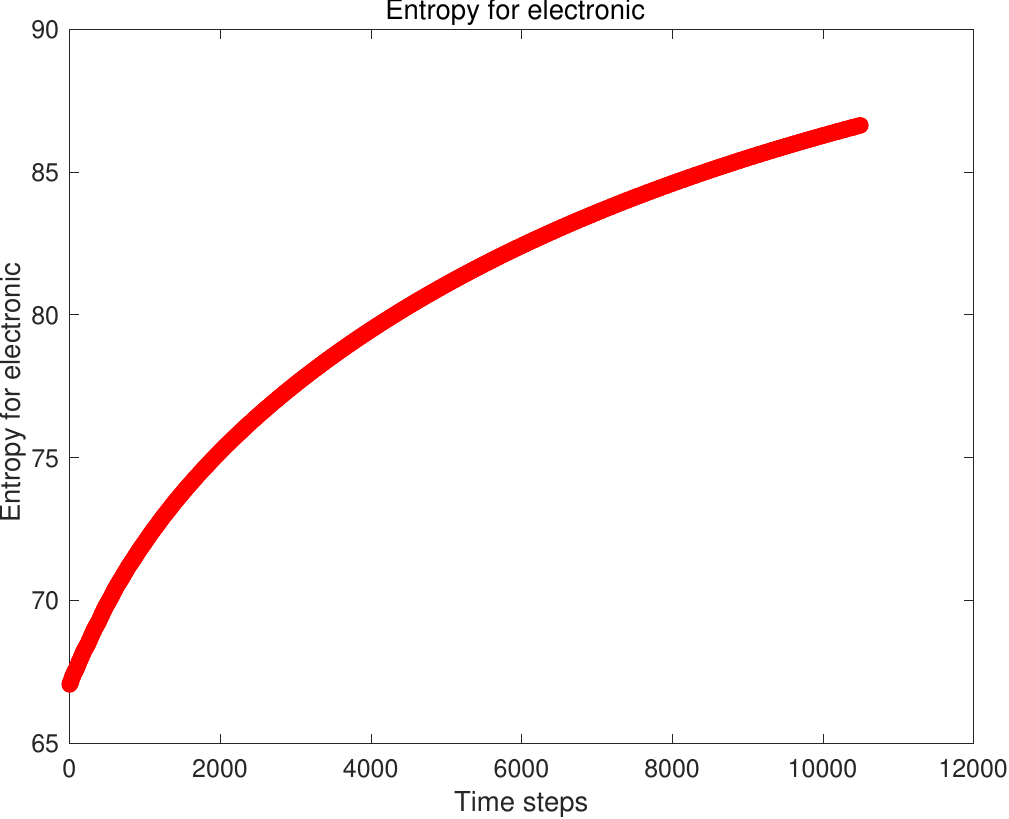}
		\caption{$\mathcal{S}_e$}
		\label{fig:real_entropy_e}
	\end{subfigure}
	\begin{subfigure}{0.3\textwidth}
		\centering
		\includegraphics[width=\textwidth]{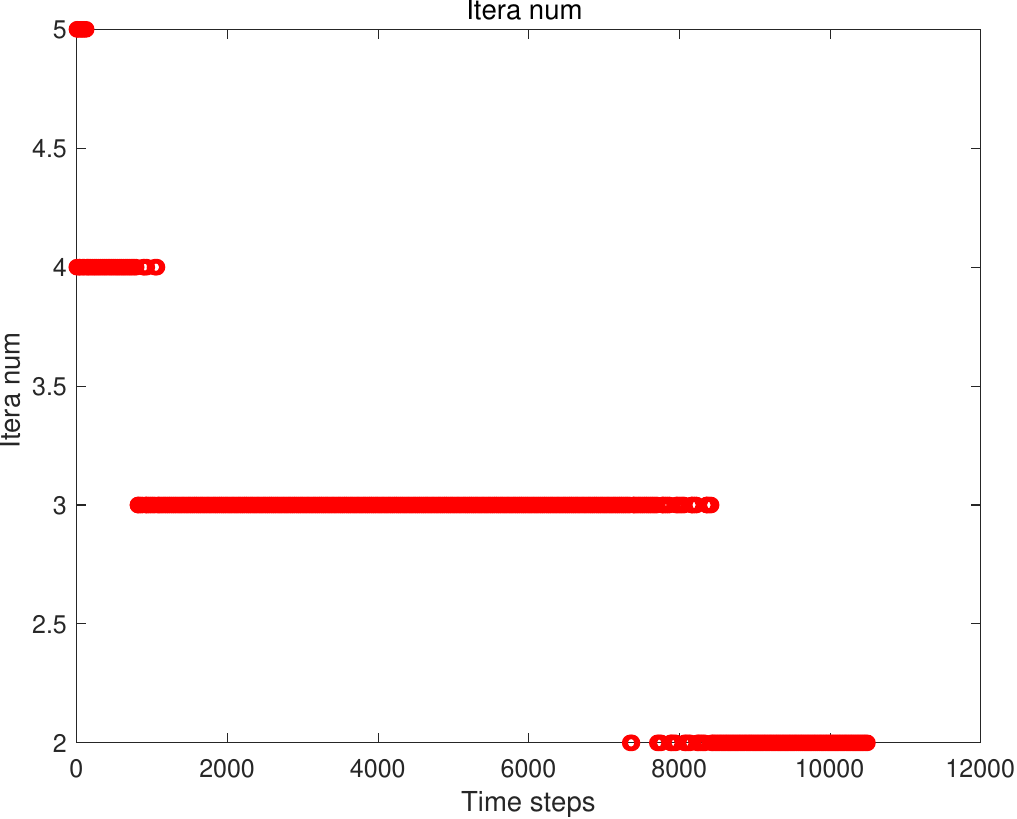}
		\caption{$iterations$}
		\label{fig:real_iter}
	\end{subfigure}
	\caption{\small Evolution of ion entropy (a), electron entropy (b), and iteration numbers (c).}
	\label{fig:real_entropy}
\end{figure}

\section{Conclusion}
\label{sec:conclusions}

This study establishes a thermodynamically consistent two-fluid magnetohydrodynamic model (2F-MHD) based on the Helmholtz free energy framework. By introducing the convex–concave design of the free energy density $f_\alpha(n,T_\alpha)$, which is convex with respect to the density $n$ and concave with respect to the temperature $T_\alpha$, the proposed model self-consistently derives fundamental thermodynamic quantities such as the chemical potential, entropy density, and internal energy. The resulting formulation rigorously satisfies the first and second laws of thermodynamics at the continuous and time semi-discrete levels.

Building on this theoretical foundation, we develop a finite element discretization for the degenerate two-dimensional system and provide a priori error estimates in both space and time. Numerical experiments demonstrate that the model captures essential plasma dynamics, including current sheet thinning, energy concentration, and magnetic reconnection, in agreement with physically observed phenomena in magnetic confinement and astrophysical plasmas.

Furthermore, this study adapts the thermodynamically consistent modeling para-digm from phase-field and continuum systems to the two-fluid MHD framework, providing a systematic route for developing physically consistent numerical methods for plasma dynamics. Future research will focus on extending this approach to 2.5D and full 3D configurations to further assess thermodynamic consistency and investigate fine-scale couplings between electromagnetic and thermal processes.

\bibliographystyle{siamplain}
\bibliography{MHD2F_ref}

\end{document}